\documentclass[10pt,a4paper]{article}
\usepackage[utf8]{inputenc}
\usepackage{amsmath}
\usepackage{amsfonts}
\usepackage{amssymb}
\usepackage{graphicx}
\usepackage[left=2.00cm, right=2.00cm, top=1.00cm, bottom=1.00cm]{geometry}

\usepackage[english]{babel}
\usepackage{amsfonts, dsfont, nicefrac,amsthm}
\usepackage{stmaryrd}
\usepackage[caption=false]{subfig} 
\usepackage{ulem}
\usepackage{float}
\usepackage{xcolor}

\newcommand\dualfunction{\phi}
\newcommand\dualfunctionperlayer{\hat{\phi}}
\newcommand\eigenfunction{\hat{\rho}}
\newcommand\densityfunction{\rho}
\newcommand\Malthus{\lambda_{c}}
\newcommand\MalthusPerLayer{\lambda}
\newcommand\MalthusIndex{c}

\newcounter{thm}

\newtheorem{remarks}[thm]{Remark}
\newtheorem{hypothesis}[thm]{Hypothesis}
\newtheorem{hypdef}[thm]{Hypothesis/Definition}
\newtheorem{definition}[thm]{Definition}
\newtheorem{corollary}[thm]{Corollary}
\newtheorem{theorem}[thm]{Theorem}
\newtheorem{lemma}[thm]{Lemma}

\title{Analysis and calibration of a linear model for structured cell populations with unidirectional motion : application to the morphogenesis of ovarian follicles} 
\author{
	Fr\'{e}d\'{e}rique Cl\'{e}ment\footnote{Project team MYCENAE, Centre INRIA de Paris, France. (frederique.clement@inria.fr)},
	Fr\'{e}d\'{e}rique Robin\footnote{Project team MYCENAE, Centre INRIA de Paris, France. (frederique.robin@inria.fr)}, and 
	Romain Yvinec \footnote{PRC, INRA, CNRS, IFCE, Université de Tours, 37380 Nouzilly, France.    (romain.yvinec@inra.fr)}.
}
\begin{document}
	\maketitle
	
	\begin{abstract}
		  We analyze a multi-type age dependent model for cell populations subject to unidirectional motion, in both a stochastic and deterministic framework. Cells are distributed into successive layers; they may divide and move irreversibly from one layer to the next. We adapt results on the large-time convergence of PDE systems and branching processes to our context, where the Perron-Frobenius or Krein-Rutman theorem can not be applied. 
		  We derive explicit analytical formulas for the asymptotic cell number moments, and the stable age distribution. We illustrate these results numerically and we apply them to the study of the morphodynamics of ovarian follicles. We prove the structural parameter identifiability of our model in the case of age independent division rates. Using a set of experimental biological data, we estimate the model parameters to fit the changes in the cell numbers in each layer during the early stages of follicle development. 
	\end{abstract}

\section{Introduction}
	We study a multi-type age dependent model in both a deterministic and stochastic framework to represent the dynamics of a population of cells distributed into successive layers. The model is a two dimensional structured model: cells are described by a continuous age variable and a discrete layer index variable.  Cells may divide and move irreversibly from one layer to the next. The cell division rate is age and layer dependent, and is assumed to be bounded below and above. After division, the age is reset and the daughter cells either remain within the same layer or move to the next one. In its stochastic formulation, our model is a multi-type Bellman-Harris branching process and in its deterministic formulation, it is a multi-type McKendrick-VonFoerster system.
	
	The model enters the general class of linear models leading to Malthusian exponential growth of the population. In the PDE case, state-of-the-art-methods call to renewal equations system \cite{metz_dynamics_1986} or, to an eigenvalue problem and general relative entropy techniques \cite{michel_general_2005,perthame_transport_2007} to show the existence of an attractive stable age distribution. Yet, in our case, the unidirectional motion prevents us from applying the Krein-Rutman theorem to solve the eigenvalue problem. As a consequence, we follow a constructive approach and explicitly solve the eigenvalue problem. On the other hand, we adapt entropy methods using weak convergences in $\mathbf{L}^1$ to obtain the large-time behavior and lower bound estimates of the speed of convergence towards the stable age distribution. 
	In the probabilistic case, classical methods rely on renewal equations \cite{harris_theory_1963} and martingale convergences \cite{jagers_population-size-dependent_2000}. Using the same eigenvalue problem as in the deterministic study, we derive a martingale convergence giving insight into the large-time fluctuations around the stable state. Again, due to the lack of reversibility in our model, we cannot apply the Perron-Frobenius theorem to study the asymptotic of the renewal equations. Nevertheless, we manage to derive explicitly the stationary solution of the renewal equations for the cell number moments in each layer as in \cite{harris_theory_1963}. We recover the deterministic stable age distribution as the solution of the renewal equation for the mean age distribution.
	
	The theoretical analysis of our model highlights the role of one particular layer: the leading layer characterized by a maximal intrinsic growth rate which turns out to be the Malthus parameter of the total population. The notion of a leading layer is a tool to understand qualitatively the asymptotic cell dynamics, which appears to operate in a multi-scale regime. All the layers upstream the leading one may extinct or grow with a rate strictly inferior to the Malthus parameter, while the remaining, downstream ones are driven by the leading layer. 
	
	We then check and illustrate numerically our theoretical results. In the stochastic case, we use a standard implementation of an exact Stochastic Simulation Algorithm. In the deterministic case, we design and implement a dedicated finite volume scheme adapted to the non-conservative form and dealing with proper boundary conditions. We verify that both the deterministic and stochastic simulated distributions agree with the analytical stable age distribution. Moreover, the availability of analytical formulas helps us to study the influence of the parameters on the asymptotic proportion of cells, Malthus parameter and stable age distribution.
	 
	Finally, we consider the specific application of ovarian follicle development inspired by the model introduced in \cite{clement_coupled_2013} and representing the proliferation of somatic cells and their organization in concentric layers around the germ cell. While the original model is formulated with a nonlinear individual-based stochastic formalism, we design a linear version based on branching processes and endowed with a straightforward deterministic counterpart. We prove the structural parameter identifiability in the case of age independent division rates. Using a set of experimental biological data, we estimate the model parameters to fit the changes in the cell numbers in each layer during the early stages of follicle development. The main interest of our approach is to benefit from the explicit formulas derived in this paper to get insight on the regime followed by the observed cell population growth.
	
	Beyond the ovarian follicle development, linear models for structured cell populations with unidirectional motion may have several applications in life science modeling, as many processes of cellular differentiation and/or developmental biology are associated with a spatially oriented development (e.g. neurogenesis on the cortex, intestinal crypt) or commitment to a cell lineage or fate (e.g. hematopoiesis, acquisition of resistance in bacterial strains).

	The paper is organized as follows. In section $2$, we describe the stochastic and deterministic model formulations and enunciate the main results. In section $3$, we give the main proofs accompanied by numerical illustrations. Section $4$ is dedicated to the application to the development of ovarian follicles. We conclude in section $5$. Technical details and classical results are provided in Supplementary materials.

\section{Model description and main results}
	\subsection{Model description}
		We consider a population of cells structured by age $a \in \mathbb{R}_+$ and distributed into layers indexed from $j=1$ to $j=J \in \mathbb{N}^*$. The cells undergo mitosis after a layer-dependent stochastic random time $\tau = \tau^{j}$, ruled by an age-and-layer-dependent instantaneous division rate $b = b_j(a)$ : $\mathbb{P}[\tau^j > t] = e^{- \int_{0}^{t}b_j(a)da}$. Each cell division time is independent from the other ones. At division, the age is reset and the two daughter cells may pass to the next layer according to layer-dependent probabilities. We note $p^{(j)}_{2,0}$ the probability that both daughter cells remain on the same layer, $p^{(j)}_{1,1}$ and $p^{(j)}_{0,2}$, the probability that a single or both daughter cell(s) move(s) from layer $j$ to layer $j+1$, with $p_{2,0}^{(j)} + p_{1,1}^{(j)} + p_{0,2}^{(j)} = 1  $. Note that the last layer is absorbing: $p_{2,0}^{(J)} = 1 $. The dynamics of the model is summarized in Figure \ref{fig:ModelDescription}.
		\begin{figure}[h]
			\centering
			\includegraphics[height=2.75cm]{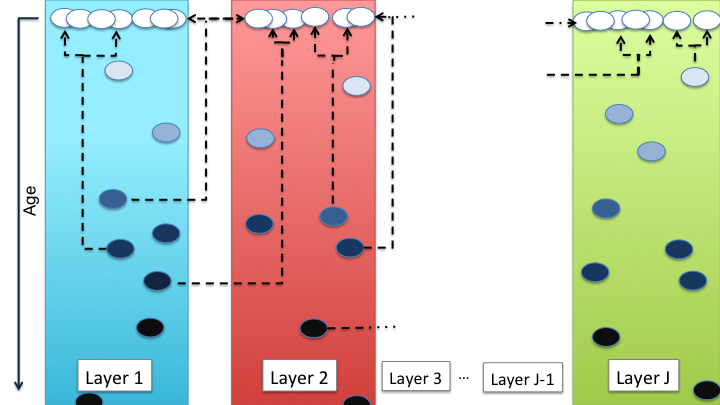}
			\caption{\textbf{Model description.} Each cell ages until an age-dependent random division time $\tau^j$. At division time, the age is reset and the two daughter cells may move only in an unidirectional way. When $j = J$, the daughter cells stay on the last layer.}
			\label{fig:ModelDescription}
		\end{figure}
		\paragraph{Stochastic model}
		Each cell in layer $j$ of age $a$ is represented by a Dirac mass $\delta_{j, a} $ where $(j, a) \in \mathcal{E} = \llbracket 1, J \rrbracket \times \mathbb{R}^+ $. Let $\mathcal{M}_P$ be the set of point measures on $\mathcal{E}$:
		 $$ \displaystyle \mathcal{M}_P := \left\{  \sum_{k=1}^{N} \delta_{j_k, a_k}, N \in \mathbb{N}^*, \, \forall k \in \llbracket 1, N \rrbracket, (j_k,a_k)\in \mathcal{E}  \right\} \, .$$
		 The cell population is represented for each time  $t \geq 0$ by a measure $Z_t \in \mathcal{M}_P$:
			\begin{equation}\label{Z_t_point_measure}
			Z_t = \sum_{k = 1}^{N_t}\delta_{I_t^{(k)}, \, A_t^{(k)}}, \quad  N_t := \, \ll Z_t, \mathds{1} \gg \, = \,  \displaystyle\sum_{j = 1}^{J}\int_{0}^{+ \infty} Z_t(dj,da) \, .
			\end{equation}
			$N_t$ is the total number of cells at time $t$. On the probability space $(\Omega, \mathcal{F}, \mathbb{P}) $, we define $Q$ as a Poisson point measure of intensity $ds \otimes \# dk \otimes d \theta$, where $ds$ and $ d \theta$ are Lebesgue measures on $\mathbb{R}_+$ and $\# dk$ is a counting measure on $\llbracket 1, J \rrbracket$.  The dynamics of $Z = (Z_t)_{t \geq 0}$ is given by the following stochastic differential equation:
			\begin{equation} \label{Z_t_Equation}
				\begin{array}{l}
				Z_{t} = \displaystyle\sum_{k = 1}^{N_0}\delta_{I^{(k)}_0, \, A_0^{(k)} + t}+ \int_{[0,t] \times \mathcal{E}} \mathds{1}_{k \leq N_{s^-}}R(k,s,Z,\theta) Q(\text{d}s,\text{d}k,\text{d}\theta) \\[0.3cm] 
				\text{	where 	}	R(k,s,Z,\theta) =  (2\delta_{I_{s-}^{(k)}, \, t - s} - \delta_{I_{s-}^{(k)}, \, A_{s-}^{(k)} + t - s}) \, \mathds{1}_{0 \leq \theta \leq m_1(s,k,Z)  } \\[0.2cm] 
			+ (\delta_{I_{s-}^{(k)}, \, t - s} + \delta_{I_{s-}^{(k)}+1, \, t - s} - \delta_{I_{s-}^{(k)}, \, A_{s-}^{(k)} + t - s}) \, \mathds{1}_{m_1(s,k,Z)  \leq \theta \leq m_2(s,k,Z)}   \\
			+(2\delta_{I_{s-}^{(k)}+1, \, t - s} - \delta_{I_{s-}^{(k)}, \, A_{s-}^{(k)} + t - s}) \, \mathds{1}_{m_2(s,k,Z)\leq \theta \leq m_3(s,k,Z) }  \\[0.2cm] 
			 \text{	and 	}	m_1(s,k,Z) = b_{I_{s-}^{(k)}}(A_{s-}^{(k)})p^{(I_{s-}^{(k)})}_{2,0},  \\
			m_2(s,k,Z)  = b_{I_{s-}^{(k)}}(A_{s-}^{(k)})(p^{(I_{s-}^{(k)})}_{2,0} + p^{(I_{s-}^{(k)})}_{1,1} ), \quad
			m_3(s,k,Z)  =b_{I_{s-}^{(k)}}(A_{s-}^{(k)}) \, .
				\end{array}	
			\end{equation}
		
			\paragraph{Deterministic model} The cell population is represented by a population density function $\densityfunction := \big( \densityfunction^{(j)}(t,a) \big)_{j \in \llbracket 1, J \rrbracket } \in \mathbf{L}^1(\mathbb{R}_+)^J$ where $\rho^{(j)}(t,a) $  is the cell age density in layer $j$ at time $t$. The population evolves according to the following system of partial differential equations: 
			\begin{equation}\label{EDP_equation}
				\resizebox{0.9\textwidth}{!} {$		\left\{
						\begin{array}{l}
						\displaystyle \partial_{t}\densityfunction^{(j)}(t,a) + \partial_{a}\densityfunction^{(j)}(t,a) = -b_{j}(a)\densityfunction^{(j)}(t,a)  \\[0.2cm] 
						\displaystyle 	\densityfunction^{(j)}(t,0) = 2p^{(j-1)}_L\int_{0}^{\infty} b_{j-1}(a)\densityfunction^{(j-1)}(t,a)da + 2p^{(j)}_S \int_{0}^{\infty} b_{j}(a)\densityfunction^{(j)}(t,a)da \\[0.2cm] 
						\densityfunction(0,a) =\densityfunction_0(a)
					\end{array}
					\right. $}
				\end{equation}
			where $ \forall j \in \llbracket 1,J-1 \rrbracket, \,
			p^{(j)}_S  = \frac{1}{2}p^{(j)}_{1,1} + p^{(j)}_{2,0}, \, p_{L}^{(j)} :=  \frac{1}{2}p^{(j)}_{1,1} + p^{(j)}_{0,2}, \, p^{(0)}_L = 0  \, \text{and} \, p^{(J)}_S = 1 \, .$
			Here, $p^{(j)}_S $ is the probability that a cell taken randomly among both daughter cells, remains on the same layer and $p^{(j)}_L  = 1 - p^{(j)}_S$ is the probability that the cell moves.  
	\subsection{Hypotheses}\label{PreliminaryHypothesis} 
	\begin{hypothesis}\label{Hypothesis_Probability}
			$\forall j \in \llbracket 1, J-1 \rrbracket$, $p^{(j)}_S,p^{(j)}_L \in (0,1)$
	\end{hypothesis}
	
	\begin{hypothesis}\label{Hypothesis_DivisionRate}
		For each layer $j$, $b_j$ is continuous bounded below and above: 
		\begin{equation*}
		\forall j \in \llbracket 1,J \rrbracket, \quad \forall a \in \mathbb{R}_{+}, \quad 0<\underline{b}_j \leq b_j(a) \leq \overline{b}_j < \infty \, .
		\end{equation*}
	\end{hypothesis}

\begin{definition}\label{Definition_B_j_and_dB_j}
			$\mathcal{B}_j$ is the distribution function of $\tau^j$ ($\mathcal{B}_j(x) = 1 - e^{- \int_{0}^{x}b_j(a)da} $) and $d\mathcal{B}_j $ its density function ($d\mathcal{B}_j(x) = b_j(x) e^{- \int_{0}^{x}b_j(a)da} $).
	\end{definition}
	\begin{hypdef}{(Intrinsic growth rate)}\label{Malthus_parameter_per_layer_def}
		The intrinsic growth rate $\MalthusPerLayer_j$ of layer $j$ is the solution of 
		\begin{equation*}
			d\mathcal{B}_j^*(\MalthusPerLayer_j):= \int_{0}^{\infty}e^{-\MalthusPerLayer_j s } d\mathcal{B}_j(s)ds= \frac{1}{2p_S^{(j)}} \, .
		\end{equation*}
	\end{hypdef}
	\begin{remarks}\label{remarks_dB*}
			\textit{$d\mathcal{B}_j^*$ is the Laplace transform of $d\mathcal{B}_j$. It is a strictly decreasing function and $ \quad  ]-\underline{b}_j, \infty[  \subset \text{ Supp} (d\mathcal{B}_j^*)  \subset  ]-\overline{b}_j, \infty[$. Hence, $\MalthusPerLayer_j > - \overline{b}_j$. Moreover, note that $d\mathcal{B}_j^*(0) =  \int_{0}^{\infty}d\mathcal{B}_j(x)dx = 1$.
			Thus, $ \MalthusPerLayer_j < 0$ when $p_S^{(j)}< \frac{1}{2} $; $ \MalthusPerLayer_j > 0$ when $p_S^{(j)}> \frac{1}{2} $ and $\MalthusPerLayer_j = 0 $  when $ p_S^{(j)} = \frac{1}{2}$. In particular, $\MalthusPerLayer_J > 0 $ as $p_S^{(J)} = 1$.}
	\end{remarks}
	\begin{remarks}\textit{
		In the classical McKendrick-VonFoerster model (one layer), the population grows exponentially with rate $\lambda_1$ (\cite{webb_theory_1985}, Chap. IV). The same result is shown for the Bellman-Harris process in \cite{harris_theory_1963} (Chap. VI).}
	\end{remarks}
	\begin{hypdef}[Malthus parameter]\label{Malthus_parameter_def}
		\textit{The Malthus parameter $\Malthus$ is defined as the unique maximal element taken among the intrinsic growth rates ($\MalthusPerLayer_j$, $j \in \llbracket 1, J \rrbracket )$ defined in (\ref{Malthus_parameter_per_layer_def}). The layer such that the index $ j = \MalthusIndex$ is the leading layer.}
	\end{hypdef}
	According to remark \ref{remarks_dB*}, $\Malthus$ is positive. We will need auxiliary hypotheses on $\MalthusPerLayer_j$ parameters in some theorems.
	\begin{hypothesis}\label{Hypothesis_MalthusPerLayerDistinct}
			All the intrinsic growth rate parameters are distinct. 
	\end{hypothesis}
	\begin{hypothesis}\label{Hypothesis_lambdaj}
		$\forall j \in \llbracket 1, J \rrbracket$,  $ \MalthusPerLayer_j > - \underset{a \rightarrow + \infty}{\lim \inf \,} b_j(a) $.
	\end{hypothesis}
		Hypothesis \ref{Hypothesis_lambdaj}  implies additional regularity for $ t \mapsto e^{-\MalthusPerLayer_j t}d\mathcal{B}_j(t)$ (see proof in \ref{Supplemental_Deterministic}):
		
	\begin{corollary}\label{Moment_Distribution_dB}
		Under hypotheses \ref{Hypothesis_DivisionRate}, \ref{Malthus_parameter_per_layer_def} and \ref{Hypothesis_lambdaj},
		$\forall j \in \llbracket 1, J \rrbracket$, $\forall k \in \mathbb{N}$, \\  $\int_{0}^{\infty} t^k e^{-\MalthusPerLayer_j t}d\mathcal{B}_j(t)dt < \infty \, .$
	\end{corollary}

	\paragraph{Stochastic initial condition}
	We suppose that the initial measure $Z_0 \in \mathcal{M}_P$ is deterministic. $(\mathcal{F}_t)_{t \in \mathbb{R}_+} $ is the natural filtration associated with $(Z_t)_{t \in \mathbb{R}_+} $ and $Q$. 
	\paragraph{Deterministic initial condition}
	We suppose that the initial population density $ \densityfunction_0$ belongs to $\mathbf{L}^{1}(\mathbb{R}_{+})^J $.
	
	\subsection{Notation}\label{Part_Notations}
	Let $f,g \in  \mathbf{L}^1(\mathbb{R}_+)^J$, we use for the scalar product:
		\begin{itemize}
			\item on $\mathbb{R}_+^J$,  $f^T(a) g(a) = \sum_{j = 1}^{J} f^{(j)}(a)g^{(j)}(a)$,
			\item on $\mathbf{L}^{1}(\mathbb{R}_{+})$, $ 	\displaystyle	\langle f^{(j)},g^{(j)} \rangle = \int_{0}^{\infty}f^{(j)}(a)g^{(j)}(a)da$, for $j \in \llbracket 1, J \rrbracket$,
			\item on $ \mathbf{L}^{1}(\mathbb{R}_{+})^{J}$, $ \ll f,g \gg = \sum_{j = 1}^{J}\int_{0}^{\infty}f^{(j)}(a)g^{(j)}(a)da$.
		\end{itemize}  
		For a martingale $ M = (M_t)_{t \geq 0}$, we note $\left \langle M,M \right \rangle_t$ its quadratic variation. 
		We also introduce \begin{equation*}
		\resizebox{0.9\textwidth}{!} {$		B(a) = diag(b_1(a), ..., b_J(a)),
			\quad  [K(a)]_{i,j} = \left\{ \begin{array}{lll}
				2p_S^{(j)}b_j(a), &i = j, &  j \in \llbracket 1, J \rrbracket \\
				2p_L^{(j-1)}b_{j-1}(a), & i = j-1, & j \in \llbracket 2, J \rrbracket 
			\end{array}\right. $}
		\end{equation*}
		We define the primal problem $\eqref{StationnaireEquation}$ as 
		\begin{equation}
			\label{StationnaireEquation} \tag{P}
		\left\{ 
		\begin{array}{l}
		\mathcal{L}^P \eigenfunction(a) = \lambda\eigenfunction(a), \, a \geq 0\\
		\displaystyle \eigenfunction(0) = \int_{0}^{\infty} K(a) \eigenfunction(a)da \\
		\ll \eigenfunction, \mathds{1}\gg  = 1 \text{ and } \eigenfunction \geq 0
		\end{array}
		\right., \quad \mathcal{L}^P \eigenfunction(a) = \partial_a \eigenfunction(a) - B(a)\eigenfunction(a),
		\end{equation}
		and the dual problem $\eqref{ProblemeAdjoint}$ is given by
		\begin{equation}\label{ProblemeAdjoint}
		\tag{D}
		\left\{ 
		\begin{array}{lll}
		\mathcal{L}^D  \dualfunction(a) = \lambda \dualfunction(a), \, a \in \mathbb{R}^*_+ \\
		\ll\eigenfunction, \dualfunction \gg  = 1 \text{ and }  \dualfunction \geq 0 \\
		\end{array}
		\right., \quad \mathcal{L}^D \dualfunction(a) =  \partial_{a} \dualfunction(a) - B(a) \dualfunction + K(a)^T \dualfunction(0).
		\end{equation}
	\subsection{Main results}
	\subsubsection{	Eigenproblem approach}
	
	\begin{theorem}[Eigenproblem] \label{Eigenproblem}
		Under hypotheses \ref{Hypothesis_Probability}, \ref{Hypothesis_DivisionRate}, \ref{Malthus_parameter_per_layer_def}, \ref{Malthus_parameter_def} and \ref{Hypothesis_lambdaj}, there exists a first eigen\-element triple $( \lambda ,\eigenfunction, \dualfunction)$ solution to equations (\ref{StationnaireEquation}) and $(\ref{ProblemeAdjoint})$ where $\eigenfunction \in \mathbf{L}^1(\mathbb{R}_+)^J$ and $\dualfunction \in \mathcal{C}_b(\mathbb{R}_+)^J$. In particular, $\lambda$ is the Malthus parameter $\Malthus$ given in Definition \ref{Malthus_parameter_def}, and $\eigenfunction$ and $\dualfunction$ are unique. 
	\end{theorem}
	Beside the dual test function $\dualfunction$, we introduce other test functions to prove large-time convergence. Let $\hat{\phi}^{(j)} $, $j \in \llbracket 1, J \rrbracket$  be a solution of
		\begin{equation}\label{phihat_equ}
		\partial_{a}\dualfunctionperlayer^{(j)}(a) - (\MalthusPerLayer_j + b_{j}(a))\dualfunctionperlayer^{(j)}(a) =  -2p_{S}^{(j)}b_{j}(a)\dualfunctionperlayer^{(j)}(0), \quad \dualfunctionperlayer^{(j)}(0)  \in \mathbb{R}^*_+ \, .
		\end{equation}
	
	\begin{theorem}\label{Exponential_decay}
			Under hypotheses \ref{Hypothesis_Probability}, \ref{Hypothesis_DivisionRate}, \ref{Malthus_parameter_per_layer_def}, \ref{Malthus_parameter_def} and \ref{Hypothesis_lambdaj}, there exist polynomials  $(\beta^{(j)}_{k}) _{ 1 \leq k \leq j \leq J }$ of degree at most $j-k $ such that
		\begin{equation}\label{InegaliDecroisExp}
			\left \langle \big|e^{-\Malthus t}\densityfunction^{(j)}(t,\cdot) - \eta\hat{\rho}^{(j)}\big|,\dualfunctionperlayer^{(j)} \right \rangle  \quad \leq \quad \sum_{k = 1}^{j} e^{-\mu_j t} \beta^{(j)}_{k}(t) \left \langle \big|\rho^{(k)}_0- \eta \eigenfunction^{(k)}\big|,\dualfunctionperlayer^{(k)} \right \rangle, 
		\end{equation}
		where $ \eta := \, \ll \densityfunction_0,\phi\gg $, $\mu_j := \Malthus - \MalthusPerLayer_j > 0$ when $j \in \llbracket 1,J \rrbracket \setminus \{\MalthusIndex\} $ and $\mu_\MalthusIndex := \underline{b}_\MalthusIndex$. In particular, there exist a polynomial $\beta$ of degree at most $J-1 $ and constant $\mu$ such that
		\begin{equation*}
		\displaystyle \ll  \big|e^{-\Malthus t}\densityfunction(t,\cdot) - \eta\hat{\rho} \big| , \dualfunctionperlayer  \gg \, \leq  \, \beta(t) e^{-\mu t } 	\ll  \big|\densityfunction_0 - \eta\hat{\rho} \big| , \dualfunctionperlayer   \gg \, .
		\end{equation*}
	\end{theorem}
	Using martingale techniques \cite{jagers_population-size-dependent_2000}, we also prove a result of convergence for the stochastic process $Z$ with the dual test function $\dualfunction$.
	\begin{theorem}\label{Multi_Conver_L2}
		Under hypotheses \ref{Hypothesis_Probability}, \ref{Hypothesis_DivisionRate}, \ref{Malthus_parameter_per_layer_def} and \ref{Malthus_parameter_def}, $W^{\phi}_t = e^{-\Malthus t}\ll\dualfunction ,Z_t\gg  $  is a square integrable martingale that converges almost surely and in $\mathbf{L}^2$ to a non-de\-ge\-nerate random variable $W_\infty^\dualfunction$.
	\end{theorem}
	\subsubsection{Renewal equation approach}Using generating function methods developed for multi-type age dependent branching processes (see \cite{harris_theory_1963}, Chap. VI), we write a system of renewal equations and obtain analytical formulas for the two first moments. We define $Y^{(j,a)}_t:= \, \langle Z_t,\mathds{1}_{j, \leq a} \rangle$ as the number of cells on layer $j$ and of age less or equal than $a$ at time $t$, and $ m^a_{i}(t) $ its mean starting from one mother cell of age $0$ on layer $1$:
			\begin{equation}\label{Definition_Moment_M}
				m^a_{j}(t) := \mathbb{E}[Y^{(j,a)}_t|Z_0 = \delta_{1,0}] \, .
			\end{equation}
			
		\begin{theorem}\label{MoyenneTempsLong}
			Under hypotheses \ref{Hypothesis_Probability}, \ref{Hypothesis_DivisionRate}, \ref{Malthus_parameter_def}, \ref{Hypothesis_MalthusPerLayerDistinct} and \ref{Hypothesis_lambdaj}, for all $a \geq 0$,
			\begin{equation}\label{M_mean_StartOneLayer}
			\quad \forall j \in \llbracket 1, J \rrbracket, \quad m^a_{j}(t)e^{-\Malthus t} \rightarrow \widetilde{m}_{j}(a), \quad t \rightarrow \infty,
			\end{equation}
			\begin{multline*}
			\text{where } \widetilde{m}_{j}(a) = \\ \left\{
			\begin{array}{ll}
			0, &  j \in \llbracket 1, c-1 \rrbracket, \\
			\displaystyle \frac{ \int_{0}^{a} \eigenfunction^{(\MalthusIndex)}(s)ds}{ 2p^{(\MalthusIndex)}_{S}\eigenfunction^{(\MalthusIndex)}(0)\int_{0}^{\infty} s d\mathcal{B}_{\MalthusIndex}(s)e^{-\Malthus s} ds},  & j = \MalthusIndex, \\[0.5cm]
		\displaystyle	\frac{ \int_{0}^{a} \hat{\rho}^{(j)}(s)ds}{ 2p^{(\MalthusIndex)}_{S} \hat{\rho}^{(\MalthusIndex)}(0)\int_{0}^{\infty} s d\mathcal{B}_{\MalthusIndex}(s)e^{-\Malthus s} ds }  \prod_{k=1}^{\MalthusIndex - 1} \frac{2p^{(k)}_{L} d\mathcal{B}^{*}_{k}(\Malthus)}{ 1 - 2p^{(k)}_{S}d\mathcal{B}^{*}_{k}(\Malthus) }, &  j \in \llbracket c + 1, J \rrbracket.
			\end{array} \right.
			\end{multline*}
		\end{theorem}
	
	\subsubsection{Calibration}
		We now consider a particular choice of the division rate:
		\begin{hypothesis}[Age-independent division rate]\label{Hypothesis_ageIndep}
			 $\forall \, (j,a)\in \mathcal{E}$,  $b_j(a) = b_j$.
		\end{hypothesis}
		We also consider a specific initial condition with $N \in \mathbb{N}^*$ cells:
		\begin{hypothesis}[First layer initial condition]\label{Hypothesis_Start}
			$Z_0 = N \delta_{1,0}$.
		\end{hypothesis}
		 Then, integrating the deterministic PDE system (\ref{EDP_equation}) with respect to age or differentiating the renewal equation system (see (\ref{MomentDOrdre1Equ})) on the mean number $M$, we obtain:
		 \begin{equation}\label{ODE_equ}
		 \resizebox{0.9\textwidth}{!} {$
			 \left\{ 
		 	\begin{array}{l}
		 	\frac{d}{dt}M(t)= A M(t) \\
		 	M(0) = (N,0,...,0 ) \in \mathbb{R}^{J} 
		 	\end{array}\right., \quad [A]_{i,j} := \left\{ \begin{array}{lll}
		 	(2p_S^{(j)} - 1)b_j, &  i = j, & j \in \llbracket 1, J \rrbracket, \\
		 	2p_L^{(j-1)}b_{j-1}, &  i = j-1, & j \in \llbracket 2, J \rrbracket.
		 	\end{array}\right. $}
		 \end{equation}
		 We prove the structural identifiability of the parameter set $ \mathbf{P}:= \{N, b_j,p_S^{(j)},  j \in \llbracket 1, J \rrbracket \}$ when we observe the vector $M(t;\mathbf{P})$ at each time $t$.
		\begin{theorem}\label{Iden_Them_All_Layer}
			Under hypotheses \ref{Hypothesis_Probability}, \ref{Hypothesis_ageIndep} and \ref{Hypothesis_Start} and complete observation of system (\ref{ODE_equ}), the parameter set $\mathbf{P}$ is identifiable.
		\end{theorem}
		We then perform the estimation of the parameter set $\mathbf{P}$ from experimental cell number data retrieved on four layers and sampled at three different time points (see Table \ref{Table_SummarizeOfDataset}). To improve practical identifiability, we embed biological specifications used in \cite{clement_coupled_2013} as a recurrence relation between successive division rates:
		\begin{equation}\label{Rec_b}
			b_j = \frac{b_1}{1 + (j - 1) \times \alpha}, \, j \in \llbracket 1, 4 \rrbracket, \, \alpha \in \mathbb{R}.
		\end{equation}
		We estimate the parameter set $\mathbf{P}_{exp} = \{N, b_1,\alpha, p_S^{(1)}, p_S^{(2)},p_S^{(3)}\}$ using the D2D software \cite{raue_data2dynamics:_2015} with an additive Gaussian noise model (see Figure \ref{fig:ParametersEstim_Fit} and Table \ref{Table_ParametersValues}). An analysis of the profile likelihood estimate shows that all parameters except $p_S^{(2)} $ are practically identifiable (see Figure \ref{fig:pledatamntotalcell}).
		\begin{figure}[!htb]
			\centering
			\includegraphics[height=4cm]{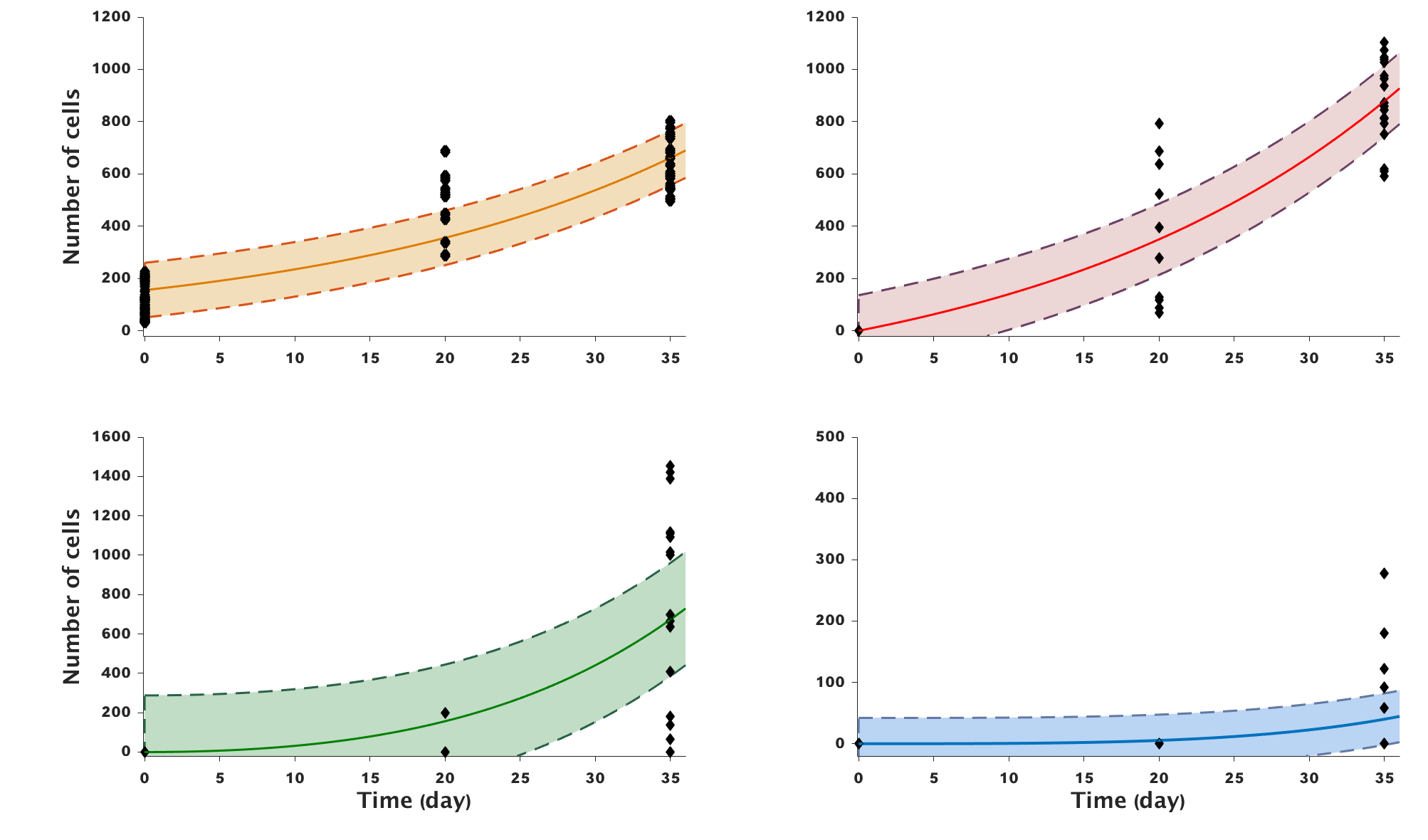}
			\label{fig:ParametersEstim_Fit}
			\caption{\textbf{Data fitting with model (\ref{ODE_equ}).} Each panel illustrates the changes in the cell number in a given layer (top-left: Layer 1, top-right: Layer 2, bottom-left: Layer 3, bottom-right: Layer 4). The black diamonds represent the experimental data, the solid lines are the best fit solutions of (\ref{ODE_equ}) and the dashed lines are drawn from the estimated variance. The parameter values (Table \ref{Table_ParametersValues}) are estimated according to the procedure described in section \ref{SM_ParmeterEstimaProcedure}. }
		\end{figure}

\section{Theoretical proof and illustrations}
			\subsection{Eigenproblem}
				We start by solving explicitly the eigenproblem (\ref{StationnaireEquation})-(\ref{ProblemeAdjoint}) to prove theorem \ref{Eigenproblem}.
				\begin{proof}[Proof of theorem \ref{Eigenproblem}]
				According to definition \ref{Definition_B_j_and_dB_j}, any solution of (\ref{StationnaireEquation}) in $\mathbf{L}^1(\mathbb{R}_+)^J$ is given by,  $\forall j \in \llbracket 1, J \rrbracket$,	
				\begin{equation}\label{Eigenproblem_rhoEqu}
				\eigenfunction^{(j)}(a) =  \eigenfunction^{(j)}(0) e^{- \lambda a} (\mathds{1} - \mathcal{B}_j)(a) \, .
				\end{equation}
				The boundary condition of the problem (\ref{StationnaireEquation}) gives us a system of equations for $\lambda$ and $ \eigenfunction^{(j)}(0)$, $j \in \llbracket 1, J \rrbracket$:
					\begin{equation}
						 \eigenfunction^{(j)}(0)  \times (1 - 2p^{(j)}_{S}d\mathcal{B}_{j}^{*}(\lambda)) = 2p^{(j-1)}_{L}d\mathcal{B}_{j-1}^{*}(\lambda) \times \eigenfunction^{(j-1)}(0) \, . \label{rho_0_hat_equation}
					\end{equation}
				This system is equivalent to
				\begin{align*}
					C(\lambda)\eigenfunction(0) = 0, \quad [C(\lambda)]_{i,j} = \left\{ \begin{array}{ll}
						1 - 2p^{(j)}_{S}d\mathcal{B}_{j}^{*}(\lambda), &i = j , \quad j \in \llbracket 1, J \rrbracket, \\
						2p^{(j-1)}_{L}d\mathcal{B}_{j-1}^{*}(\lambda) , &i = j-1, \quad j \in \llbracket 2, J \rrbracket.
					\end{array}\right. 
				\end{align*}
			 Let $\Lambda := \{\MalthusPerLayer_j, j \in \llbracket 1, J \rrbracket \}$. The eigenvalues of the matrix $C(\lambda)$ are $1 - 2p^{(j)}_{S}d\mathcal{B}_{j}^{*}(\lambda) $, $j \in \llbracket 1, J \rrbracket$. Thus, if $\lambda \notin \Lambda$, according to hypothesis \ref{Malthus_parameter_per_layer_def}, $0$ is not an eigenvalue of $C(\lambda)$ which implies that $\eigenfunction(0) =0 $. As $\eigenfunction$ satisfies both (\ref{Eigenproblem_rhoEqu}) and the normalization $\ll\eigenfunction,\mathds{1}\gg \,  = 1 $, we obtain a contradiction. So, necessary $\lambda \in \Lambda $. \\
			 We choose $\lambda = \Malthus$ the maximum element of $\Lambda$ according to hypothesis \ref{Malthus_parameter_def}. Then, using (\ref{rho_0_hat_equation}) when $j = \MalthusIndex$, we have:
			 \begin{equation*}
			 	\eigenfunction^{(\MalthusIndex)}(0) \times (1 - 2p^{(\MalthusIndex)}_{S}d\mathcal{B}_{\MalthusIndex}^{*}(\Malthus) )= 2p^{(\MalthusIndex-1)}_{L}d\mathcal{B}_{\MalthusIndex-1}^{*}(\Malthus) \times \eigenfunction^{(\MalthusIndex-1)}(0) \, .
			 \end{equation*}
			 Note that $ 1 - 2p^{(\MalthusIndex)}_{S}d\mathcal{B}_{\MalthusIndex}^{*}(\Malthus) = 0 $, so $ \eigenfunction^{(\MalthusIndex-1)}(0) = 0$ and by backward recurrence using (\ref{rho_0_hat_equation}) from $j=\MalthusIndex - 1$ to $1$, it comes that $\eigenfunction^{(j)}(0) = 0$ when $j< \MalthusIndex $.
			By hypothesis \ref{Malthus_parameter_def}, $\max (\Lambda)$ is unique. Thus, when $j > \MalthusIndex$, $\MalthusPerLayer_j \neq \Malthus$ and $1 - 2p^{(j)}_{S}d\mathcal{B}_{j}^{*}(\Malthus) \neq 0$. Solving (\ref{rho_0_hat_equation}) from $j = \MalthusIndex + 1$ to $J$, we obtain:
				\begin{equation*}
				\eigenfunction^{(j)}(0) =  \eigenfunction^{(\MalthusIndex)}(0) \times \prod_{k = \MalthusIndex + 1 }^{j} \frac{2p^{(k-1)}_{L}d\mathcal{B}_{k-1}^{*}(\Malthus)}{1 - 2p^{(k)}_{S}d\mathcal{B}_{k}^{*}(\Malthus)},   \quad \forall j 	 \in \llbracket \MalthusIndex + 1,J\rrbracket \, .
				\end{equation*}
				We deduce $\eigenfunction^{(\MalthusIndex)}(0) $ from the normalization $\ll\eigenfunction, \mathds{1}\gg \,  = 1$. Hence, $\eigenfunction$ is uniquely determined by (\ref{Eigenproblem_rhoEqu}) together with the following boundary value:
				\begin{equation}\label{Eigenproblem_rho0_value}
				\eigenfunction^{(j)}(0) = \left\{
						\begin{array}{ll}
						0, & j \in  \llbracket 1,\MalthusIndex - 1 \rrbracket,  \\[0.35cm]
						\frac{1}{\sum_{j= \MalthusIndex }^{J} \int_{0}^{\infty}\eigenfunction^{(j)}(a)da \prod_{k = \MalthusIndex + 1}^{j} \frac{2p^{(k-1)}_L d\mathcal{B}_{k-1}^{*}(\MalthusPerLayer_\MalthusIndex)}{1 - 2p^{(k)}_Sd\mathcal{B}_{k}^{*}(\Malthus)} } \, , & j = \MalthusIndex, \\[0.7cm]
						\eigenfunction^{(\MalthusIndex)}(0) \prod_{k = \MalthusIndex + 1}^{j} \frac{2p^{(k-1)}_Ld\mathcal{B}_{k-1}^{*}(\MalthusPerLayer_\MalthusIndex)}{1 - 2p^{(k)}_Sd\mathcal{B}_{k}^{*}(\Malthus)} \, , & j \in  \llbracket \MalthusIndex + 1,J \rrbracket.
						\end{array}
						\right. 
				\end{equation}
				For the ODE system (\ref{ProblemeAdjoint}), any solution  is given by, for $j \in \llbracket 1, J \rrbracket$,
				\begin{equation*}
				\displaystyle	\dualfunction^{(j)}(a) =  \left [ \dualfunction^{(j)}(0) - 2 \big(\dualfunction^{(j)}(0)p_S^{(j)} + \dualfunction^{(j+1)}(0)p_L^{(j)} \big)\int_{0}^{a}e^{-\Malthus s}d\mathcal{B}_j(s)ds  \right ]  e^{\int_{0}^{a}\Malthus+ b_{j}(s)ds} \, .
				\end{equation*}
				As $\displaystyle \int_{0}^{a}b_{j}(s)e^{-\int_{0}^{s}\Malthus + b_{j}(u)du}ds $ is equal to $ \displaystyle d\mathcal{B}_{j}^{*}(\Malthus)  - \int_{a}^{\infty}b_{j}(s)e^{-\int_{0}^{s}\Malthus + b_{j}(u)du}ds $, we get
				\begin{multline*}
					\dualfunction^{(j)}(a) 
				\displaystyle	= \left [ \dualfunction^{(j)}(0) \left (1 - 2p_S^{(j)}d\mathcal{B}_{j}^{*}(\Malthus)+ 2p_S^{(j)}\int_{a}^{+ \infty}b_{j}(s)e^{-\int_{0}^{s}\Malthus + b_{j}(u)du}ds \right ) \right.  \\
					\left. 
					 - \dualfunction^{(j+1)}(0) \left (  2p_L^{(j)}d\mathcal{B}_{j}^{*}(\Malthus)- 2p_L^{(j)}\int_{a}^{+ \infty}b_{j}(s)e^{-\int_{0}^{s}\Malthus + b_{j}(u)du}ds\right ) \right ]  e^{\int_{0}^{a}\Malthus+ b_{j}(s)ds} \, .
				\end{multline*}
				Searching for $\dualfunction \in \mathcal{C}_b(\mathbb{R}_+)^J$, it comes that 
				\begin{equation}\label{Recursive_equation}
				\forall j \in \llbracket1, J \rrbracket , \quad  \dualfunction^{(j)}(0)\left (1 - 2p_S^{(j)}d\mathcal{B}^*_{j}(\Malthus) \right )  - \dualfunction^{(j+1)}(0)2p_L^{(j)}d\mathcal{B}^*_{j}(\Malthus)= 0 \, .
				\end{equation}
				According to definition \ref{Malthus_parameter_per_layer_def}, when $j = \MalthusIndex$ in (\ref{Recursive_equation}) we get $\dualfunction^{(\MalthusIndex+1)}(0) = 0$. Recursively, $ \dualfunction^{(j)}(0) = 0$ when $j > \MalthusIndex $.
				Solving (\ref{Recursive_equation}) from $j=1$ to $\MalthusIndex - 1$, we get 
					\begin{equation}\label{bound_phi}
				\forall j 	 \in \llbracket 1,\MalthusIndex - 1 \rrbracket, \quad	\dualfunction^{(j)}(0)=  \dualfunction^{(\MalthusIndex)}(0) \times \prod_{k = j}^{\MalthusIndex - 1 } \frac{2p^{(k-1)}_{L}d\mathcal{B}_{k-1}^{*}(\Malthus)}{1 - 2p^{(k)}_{S}d\mathcal{B}_{k}^{*}(\Malthus)}  \, .
				\end{equation}
				Again, we deduce $\dualfunction^{(\MalthusIndex)}(0)$ from the normalization $  1 = \, \ll\eigenfunction, \dualfunction\gg \, = \langle \eigenfunction^{(\MalthusIndex)},\dualfunction^{(\MalthusIndex)} \rangle $. Using corollary \ref{Moment_Distribution_dB}, we apply Fubini theorem:
				\begin{equation}\label{bound_phi_2}
				\dualfunction^{(\MalthusIndex)}(0) = \frac{1}{	2\eigenfunction^{(\MalthusIndex)}(0)p_S^{(\MalthusIndex)}\int_{0}^{\infty}\big(\int_{a}^{+ \infty}e^{-\MalthusPerLayer_\MalthusIndex s}d\mathcal{B}_\MalthusIndex(s)ds \big)da} = \frac{1}{	2\eigenfunction^{(\MalthusIndex)}(0)p_S^{(\MalthusIndex)} \int_{0}^{\infty}se^{- \Malthus s}d\mathcal{B}_{\MalthusIndex}(s) ds}.
				\end{equation}
				Hence, the dual function $\dualfunction$ is uniquely determined by
				\begin{equation}\label{Eigenproblem_phi}
					\dualfunction^{(j)}(a) = 2 \left [p_S^{(j)}\dualfunction^{(j)}(0) + p_L^{(j)}\dualfunction^{(j+1)}(0) \right ] \int_{a}^{+ \infty}b_{j}(s)e^{-\int_{a}^{s}\Malthus + b_{j}(u)du}ds \, .
				\end{equation}
				together with the boundary value \eqref{bound_phi} and \eqref{bound_phi_2} ($\dualfunction$ is null on the layers upstream the leading layer).
			\end{proof}
			From theorem \ref{Eigenproblem}, we deduce the following bounds on $\dualfunction$ (see proof in \ref{Supplemental_Deterministic}).
			
			\begin{corollary}\label{Bornesd}
				According to hypotheses \ref{Hypothesis_DivisionRate}, \ref{Malthus_parameter_per_layer_def} and \ref{Malthus_parameter_def},
				\begin{equation}\label{Coro_ineq}
				\forall j \in \llbracket 1, J \rrbracket, \quad 	\frac{\underline{b}_j}{\Malthus + \overline{b}_j} \leq   \frac{\dualfunction^{(j)}(a)}{2 [p_S^{(j)}\dualfunction^{(j)}(0) + p_L^{(j)}\dualfunction^{(j+1)}(0) ]}\leq 1.
				\end{equation}
			\end{corollary}
			To conclude this section, we also solve the additional dual problem on isolated layers which is needed to obtain the large-time convergence (see proof in \ref{Supplemental_Deterministic}).
			\begin{lemma}\label{BornedHat}
				According to hypotheses \ref{Hypothesis_DivisionRate}, \ref{Malthus_parameter_per_layer_def} and \ref{Hypothesis_lambdaj}, any solution $\dualfunctionperlayer $ of \eqref{phihat_equ} satisfies
				\begin{equation}\label{EDO_hat_phi}
				\forall j  \in \llbracket 1, J \rrbracket, \quad	\dualfunctionperlayer^{(j)}(a) =  2p_S^{(j)}\dualfunctionperlayer^{(j)}(0)\int_{a}^{+ \infty} b_j(s)e^{-\MalthusPerLayer_j s - \int_{a}^{s}b_j(u)du}ds
				\end{equation}
				and, $	\forall a \in \mathbb{R}_+  \cup \{ + \infty\}$, $\frac{ \underline{b}_j}{\MalthusPerLayer_j + \overline{b}_j} \leq \frac{\dualfunctionperlayer^{(j)}(a)}{2p_S^{(j)}\dualfunctionperlayer^{(j)}(0)} \, < \, + \infty \, .$
			\end{lemma} 
			In all the sequel, we fix 
			\begin{equation}\label{phi_hat_and_phi_0}
			 \dualfunctionperlayer^{(\MalthusIndex)}(0) = \dualfunction^{(\MalthusIndex)}(0) , \quad \forall j \in \llbracket 1, \MalthusIndex - 1 \rrbracket \quad \dualfunctionperlayer^{(j)}(0)   =\dualfunction^{(j)}(0) + \frac{p_L^{(j)}}{p_S^{(j)}} \dualfunction^{(j+1)}(0). 
			\end{equation}
				A first consequence is that $\dualfunctionperlayer^{(\MalthusIndex)} = \dualfunction^{(\MalthusIndex)} $ and moreover, from corollary \ref{Bornesd} and lemma \ref{BornedHat}, we have 
				\begin{equation}\label{Remarks_dualvsdualperlayer}
					\dualfunction^{(j)}(a)  \leq    \frac{{\MalthusPerLayer_j + \overline{b}_j}}{\underline{b}_j } \dualfunctionperlayer^{(j)}(a) \, .
				\end{equation}	
			
			\subsection{Asymptotic study for the deterministic formalism}
			Adapting the method of characteristic, it is classical to construct the unique solution in \\ $\mathcal{C}^{1}\big(\mathbb{R}_+,\mathbf{L}^{1}(\mathbb{R}_{+})^J\big)$ of \eqref{EDP_equation} (\cite{webb_theory_1985}, Chap. I). 
			Let  $\densityfunction $ the solution of (\ref{EDP_equation}),  $\eigenfunction $ and $ \dualfunction$ given by theorem \ref{Eigenproblem} and  $\eta =  \ll\densityfunction_0,\dualfunction\gg \, $. We define $h$ as
			\begin{equation}\label{Definition_nt}
			h(t,a) = e^{-\Malthus t}\densityfunction(t,a) - \eta \eigenfunction(a), \quad (t,a) \in \mathbb{R}_+ \times \mathbb{R}_+ \, .
			\end{equation}
			Following \cite{michel_general_2005}, we first show a conservation principle (see proof in \ref{Supplemental_Deterministic}).
				\begin{lemma}[Conservation principle]\label{PrincipeDeConservation}The function $h$ satisfies the conservation principle
				$$\ll h(t,\cdot), \dualfunction \gg \,  = 0 \, .$$
			\end{lemma}
			Secondly, we prove that $h$ is solution of the following PDE system (see proof in \ref{Supplemental_Deterministic}).
			
			\begin{lemma}\label{ValeurAbsolue}
				$h$ is solution of
				\begin{equation}\label{Equ1}
				\left\{
				\begin{array}{lll}
				\partial_{t}\big|h(t,a)\big| + \partial_{a}\big|h(t,a)\big| + \left  (\Malthus +B(a) \right) \big|h(t,a)\big| = 0, \\
				\big|h(t,0)\big| = \big|\int_{0}^{+ \infty}K(a)h(t,a)da\big|.
				\end{array}
				\right.
				\end{equation}	
			\end{lemma}
			Together with the above lemmas \ref{BornedHat}, \ref{PrincipeDeConservation} and \ref{ValeurAbsolue}, we now prove the following key estimates required for the asymptotic behavior.
			
			\begin{lemma}\label{Exponential_decay_lemma}
				$\forall j \in \llbracket 1, J \rrbracket$, the component $h^{(j)}$  of  $h$ verifies the inequality 
				\begin{equation}\label{Inegalite_nt}
				\partial_{t}\left \langle \big|h^{(j)}(t,\cdot)\big|,\dualfunctionperlayer^{(j)}\right \rangle \, \leq \,  \alpha_{j-1} \left \langle |h^{(j-1)}(t,\cdot)|,\dualfunctionperlayer^{(j-1)}\right \rangle - \mu_j\left \langle  \big|h^{(j)}(t,\cdot)\big|,\dualfunctionperlayer^{(j)}\right \rangle + \, r_j(t) \, ,
				\end{equation} 
				where $\alpha_0 := 0$,  for $j \in \llbracket 1, J \rrbracket $, $\alpha_j := \frac{p_L^{(j)}}{p_S^{(j)}} \frac{\overline{b}_{j}}{\underline{b}_{j} } \frac{\dualfunctionperlayer^{(j+1)}(0)}{\dualfunctionperlayer^{(j)}(0)} ( \MalthusPerLayer_{j} + \overline{b}_{j}) $ and
				\begin{equation*}
		\mu_j  =	\left\{
			\begin{array}{ll}
				 \Malthus - \MalthusPerLayer_j, & 	j\neq \MalthusIndex  \\
				\underline{b}_\MalthusIndex, & j = \MalthusIndex
			\end{array}
			\right., \,	r_j(t) :=	\left\{
						\begin{array}{ll}
						0, & 	j\neq \MalthusIndex\\
						\displaystyle \sum_{j=1}^{\MalthusIndex-1} \frac{{\MalthusPerLayer_j + \overline{b}_j}}{\underline{b}_j } \left \langle \big|h^{(j)}(t,\cdot)\big|,\dualfunctionperlayer^{(j)}\right \rangle, & j = \MalthusIndex \, .
						\end{array}
					\right.
				\end{equation*} 
			\end{lemma}
				\begin{proof}[Proof of lemma \ref{Exponential_decay_lemma}] Remind that $p_L^{(0)} = 0$ so that all the following computations are consistent with $j = 1$.
				Multiplying (\ref{Equ1}) by $\dualfunctionperlayer$ and using \eqref{phihat_equ}, it comes for any $j$
				\begin{align}\label{Proof_h_system}
				 \resizebox{0.9\textwidth}{!} {$\left\{	\begin{array}{c}
				\partial_t \big|h^{(j)}(t,a)\big| \dualfunctionperlayer^{(j)}(a) + \partial_{a}\big|h^{(j)}(t,a)\big| \dualfunctionperlayer^{(j)}(a)= 
				-2p_S^{(j)}\dualfunctionperlayer^{(j)}(0)b_j(a)\big|h^{(j)}(t,a)\big| + [\MalthusPerLayer_j - \Malthus]\big|h^{(j)}(t,a)\big| \dualfunctionperlayer^{(j)}(a),\\[0.3cm]
				\big|h^{(j)}(t,0)\big| \dualfunctionperlayer^{(j)}(0) = \dualfunctionperlayer^{(j)}(0) \big|2p_S^{(j)} \left \langle b_{j},h^{(j)}(t,\cdot) \right \rangle  + 2p_L^{(j-1)} \left \langle b_{j-1}, h^{(j-1)}(t,\cdot) \right  \rangle \big|.
				\end{array}
				\right. $}
				\end{align}
				As $\densityfunction(t,\cdot)$ and $\eigenfunction$ belong to $\mathbf{L}^1(\mathbb{R}_+)^J $ and $\dualfunctionperlayer$ is a bounded function (from lemma \ref{BornedHat}) we deduce that $ \ll h(t,\cdot), \dualfunctionperlayer \gg < \infty$.
				Integrating (\ref{Proof_h_system}) with respect to age, we have
				 \begin{multline}\label{Equality_hvsphi}
				\partial_{t} \left \langle \big|h^{(j)}(t,\cdot)\big|,\dualfunctionperlayer^{(j)} \right  \rangle \quad 
				 	=
				 	  \dualfunctionperlayer^{(j)}(0)  \left [ \big|h^{(j)}(t,0)\big|-2p_S^{(j)} \left \langle \big|h^{(j)}(t,\cdot)\big|,b_j \right \rangle \right ]  \\
				 	  + (\MalthusPerLayer_j - \Malthus) \left \langle \big|h^{(j)}(t,\cdot)\big|,\dualfunctionperlayer^{(j)} \right \rangle  \, .
				 \end{multline}
				We deal with the first term in the right hand-side of \eqref{Equality_hvsphi}. When $j \neq \MalthusIndex$, using first the boundary value in \eqref{Proof_h_system}, a  triangular inequality and lemma \ref{BornedHat}, we get
					\begin{align*}
					 \dualfunctionperlayer^{(j)}(0) \left ( \big|h^{(j)}(t,0)\big|-2p_S^{(j)} \left \langle  \big|h^{(j)}(t,\cdot)\big|,b_j \right \rangle  \right )& \leq & 2p_L^{(j-1)}  \dualfunctionperlayer^{(j)}(0)\left \langle \big|h^{(j-1)}(t,\cdot)\big|,b_{j-1}\right \rangle \\
					& \leq & \alpha_{j-1}\left \langle |h^{(j-1)}(t,\cdot)|,\dualfunctionperlayer^{(j-1)}\right \rangle \, .
				\end{align*}
				 Thus, for $j \neq c$,
				 \begin{equation*}
				 \partial_{t} \left \langle  \big|h^{(j)}(t,\cdot)\big|,\dualfunctionperlayer^{(j)} \right \rangle \quad
				 \leq  \quad \alpha_{j-1}  \left \langle  |h^{(j-1)}(t,\cdot)|,\dualfunctionperlayer^{(j-1)} \right \rangle - \mu_j  \left \langle   \big|h^{(j)}(t,\cdot)\big|,\dualfunctionperlayer^{(j)} \right \rangle   \, .
				 \end{equation*}
				 When $j = \MalthusIndex$, using the boundary value in \eqref{Proof_h_system} and a triangular inequality, we get
				\begin{multline}\label{InequA}
						\partial_{t}\left \langle\big|h^{(\MalthusIndex)}(t,\cdot)\big|,\dualfunctionperlayer^{(\MalthusIndex)} \right \rangle \quad 
						\leq 
					  2p_S^{(\MalthusIndex)}\dualfunctionperlayer^{(\MalthusIndex)}(0)  \big[ \big|\left \langle h^{(\MalthusIndex)}(t,\cdot),b_\MalthusIndex \right \rangle \big|-\left \langle \big|h^{(\MalthusIndex)}(t,\cdot)\big|,b_\MalthusIndex \right \rangle \big] \\
					  + 2p_L^{(\MalthusIndex-1)}\dualfunctionperlayer^{(\MalthusIndex)}(0) \big|\left \langle h^{(\MalthusIndex-1)}(t,\cdot),b_{\MalthusIndex-1} \right \rangle \big| \, .
				\end{multline}
				To exhibit a term $\left \langle\big|h^{(\MalthusIndex)}(t,\cdot)\big|,\dualfunctionperlayer^{(\MalthusIndex)} \right \rangle$ in the right hand-side of \eqref{InequA}, we need a more refined analysis. 
				According to the conservation principle (lemma \ref{PrincipeDeConservation}), for any constant $ \gamma$ (to be chosen later), we obtain  
				\begin{equation}\label{IneC}
					\resizebox{1.\textwidth}{!} { $ \begin{array}{rl}
					 2p_S^{(\MalthusIndex)}\dualfunctionperlayer^{(\MalthusIndex)}(0) \big|\left \langle h^{(\MalthusIndex)}(t,\cdot),b_\MalthusIndex \right \rangle  \big| = &  \big|2p_S^{(\MalthusIndex)}\dualfunctionperlayer^{(\MalthusIndex)}(0) \left \langle h^{(\MalthusIndex)}(t,\cdot),b_\MalthusIndex \right \rangle  - \gamma \ll h(t,\cdot), \dualfunction\gg \, \big| \\
					\leq & \big| \left \langle h^{(\MalthusIndex)}(t,\cdot),2p_S^{(\MalthusIndex)}\dualfunctionperlayer^{(\MalthusIndex)}(0) b_{\MalthusIndex} - \gamma \dualfunction^{(\MalthusIndex)} \right \rangle  \big|   + \gamma \sum_{j = 1}^{\MalthusIndex-1} \left \langle \big|h^{(j)}(t,\cdot)\big|,\dualfunction^{(j)}\right \rangle. 
					\end{array} $}
				\end{equation}
				where we used a triangular inequality in the latter estimate. Moreover, according to \eqref{Remarks_dualvsdualperlayer}, we have
				\begin{equation}\label{IneqD}
					\forall j \in \llbracket 1, \MalthusIndex - 1 \rrbracket, \quad \left \langle \big|h^{(j)}(t,\cdot)\big|,\dualfunction^{(j)}\right \rangle  \quad \leq \quad \frac{{\MalthusPerLayer_j + \overline{b}_j}}{\underline{b}_j } \left \langle \big|h^{(j)}(t,\cdot)\big|,\dualfunctionperlayer^{(j)} \right \rangle \, ,
				\end{equation}
				and according to corollary \ref{Bornesd},
				\begin{equation}\label{Gamma_choice}
					\dualfunction^{(\MalthusIndex)}(a)  \leq \frac{2p_S^{(\MalthusIndex)}\dualfunction^{(\MalthusIndex)}(0) }{\underline{b}_\MalthusIndex} b_\MalthusIndex(a). 
				\end{equation}
				
				We want to find at least one constant $\gamma$ such that for all $a \geq 0$ , $2p_S^{(\MalthusIndex)}\dualfunctionperlayer^{(\MalthusIndex)}(0) b_{\MalthusIndex}(a) - \gamma \dualfunction^{(\MalthusIndex)}(a) > 0 $.
				From \eqref{Gamma_choice}, we choose $\gamma = \underline{b}_\MalthusIndex $, and deduce from (\ref{IneC}) and (\ref{IneqD})
					\begin{equation}\label{InequB}
						\resizebox{0.9\textwidth}{!} {$ 
						\begin{array}{rl}
					2p_S^{(\MalthusIndex)}\dualfunctionperlayer^{\MalthusIndex}(0) \big| \left \langle   h^{(\MalthusIndex)}(t,\cdot),b_\MalthusIndex \right \rangle  \big| \leq &  2p_S^{(\MalthusIndex)}\dualfunctionperlayer^{(\MalthusIndex)}(0) \left \langle  \big| h^{(\MalthusIndex)}(t,\cdot) \big|, b_{\MalthusIndex} \right \rangle   - \underline{b}_\MalthusIndex\left \langle  \big| h^{(\MalthusIndex)}(t,\cdot) \big|, \dualfunction^{(\MalthusIndex)}\right \rangle \\
					& + \quad \underline{b}_\MalthusIndex\sum_{j = 1}^{\MalthusIndex-1}  \frac{{\MalthusPerLayer_j + \overline{b}_j}}{\underline{b}_j }  \left \langle  \big|h^{(j)}(t,\cdot)\big|,\dualfunctionperlayer^{(j)} \right \rangle. 
						\end{array}  $}
				\end{equation}
			As before, using lemma \ref{BornedHat}, we obtain
			\begin{equation*}
				 2p_L^{(\MalthusIndex-1)}\dualfunctionperlayer^{(\MalthusIndex)}(0) \big|\left \langle h^{(\MalthusIndex-1)}(t,\cdot),b_{\MalthusIndex-1} \right \rangle \big| \leq \alpha_{\MalthusIndex-1} \left \langle \big|h^{(\MalthusIndex-1)}(t,\cdot)\big|,\dualfunctionperlayer^{(\MalthusIndex-1)}\right \rangle \, .
			\end{equation*}
			Combining the latter inequality with (\ref{InequB}) and (\ref{InequA}), we deduce \eqref{Inegalite_nt} for $j = c$.
			\end{proof}
			
			We now have all the elements to prove theorem \ref{Exponential_decay}.
				\begin{proof}[Proof of theorem \ref{Exponential_decay}]
				We proceed by recurrence from the index $j = 1$ to $J$. For $j = 1$, we can apply Gronwall lemma in inequality \eqref{Inegalite_nt} to get
					\begin{eqnarray*}
									\left \langle |h^{(1)}(t,\cdot)|,\dualfunctionperlayer^{(1)}\right \rangle \leq 
									& e^{- \mu_1 t} \left \langle |h^{(1)}(0,\cdot)|,\dualfunctionperlayer^{(1)}\right \rangle \, .
								\end{eqnarray*}
				We suppose that for a fixed $ 2 \leq j \leq J$ and for all ranks $ 1 \leq i \leq j-1$, there exist polynomials $\beta^{(i)}_{k}$, $k \in \llbracket 1, i\rrbracket$, of degree at most $i -k $ such that
				\begin{align}\label{Expo_decay_rec}
						\left \langle |h^{(i)}(t,\cdot)|,\dualfunctionperlayer^{(i)}\right \rangle \quad  \leq \quad 
									 \sum_{k=1}^{i}\beta^{(i)}_{k}(t)e^{- \mu_k t}\left \langle |h^{(k)}(0,\cdot)|,\dualfunctionperlayer^{(k)}\right \rangle \, .
					\end{align}
					Applying this recurrence hypothesis in inequality \eqref{Inegalite_nt} for $j$, there exist polynomials $\widetilde{\beta}^{(j)}_k(t) $ for $k \in \llbracket 1, j-1 \rrbracket$ (same degree than $\beta^{(j-1)}_k(t)$ ):
				\begin{eqnarray*}
					\partial_{t}\left \langle \big|h^{(j)}(t,\cdot)\big|,\dualfunctionperlayer^{(j)}\right \rangle
					& \leq \sum_{k=1}^{j-1} \widetilde{\beta}^{(j)}_{k}(t)e^{- \mu_k t}\left \langle |h^{(k)}(0,\cdot)|,\dualfunctionperlayer^{(k)}\right \rangle -  \mu_j\left \langle  \big|h^{(j)}(t,\cdot)\big|,\dualfunctionperlayer^{(j)}\right \rangle \, .
				\end{eqnarray*} 
				We  get from a modified version of Gronwall lemma (see lemma \ref{LemmeGronwallForUs}):
				\begin{equation*}
					\left \langle |h^{(j)}(t,\cdot)|,\dualfunctionperlayer^{(j)}\right \rangle \,
					\leq \,  \sum_{k=1}^{j} \beta^{(j)}_{k}(t) e^{- \mu_k t} \left \langle |h^{(k)}(0,\cdot)|,\dualfunctionperlayer^{(k)}\right \rangle \, .
				\end{equation*}
				where $ \beta^{(j)}_{j}$ is a constant and for $k \in \llbracket 1,  j-1 \rrbracket$, $ \beta^{(j)}_{k}$ is a polynomial of degree at most $(j-1 - k) + 1 = j-k$ (the degree only increases by $1$ when $\mu_k = \mu_j $).
				This achieves the recurrence. 
			\end{proof}
			
			\subsection{Asymptotic study of the martingale problem}
				The existence and uniqueness of the SDE (\ref{Z_t_Equation}) is proved in a more general context than ours in \cite{tran_large_2008}.
				Following the approach proposed in \cite{tran_large_2008}, we first derive the generator of the process Z solution of (\ref{Z_t_Equation}). In this part, we consider $F \in \mathcal{C}^1(\mathbb{R}_+, \mathbb{R}_+)$ and $f \in \mathcal{C}^1_b(\mathcal{E}, \mathbb{R}_+)$.
				
			\begin{theorem}[Infinitesimal generator of $(Z_t)$]\label{MultiCouche_Gene_Theo} Under hypotheses \ref{Hypothesis_Probability} and \ref{Hypothesis_DivisionRate}, the process $Z$ defined in (\ref{Z_t_Equation}) and starting from $Z_0$ is a Markovian process in the Skhorod space $\mathbb{D}([0,T],  \mathcal{M}_P(\llbracket 1, J\rrbracket  \times \mathbb{R}_+)) $. Let $T >0$, $Z$ satisfies 
				\begin{equation}\label{SDE_N_tandZ_t}
					\mathbb{E} \big[ \sup_{t \leq T}  N_t \big] < \infty, \quad  \mathbb{E} \big[ \sup_{t \leq T}  \ll a, Z_t \gg  \big] < \infty, 
				\end{equation} 
			and its infinitesimal generator is 
		\begin{equation*}
			\footnotesize
			\begin{aligned}
				&	\mathcal{G}F\big[\ll f,Z\gg \, \big] = \ll F'[\ll Z,f\gg \, ]\partial_af , Z\gg \,   \\
				&	  + \displaystyle \sum_{j = 1}^{J}\int_{0}^{\infty} \big( F\big[ \ll f,2 \delta_{j,0} - \delta_{j,a} +Z\gg \,  \big] - F\big[\ll f,Z\gg \, \big]  \big) p^{(j)}_{2,0}b_j(a) Z(dj,da)  \\
				&	 + \displaystyle \sum_{j = 1}^{J}\int_{0}^{\infty} \big( F\big[   \ll f ,\delta_{j,0} + \delta_{j+1,0} - \delta_{j,a} +Z\gg \, \big] - F\big[\ll f,Z\gg \, \big]  \big) p^{(j)}_{1,1}b_j(a) Z(dj,da)\\
				&	+ \displaystyle \sum_{j = 1}^{J}\int_{0}^{\infty} \big(F\big[ \ll f ,2\delta_{j+1,0} - \delta_{j,a} +Z\gg \, \big] - F\big[\ll f,Z\gg \, \big] \big) p^{(j)}_{0,2}b_j(a) Z(dj,da) \, .
			\end{aligned}
			\end{equation*}
			\end{theorem}
			From this theorem, we derive the following Dynkin formula : 
			\begin{lemma}[Dynkin formula]\label{Multi_Dynkin}
				Let $T > 0$. Under hypotheses \ref{Hypothesis_Probability} and \ref{Hypothesis_DivisionRate}, $\forall t \in [0,T]$, 
				\begin{equation*}
					\displaystyle F[\ll f,Z_t\gg \, ] =  F[\ll f,Z_0\gg \, ] + \int_{0}^{t} \mathcal{G}F[\ll f,Z_s\gg \, ]ds + M_t^{F,f}
				\end{equation*}
				where $M^{F,f} $ is a martingale. Moreover,
				\begin{equation}\label{PbmDeMartingale_Equ}
					\ll f,Z_t\gg \,  =  \, \ll f,Z_0\gg \,  + \int_{0}^{t} \ll \mathcal{L}^Df,Z_s\gg \, ds + M_t^{f}
				\end{equation}
				where $\mathcal{L}^D$ the dual operator in (\ref{ProblemeAdjoint}) and $M^{f} $ is a $\mathbf{L}^2-$martingale defined by
				\begin{equation}\label{Martingale_Ecriture}
				\footnotesize
				\begin{aligned}
				&M_t^{f} =  \int_{0}^{t} \ll B(\cdot)f(\cdot) - K(\cdot)^Tf(0),Z_s\gg \,  ds  \\
				&	 + \int \int_{[0,t] \times \mathcal{E}} \mathds{1}_{k \leq N_{s^-}} \ll f,2 \delta_{I^{(k)}_{s^-},0} -\delta_{I^{(k)}_{s^-},A^{(k)}_{s^-}}\gg \,   \mathds{1}_{0 \leq \theta \leq m_1(s,k,Z)}Q(ds,dk,d\theta) \\
				&	 + \int \int_{[0,t] \times \mathcal{E}} \mathds{1}_{k \leq N_{s^-}} \ll f,\delta_{I^{(k)}_{s^-},0} + \delta_{I^{(k)}_{s^-} + 1,0} -\delta_{I^{(k)}_{s^-},A^{(k)}_{s^-}}\gg \,   \mathds{1}_{m_1(s,k,Z)  \leq \theta \leq m_2(s,k,Z)}Q(ds,dk,d\theta)  \\
				&	 + \int \int_{[0,t] \times \mathcal{E}} \mathds{1}_{k \leq N_{s^-}}\ll f,2\delta_{I^{(k)}_{s^-} + 1,0}  -\delta_{I^{(k)}_{s^-},A^{(k)}_{s^-}}\gg \, \mathds{1}_{m_2(s,k,Z) \leq \theta \leq m_3(s,k,Z)}Q(ds,dk,d\theta)  
				\end{aligned}
				\phantom{\hspace{6cm}}
				\end{equation}
			and 
			\begin{equation}\label{PbmDeMartingale_Crochet}
				\footnotesize
				\begin{aligned}
				\left \langle M^{f},M^{f}\right \rangle_t = & \displaystyle \int_{0}^{t} \big[ \sum_{j = 1}^{J} \int_{\mathbb{R}_+}[\ll f,2 \delta_{j,0} - \delta_{j,a}\gg \, ]^2 b_{j}(a)p^{(j)}_{2,0}Z_s(dj,da) \\
			& +\sum_{j = 1}^{J} \int_{\mathbb{R}_+} [ \ll f, \delta_{j,0} + \delta_{j+1,0} - \delta_{j,a}\gg \,  ]^2 b_{j}(a)p^{(j)}_{1,1}Z_s(dj,da) \\
			& + \sum_{j = 1}^{J} \int_{\mathbb{R}_+} [ \ll f,2\delta_{j+1,0} - \delta_{j,a}\gg \, ]^2   b_{j}(a)p^{(j)}_{0,2}Z_s(dj,da) \big]ds  \, .
					\end{aligned}
			\end{equation}
			\end{lemma}
			The proofs of theorem \ref{MultiCouche_Gene_Theo} and lemma \ref{Multi_Dynkin} are classical and provided in \ref{Supplemental_Stochastic} for reader convenience. 
			We now have all the elements to prove theorem \ref{Multi_Conver_L2}.
			\begin{proof}[Proof of theorem \ref{Multi_Conver_L2}]
				We apply the Dynkin formula (\ref{PbmDeMartingale_Equ}) with the dual test function $\dualfunction$ and obtain $	\displaystyle \ll \dualfunction,Z_t\gg \,  = \ll \dualfunction, Z_0\gg \,  + \Malthus \int_{0}^{t}\ll\dualfunction, Z_s\gg \, ds + M^\dualfunction_t. $
				As $\dualfunction$ is bounded, $\ll\dualfunction,Z_t\gg \,  $ has finite expectation for all time $t$ according to \eqref{SDE_N_tandZ_t}. Thus,
				\begin{equation}\label{SDE_eqA}
				\mathbb{E}\big[ \ll \dualfunction,Z_t\gg \,  \big] = \mathbb{E}\big[ \ll \dualfunction,Z_0\gg \,  \big] + \Malthus \mathbb{E}\big[ \int_{0}^{t}\ll \dualfunction, Z_s\gg \, ds \big].
				\end{equation}
				Using Fubini theorem and solving equation \eqref{SDE_eqA}, we obtain:
				\begin{equation*}
				 \mathbb{E}\big[ \ll\dualfunction,Z_t\gg \,  \big] = e^{\Malthus t}\mathbb{E}\big[ \ll\dualfunction,Z_0\gg \,  \big] \, \Rightarrow \, \mathbb{E}\big[ e^{-\Malthus t} \ll \dualfunction, Z_t \gg \big] = \mathbb{E}\big[ \ll\dualfunction,Z_0\gg \,  \big].
				\end{equation*}
				 Hence, $W^\dualfunction_t = e^{-\Malthus t} \ll \dualfunction, Z_t \gg $ is a martingale. According to martingale convergence theorems (see Theorem 7.11 in \cite{klebaner_introduction_2012}), $W^\dualfunction_t$ converges to an integrable random variable $ W^\dualfunction_\infty \geq 0$, $\mathbb{P}-$p.s. when $t$ goes to infinity. 
				 To prove that $W^\dualfunction_\infty$ is non-degenerated, we will show that the convergence holds in $\mathbf{L}^2$. Indeed, from the $\mathbf{L}^2$ and almost sure convergence, we deduce the $\mathbf{L}^1 $ convergence. Then, applying the dominated convergence theorem, we have: 
				 \vspace{-0.2cm}
				 \begin{equation*}
				 	\mathbb{E}[ W^\dualfunction_\infty] :=\mathbb{E}[ \lim\limits_{t \rightarrow \infty} W^\dualfunction_t]= 	\lim\limits_{t \rightarrow \infty} \mathbb{E}[W^\phi_t] = \mathbb{E}[ W^\dualfunction_0] > 0.
				 \end{equation*}
				 Consequently, $W^\dualfunction_\infty $ is non-degenerated. To show the $\mathbf{L}^2$ convergence, we compute the quadratic variation of $W^\dualfunction$. Applying Ito formula (see \cite{protter_stochastic_2004} p. 78-81) with $F(t,\ll\dualfunction,Z_t\gg \, ) = e^{-\Malthus t}\ll\dualfunction,Z_t\gg \, $, we deduce:
				\begin{equation*}
				\footnotesize
				\begin{aligned} 
				& \displaystyle	W^\dualfunction_t =  \ll\dualfunction,Z_0\gg \,  +  \int_{0}^{t}\big[ \int_{\mathcal{E}} e^{-\Malthus s}( \partial_a \dualfunction^{(j)}(a) -\Malthus \dualfunction^{(j)}(a))  Z_s(dj,da)\big] ds \\
				&		\displaystyle	 + \int \int_{[0,t] \times \mathcal{E}} \mathds{1}_{k \leq N_{s^-}} e^{-\Malthus s}\ll \phi,2 \delta_{I_{s^-}^{(k)},0} - \delta_{I_{s^-}^{(k)},A_{s^-}^{(k)}} \gg \,  \mathds{1}_{0 \leq \theta \leq m_1(s,k,Z)}Q(ds,dk,d\theta) \\
				&		\displaystyle	 + \int \int_{[0,t] \times \mathcal{E}} \mathds{1}_{k \leq N_{s^-}} e^{-\Malthus s}\ll \phi, \delta_{I_{s^-}^{(k)},0} + \delta_{I_{s^-}^{(k)} + 1,0} - \delta_{I_{s^-}^{(k)},A_{s^-}^{(k)}} \gg \,  \mathds{1}_{m_1(s,k,Z)  \leq \theta \leq m_2(s,k,Z)}Q(ds,dk,d\theta) \\
				&		\displaystyle	 + \int \int_{[0,t] \times \mathcal{E}} \mathds{1}_{k \leq N_{s^-}} e^{-\Malthus s} \ll \phi,2 \delta_{I_{s^-}^{(k)} + 1,0} - \delta_{I_{s^-}^{(k)},A_{s^-}^{(k)}} \gg \, \mathds{1}_{m_2(s,k,Z) \leq \theta \leq m_3(s,k,Z)}Q(ds,dk,d\theta) 	\, .
				\end{aligned}
				\phantom{\hspace{6cm}}
				\end{equation*}
				
				As $\mathcal{L}^D \dualfunction = \Malthus \dualfunction  $, we have 
				\begin{align*} \int_{\mathcal{E}} (\partial_a \dualfunction^{(j)}(a) -\Malthus \dualfunction^{(j)}(a))Z_s(dj,da) = \, \ll B(\cdot)\dualfunction(\cdot) - K^T(\cdot)\dualfunction(0),Z_s\gg \, . \end{align*}
				Consequently, from (\ref{Martingale_Ecriture}), we deduce 
				\begin{equation}\label{Martingale_forme_for_phi}
				\displaystyle	W^\dualfunction_t =  \ll\dualfunction,Z_0\gg \,  +  \int_{0}^{t}e^{-\Malthus s}dM^\dualfunction_s  \, .
				\end{equation}
				where $dM^\dualfunction_s$ is defined as $ \displaystyle M^\dualfunction_t = \int_{0}^{t}dM^\dualfunction_s $. According to \eqref{PbmDeMartingale_Crochet} and \eqref{Martingale_forme_for_phi}, we get
				\begin{multline*}
				\displaystyle \left \langle W^\dualfunction_\cdot,W^\dualfunction_\cdot\right \rangle_t = \int_{0}^{t} e^{- 2\Malthus s} d \left \langle M^{\dualfunction},M^{\dualfunction}\right \rangle_s ds \\
							 = 	\displaystyle \int_{0}^{t} e^{- 2\Malthus s} \left [  \int_{\mathcal{E}} \left (p^{(j)}_{2,0}[\ll\dualfunction,2 \delta_{j,0} - \delta_{j,a}\gg \, ]^2 + p^{(j)}_{1,1} [\ll\dualfunction, \delta_{j,0} + \delta_{j+1,0} - \delta_{j,a}\gg \,  ]^2  \right.  \right.\\
							\left.  \left.+ p^{(j)}_{0,2}[ \ll\dualfunction,2\delta_{j+1,0} - \delta_{j,a}\gg \, ]^2 \right ) b_j(a)Z_s(dj,da)\right]ds \, .
				\end{multline*}
				Since, $\dualfunction $ and $b$ are bounded, there exists a constant $K > 0$ such that  
				\begin{equation*}
						\displaystyle \left \langle W^\dualfunction,W^\dualfunction \right \rangle_t  \leq K \int_{0}^{t} e^{- 2 \Malthus s} \left [  \int_{\mathcal{E}}Z_s(dj,da)\right ]ds  \, .
				\end{equation*}
				Taking the expectation and using moment estimate \eqref{SDE_N_tandZ_t}, we get $	\mathbb{E}[\langle W^\dualfunction,W^\dualfunction\rangle_t] < \infty $. Thanks to the Burkholder-Davis-Gundy inequality (see Theorem 48, \cite{protter_stochastic_2004}), we deduce that $
					\mathbb{E} [ \sup_{t \leq T}  \left ( W^\dualfunction_t \right )^2] < \infty$,
				and thus the $\mathbf{L}^2$ convergence of $W^\dualfunction $.
			\end{proof}
			
			\subsection{Asymptotic study of the renewal equations}
				We now turn to the study of renewal equations associated with the branching process $Z$. Following \cite{harris_theory_1963} (Chap. VI), we introduce generating functions that determine the cell moments.
				In all this subsection, we consider $a \in \mathbb{R}_+ \cup \{+ \infty \}$.
				We recall that $Y^{(j,a)}_t= \, \langle Z_t,\mathds{1}_{j}\mathds{1}_{\leq a}  \rangle$ and $ Y_t^a =  (Y^{(j,a)}_t)_{j \in \llbracket 1, J \rrbracket}$. 
				For $\mathbf{s} = (s_{1}, ..., s_{J}) \in \mathbb{R}^{J}$ and $\mathbf{j} = (j_{1},...,j_{J}) \in \mathbb{N}^{J}$, we use classical vector notation $ \mathbf{s}^{\mathbf{j}} = \prod_{i=1}^{J}s_{i}^{j_{i}}$. 
				
				\begin{definition}\label{DefinitionFonctionGeneratrice}
					We define $ F^a[\mathbf{s};t] = (F^{(i,a)}[\mathbf{s};t])_{i \in \llbracket 1, J \rrbracket}$ where $F^{(i,a)} $ is the generating function associated with $Y^{a}_t $ starting with $Z_0 = \delta_{i,0} $:
					$$F^{(i,a)}[\mathbf{s};t] := \mathbb{E}[\mathbf{s}^{Y_t^a} \rvert Z_0 = \delta_{i,0}] \, .$$
				\end{definition}
				We obtain a system of renewal equations for $F$ and \\
				$M^a(t) := ( \mathbb{E}[Y^{(j,a)}_t|Z_0 = \delta_{i,0}])_{i,j \in \llbracket 1, J \rrbracket}$.
				
				\begin{lemma}[Renewal equations for $F$]\label{FoncGenForme2}For $i \in \llbracket 1, J \rrbracket$, $F^{(i,a)} $ satisfies:
					\begin{equation}\label{FonctionGeneratriceEquation}
						\forall i \in \llbracket 1,J \rrbracket, 
						\quad F^{(i,a)}[\mathbf{s};t]= (s_{i} \mathds{1}_{t \leq a} + \mathds{1}_{t > a} )(1- \mathcal{B}_{i}(t)) + f^{(i)}(F^a [\mathbf{s},.] )\ast d\mathcal{B}_{i}(t)
					\end{equation}
					where $f^{(i)}$ is given by $	f^{(i)}(\mathbf{s}) := p_{2,0}^{(i)}s_{i}^{2} + p_{1,1}^{(i)}s_{i}s_{i+1} + p_{0,2}^{(i)}s_{i+1}^{2}  $.
					\end{lemma}
				
				\begin{lemma}[Renewal equations for $M$]\label{MomentDOrdre1}
				 	For $(i,j) \in \llbracket 1, J\rrbracket ^{2}$, $M^a_{i,j}$ satisfies:
					 \begin{equation}
						 \label{MomentDOrdre1Equ}
						 M^a_{i,j}(t) = \delta_{i,j}(1- \mathcal{B}_{i}(t)) \mathds{1}_{t \leq a}  + 2p^{(i)}_{S} M^a_{i,j}\ast d\mathcal{B}_{i}(t) + 2p^{(i)}_{L} M^a_{i+1,j}\ast d\mathcal{B}_{i}(t)  \, .
					 \end{equation}
				\end{lemma}
				The proofs of lemma \ref{FoncGenForme2} and \ref{MomentDOrdre1} are given in \ref{Supplemental_Stochastic}. 
			
				\begin{theorem}\label{MoyenneTempsLong_Interm}
					Under hypotheses \ref{Hypothesis_Probability}, \ref{Hypothesis_DivisionRate}, \ref{Malthus_parameter_def}, \ref{Hypothesis_MalthusPerLayerDistinct} and \ref{Hypothesis_lambdaj}, 
					\begin{equation}\label{M_mean_age}
					\forall i \in \llbracket 1, J \rrbracket, \quad \forall k \in \llbracket 0, J - i \rrbracket, \quad M^a_{i,i+k}(t) \sim \widetilde{M}_{i,i+k}(a) e^{\MalthusPerLayer_{i,i+k} t}, \quad t \rightarrow \infty
					\end{equation}
					where $\MalthusPerLayer_{i,i+k} = \underset{j \in \llbracket i, i+k \rrbracket}{\max} \MalthusPerLayer_j$, 
					\begin{equation}\label{M_tilde_ik_def_age_2}
					\widetilde{M}_{i,i}(a) = 
					\displaystyle	\frac{ \int_0^a (1 - \mathcal{B}_i (t))e^{-\MalthusPerLayer_{i}t}dt}{2p_S^{(i)} \int_0^\infty t d\mathcal{B}_i(t) e^{- \MalthusPerLayer_{i} t } dt } 
					\end{equation}
					and, for $k \in \llbracket 1, J-i \rrbracket $
					\begin{equation}\label{M_tilde_ik_def_age_1}
					\widetilde{M}_{i,i+k}(a) = 
					\left\{
					\begin{array}{ll}
						\displaystyle	\frac{2p_L^{(i)}d \mathcal{B}_i^*(\MalthusPerLayer_{i,i+k})}{1 -2p_S^{(i)}d \mathcal{B}_i^*(\MalthusPerLayer_{i,i+k}) }\widetilde{M}_{i+1,i+k}(a), & \text{ if } \MalthusPerLayer_{i,i+k} \neq  \MalthusPerLayer_i  \, (i) \\[0.5cm]
						\displaystyle	\frac{2p_L^{(i)}d \mathcal{B}_i^*(\MalthusPerLayer_{i})}{2p_S^{(i)} \int_0^\infty t d\mathcal{B}_i(t) e^{- \MalthusPerLayer_{i} t } dt } \int_{0}^{\infty}M^a_{i+1,i+k}(t)e^{-\MalthusPerLayer_{i}t}dt, & \!\! \text{ if } \MalthusPerLayer_{i,i+k} =  \MalthusPerLayer_i   (ii) .\\[0.5cm]
					\end{array}  
					\right. 
					\end{equation}
				\end{theorem}
				\begin{proof}
					Let the mother cell index $i \in \llbracket 1, J \rrbracket $.
					As no daughter cell can move upstream to its mother layer,  the mean number of cells on layer $j < i $ is null (for all $t \geq 0$ and for  $j <  i $, $M^a_{i,j}(t) = 0 $).
					 We consider the layers downstream the mother one ($j \geq i$) and proceed by recurrence: 
					\begin{equation*}
					\mathcal{H}^{k} : \quad \forall i \in \llbracket 1, J-k \rrbracket, \begin{array}{l}
					M^a_{i,i+k}(t) \sim \widetilde{M}_{i,i+k}(a) e^{\MalthusPerLayer_{i,i+k} t}, \text{  as  } t\rightarrow \infty \, .
					\end{array}
					\end{equation*}
					We first deal with $\mathcal{H}^{0} $. We consider the solution of \eqref{MomentDOrdre1Equ} for $j = i$:
					\begin{equation}
					\forall t \in \mathbb{R}_+, \quad	M^a_{i,i}(t) = \left (1- \mathcal{B}_{i}(t) \right )\mathds{1}_{t \leq a}   + 2p^{(i)}_{S} M^a_{i,i}\ast d\mathcal{B}_{i}(t) \, .
					\end{equation}
					We recognize a renewal equation as presented in \cite{harris_theory_1963}(p.161, eq.$(1)$) for $M_{i,i} $, which is similar to a single type age-dependent process. The main results on renewal equations are recalled in \ref{Supplemental_Branching}. Here, the mean number of children is $m= 2p^{(i)}_{S} > 0$ and the life time distribution is $\mathcal{B}_{i}$. 
					From hypothesis \ref{Hypothesis_DivisionRate}, we have
					\begin{equation*}
					\int_{0}^{\infty} \left ( 1- \mathcal{B}_{i}(t) \right )\mathds{1}_{t \leq a}    e^{-\MalthusPerLayer_i t}dt \leq \frac{1}{\bar{b}_i} \int_{0}^{\infty} \mathds{1}_{t \leq a}    d\mathcal{B}_{i}(t) e^{-\MalthusPerLayer_i t}dt  \leq    \frac{1}{\bar{b}_i} \int_{0}^{\infty}d\mathcal{B}_{i}(t) e^{-\MalthusPerLayer_i t}dt < \infty
					\end{equation*} 
					according to hypothesis \ref{Malthus_parameter_per_layer_def}. Thus, $ t\mapsto \mathds{1}_{t \leq a}  \left ( 1- \mathcal{B}_{i}(t) \right )e^{-\MalthusPerLayer_i t}$ is in $ \mathbf{L}^1(\mathbb{R}_+) $. Using hypotheses \ref{Malthus_parameter_per_layer_def} and \ref{Hypothesis_lambdaj}, we apply corollary \ref{Moment_Distribution_dB} and lemma \ref{HarrisLemma_2} (see lemma $2$ of \cite{harris_theory_1963},p.161) and obtain:
					\begin{equation*}
					\begin{array}{l}
					M^a_{i,i}(t) \sim \widetilde{M}_{i,i}(a) e^{\MalthusPerLayer_i t}, \text{  as  } t\rightarrow \infty,  \text{  where  } 
					\displaystyle \widetilde{M}_{i,i}(a) = \frac{  \int_{0}^{a}(1 - \mathcal{B}_{i}(t))  e^{-\MalthusPerLayer_i t }dt }{2p^{(i)}_{S}\int_0^\infty t d\mathcal{B}_i(t) e^{- \MalthusPerLayer_i t } dt } \, .
					\end{array}
					\end{equation*}
					Hence, $ \mathcal{H}^{0} $ is verified. We then suppose that $ \mathcal{H}^{k-1} $ is true for a given rank $k-1 \geq 0$ and consider the next rank $k$. According to \eqref{MomentDOrdre1Equ}, $M^a_{i, i+k} $ is a solution of the equation:
					\begin{equation}\label{M_ikplusone}
					M^a_{i,i+k}(t) = 2p^{(i)}_{S} M^a_{i,i+k}\ast d\mathcal{B}_{i}(t) + 2p^{(i)}_{L} M^a_{i+1,i+k} \ast d\mathcal{B}_{i}(t)  \, .
					\end{equation}
					We distinguish two cases :  $\MalthusPerLayer_{i,i+k} \neq  \MalthusPerLayer_i$ and $\MalthusPerLayer_{i,i+k} =  \MalthusPerLayer_i$. We first consider $\MalthusPerLayer_{i,i+k} =  \MalthusPerLayer_i$ and show that $ f(t) = M^a_{i+1,i+k} \ast d\mathcal{B}_{i}(t)e^{- \MalthusPerLayer_{i}  t} $ belongs to $\mathbf{L}^1(\mathbb{R}_+)$. Let $R > 0$. Using Fubini theorem, we deduce that:
					\begin{equation*}
					\int_{0}^{R}  f(t) dt
					=\int_{0}^{R}  \left [ \int_{u}^{R}e^{- \MalthusPerLayer_{i}(t-u)} M^a_{i+1,i+k} (t-u)dt \right ] e^{- \MalthusPerLayer_{i} u}d\mathcal{B}_{i}(u)du \, .
					\end{equation*}
					Applying a change of variable and using that $M^a_{i+1,i+k}(t) \geq 0 $ for all $t \geq 0$, we have:
					\begin{equation*}
					\int_{u}^{R}e^{- \MalthusPerLayer_{i} (t-u)} M^a_{i+1,i+k} (t-u)dt
					\leq  \int_{0}^{R}e^{-\MalthusPerLayer_{i}t} M^a_{i+1,i+k} (t)dt \, .
					\end{equation*}
					According to $\mathcal{H}^{k}$, we know that $ M^a_{i+1,i+k}(t) \sim \widetilde{M}_{i+1,i+k}(a) e^{\MalthusPerLayer_{i+1,i+k} t}$ as $t \rightarrow \infty$. Then,
					\begin{multline*}
					\int_{0}^{R}e^{- \MalthusPerLayer_{i}  t} M^a_{i+1,i+k} (t)dt = \int_{0}^{R}e^{- \MalthusPerLayer_{i+1,i+k}  t} M^a_{i+1,i+k} (t) e^{- (\MalthusPerLayer_{i} - \MalthusPerLayer_{i+1,i+k} ) t}dt \\
					\leq K \int_{0}^{R} e^{- (\MalthusPerLayer_{i} - \MalthusPerLayer_{i+1,i+k} ) t}dt \quad < \infty
					\end{multline*}
					when $R \rightarrow \infty$, as $ \MalthusPerLayer_i = \MalthusPerLayer_{i,i+k} > \MalthusPerLayer_{i+1,i+k}$. 
					Moreover, $\int_{0}^{R} e^{- \MalthusPerLayer_{i}  u}d\mathcal{B}_{i}(u)du \leq d\mathcal{B}_{i}^*(\MalthusPerLayer_i ) < \infty$ according to hypothesis \ref{Malthus_parameter_def}. Finally, we obtain an estimate for $ \int_{0}^{R}f(t) dt$ that does not depend on $R$. So, $f$  is integrable. We can apply lemma \ref{HarrisLemma_2} and deduce $M^a_{i,i+k}(t) \sim \widetilde{M}_{i,i+k}(a) e^{\MalthusPerLayer_{i,i+k} t}, \text{  as  } t\rightarrow \infty$, with $\widetilde{M}_{i,i+k}(a)$ given in \eqref{M_tilde_ik_def_age_1}(ii).\\
					We now consider the case $\MalthusPerLayer_{i,i+k} \neq  \MalthusPerLayer_i$ and introduce the following notations :
					\begin{equation*}
					\widehat{M}^a_{i ,i+k}(t) = M_{i,i+k}^a(t)e^{-\MalthusPerLayer_{ i ,i+k}t}, \quad \widehat{d\mathcal{B}_{i}}(t) = \frac{d\mathcal{B}_{i}(t)}{d\mathcal{B}^*_{i}(\MalthusPerLayer_{i,i+k})}e^{-\MalthusPerLayer_{i,i+k}t} \, .
					\end{equation*}
					In this case, $ \MalthusPerLayer_{i,i+k} > \MalthusPerLayer_i$, so that $ 2p^{(i)}_{S}d\mathcal{B}_{i}^{*}(\MalthusPerLayer_{i,i+k}) <  2p^{(i)}_{S}d\mathcal{B}_{i}^{*}(\MalthusPerLayer_i) = 1   $. We want to apply lemma \ref{HarrisLemma_4} (see lemma $4$ of \cite{harris_theory_1963}, p.163).
					We rescale \eqref{M_ikplusone} by $e^{-\MalthusPerLayer_{i,i+k}t}$ and obtain the following renewal equation for $ \widehat{M}^a_{i,i+1}$:
					\begin{equation*}
					\widehat{M}^a_{i,i+k}(t) =  2p^{(i)}_{S} d\mathcal{B}^*_{i}(\MalthusPerLayer_{i,i+k})\widehat{M}^a_{i,i+k}\ast \widehat{d\mathcal{B}_{i}}(t) + 2p^{(i)}_{L}  M^a_{i+1,i+k} \ast d\mathcal{B}_{i}(t) e^{-\MalthusPerLayer_{i,i+k}t} \, .
					\end{equation*}
					We compute the limit of $f(t)= M^a_{i+1,i+k} \ast d\mathcal{B}_{i}(t) e^{-\MalthusPerLayer_{i,i+k}t}$:
					\begin{equation*}
					f(t)  =  \int_{0}^{\infty} \mathds{1}_{[0,t]}(u) M^a_{i+1, i+k}(t-u)e^{-\MalthusPerLayer_{i,i+k}(t-u)}  e^{-\MalthusPerLayer_{i,i+k}u}  d\mathcal{B}_{i}(u)du  \, .
					\end{equation*}
					According to $\mathcal{H}^{k-1}$, $M^a_{i+1, i+k}(t) \sim e^{-\MalthusPerLayer_{i+1,i+k}t}\widetilde{M}_{i+1, i+k}(a)$. As $\MalthusPerLayer_{i,i+k} \neq \MalthusPerLayer_{i} $, we have $\MalthusPerLayer_{i,i+k} = \MalthusPerLayer_{i+1,i+k} $. Hence, $M^a_{i+1, i+k}(t)e^{-\MalthusPerLayer_{i,i+k}t} $ is dominated by a constant $K$ such that $  \int_{0}^{\infty} K e^{-\MalthusPerLayer_{i,i+k}u}d\mathcal{B}_{i}(u)du < \infty$. We apply the Lebesgue dominated convergence theorem and obtain $	\lim\limits_{t \rightarrow \infty}f(t) = \widetilde{M}_{i+1, i+k}(a)d\mathcal{B}_{i}^{*}(\MalthusPerLayer_{i,i+k})$.
					Applying lemma \ref{HarrisLemma_4}, we obtain that:
					\begin{equation*}
					\lim\limits_{t \rightarrow \infty} \widehat{M}^a_{i,i+k}(t) =  \frac{2p^{(i)}_{L}\widetilde{M}_{i+1, i+k}(a)d\mathcal{B}_{i}^{*}(\MalthusPerLayer_{i,i+k}) }{1 - 2p^{(i)}_{S} d\mathcal{B}^*_{i}(\MalthusPerLayer_{i,i+k})} = \widetilde{M}_{i,i+k}(a),
					\end{equation*}
					and the recurrence is proved.
				\end{proof}
				We have now all the elements to prove theorem \ref{MoyenneTempsLong}.
				
					\begin{proof}[Proof of theorem \ref{MoyenneTempsLong}]
						According to theorem \ref{MoyenneTempsLong_Interm}, we have: 
						\begin{equation}\label{MoyenneTempsLong_eq1}
						\forall j \in \llbracket 1, J \rrbracket, \quad m^a_j(t) \sim  \widetilde{M}_{1,j}(a) e^{ \MalthusPerLayer_{1,j}  t}, \quad \text{as } t \rightarrow \infty.
						\end{equation}
						When $j < \MalthusIndex$, we deduce directly from \eqref{MoyenneTempsLong_eq1} that $\widetilde{m}_j(a)= 0 $. We then consider the leading layer $j = \MalthusIndex$. For $k \in \llbracket 1, c-1 \rrbracket$, $\MalthusPerLayer_{k, \MalthusIndex} \neq \MalthusPerLayer_{k} $ so, $\widetilde{M}_{k,\MalthusIndex}(a)$ is related to $\widetilde{M}_{k+1,\MalthusIndex}(a)$ by \eqref{M_tilde_ik_def_age_1}(i).  Thus, we obtain:
						\begin{equation}\label{m_tilde_j_e1}
						\widetilde{m}_{\MalthusIndex}(a) = \prod_{m = 1 }^{\MalthusIndex - 1} \frac{2p_L^{(m)}d \mathcal{B}_m^*(\Malthus)}{1-2p_S^{(m)}(d \mathcal{B}_m^*)(\Malthus) } \widetilde{M}_{\MalthusIndex,\MalthusIndex}(a) \, .
						\end{equation}
						$\widetilde{M}_{\MalthusIndex,\MalthusIndex}(a)$ is given by \eqref{M_tilde_ik_def_age_2} and we deduce $\widetilde{m}_{\MalthusIndex}(a)$. We turn to the layers $j >\MalthusIndex$. For $k \in \llbracket 1, c-1 \rrbracket$, we have $\Malthus = \MalthusPerLayer_{k,j} \neq \MalthusPerLayer_{k} $. We obtain from \eqref{M_tilde_ik_def_age_1}(i)
					
						\begin{equation}\label{M_k_E1}
						\widetilde{m}_{j}(a) = \prod_{m = 1 }^{ \MalthusIndex
							 - 1} \frac{2p_L^{(m)}d \mathcal{B}_m^*(\Malthus)}{1-2p_S^{(m)}(d \mathcal{B}_m^*)(\Malthus) } 	\widetilde{M}_{\MalthusIndex,j}(a).
						\end{equation}
						Then, as $\Malthus = \MalthusPerLayer_{\MalthusIndex, j} $, we use \eqref{M_tilde_ik_def_age_1}(ii) and obtain:
						\begin{equation}\label{M_k_E2}
						\widetilde{M}_{\MalthusIndex,j }(a) = \frac{2p_L^{(\MalthusIndex)}d \mathcal{B}_\MalthusIndex^*(\Malthus)}{2p_S^{(\MalthusIndex)}\int_{0}^{\infty} t e^{- \Malthus t} d \mathcal{B}_{\MalthusIndex}(t)dt } \int_{0}^{\infty}M^a_{\MalthusIndex+1,j}(t)e^{-\Malthus t}dt.
						\end{equation}
						Then,  we apply the Laplace transform to \eqref{MomentDOrdre1Equ} for $\alpha = \Malthus$. Theorem \ref{MoyenneTempsLong_Interm} and the fact that $ \MalthusPerLayer_\MalthusIndex =  \MalthusPerLayer_{\MalthusIndex,j} $ guarantee that we can apply the Laplace transform to \eqref{MomentDOrdre1Equ} (see details in \ref{AdditDetails}).  We obtain:
						
						\begin{equation}\label{M_k_E3}
						\int_{0}^{\infty}M^a_{\MalthusIndex+1,j}(t)e^{-\Malthus t}dt =  \prod_{k=\MalthusIndex + 1}^{j-1} \frac{2p^{(k)}_{L} d\mathcal{B}_{k}^*(\Malthus)}{1 - 2p^{(k)}_{S}  d\mathcal{B}_{k}^*(\Malthus)} \times \frac{ \int_{0}^{a} \eigenfunction^{(j)}(s)ds }{ (1 - 2p^{(j)}_{S}  d\mathcal{B}_{j}^*(\Malthus)) \times \eigenfunction^{(j)}(0)}.
						\end{equation}
						Combining \eqref{M_k_E1}, \eqref{M_k_E2} and \eqref{M_k_E3} and the value of $\eigenfunction^{(j)}(0)$ given in \eqref{Eigenproblem_rho0_value}, we obtain $\widetilde{m}_j(a)$.
					\end{proof}
				
					We also study the asymptotic behavior of the second moment in \ref{Supplemental_Branching} (see theorem \ref{MomentDOrdre2Long}).
					\begin{remarks}
							These results can be extended in a case when the mother cell is not necessary of age $0$ (for the one layer case, see \cite{harris_theory_1963}, p.153). 
					\end{remarks}
					\begin{remarks}
					Using the same procedure as in theorem \ref{MoyenneTempsLong_Interm}, we can obtain a better estimate for the convergence of the deterministic solution $\rho$ than that in theorem \ref{Exponential_decay}. Indeed, we can consider the study of $h(t,x) = e^{-\MalthusPerLayer_{1,j} t}\densityfunction(t,x) - \eta\hat{\rho}_{1,j}(x)$ where $\hat{\rho}_{1,j} $ is the eigenvector of the sub-system composed of the $j$-th first layer, and find the proper function $\dualfunction_{1,j} $.
					\end{remarks}
			\subsection{ Numerical illustration}
				We perform a numerical illustration with age independent division rates (which satisfy hypothesis \ref{Hypothesis_DivisionRate}).
				Figure \ref{FormulasIllustration} illustrates the exponential growth of the number of cells, either for the original solution of the model \eqref{Z_t_Equation} (left panel) or the renormalized solution (right panel), checking the results given in theorems \ref{MoyenneTempsLong} and \ref{MomentDOrdre2Long}.
				Figure \ref{IllustrationMalthus} instantiates the effect of the parameters $b_1$ and $p_S^{(1)}$ on the leading layer (left panel) and the asymptotic proportion of cells (right panel). Note that the layer with the highest number of cells is not necessary the leading one. 
				As can be seen in Figure \ref{fig:AgeProfil}, the renormalized solutions of the SDE (\ref{Z_t_Equation}) and PDE \eqref{EDP_equation} match the stable age distribution $\eigenfunction$ (see theorems \ref{Eigenproblem} and \ref{MoyenneTempsLong}). Asymptotically, the age distribution decreases with age, which corresponds to a proliferating pool of young cells, and is consistent with the fact that $\eigenfunction^{(j)}$ is proportional to $e^{- \Malthus a} \mathbb{P}[\tau^{(j)} > a] $. The convergence speeds differ between layers (here, the leading layer is the first one and the stable state of each layer is reached sequentially), corroborating the inequality given in theorem \ref{Exponential_decay}. 
					 
				\begin{figure}[h]
					\centering
					\subfloat[Exponential growth and asymptotic behavior]{\includegraphics[height=4.cm]{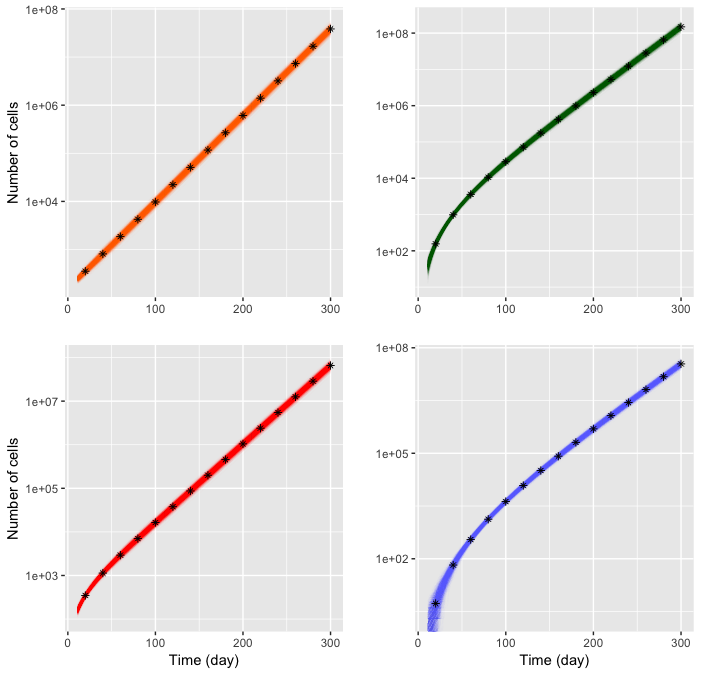}
						\includegraphics[height=4.cm]{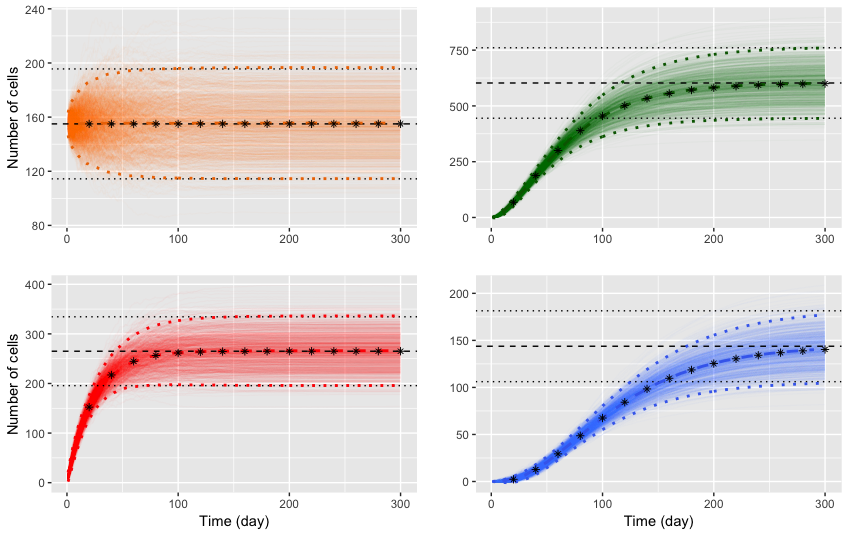} \label{FormulasIllustration}} \\
					\subfloat[Leading layer index and asymptotic proportion of cells ]{\includegraphics[height=2.5cm]{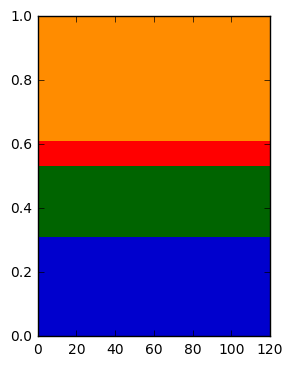}
						\includegraphics[height=2.5cm]{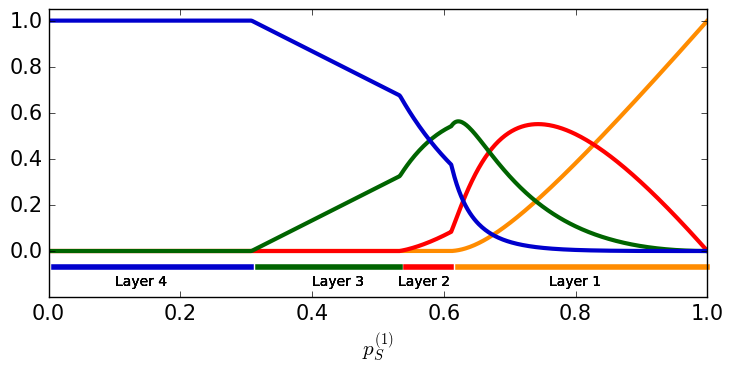}\label{IllustrationMalthus}}
					
					\caption{\textbf{Exponential growth and asymptotic moments}.  \textbf{Figure \ref{FormulasIllustration}}: Outputs of 1000 simulations of the SDE (\ref{Z_t_Equation}) according to the algorithm \ref{alg:SimulationStochasticProcess} with $p_S^{(j)}$, $b_j$ given in Figure \ref{Table_ParametersValues}, $p_{1,1}^{(j)} = 0$ and $Z_0 = 155 \delta_{1,0}$. \textbf{Left panel:} the solid color lines correspond to the outputs of the stochastic simulations while the black stars correspond to the numerical solutions of the ODE \eqref{ODE_equ} with the initial number of cells on the first layer $N = 155$ (orange: Layer 1, red: Layer 2, green: Layer 3, blue: Layer 4). \textbf{Right panel:} the color solid lines correspond to the renormalization of the outputs of the stochastic simulations by $e^{- \Malthus t}$. The black stars are the numerical solutions of the ODE \eqref{ODE_equ}. The color and black dashed lines correspond to the empirical means of the simulations and the analytical asymptotic means ($155\widetilde{m}_{j}(\infty)$, theorem \ref{MoyenneTempsLong}), respectively. The color and black dotted lines represent the empirical and analytical asymptotic $95 \% $ confidence intervals ($1.96\sqrt{v_j(\infty)}$, corollary \ref{Corollary_Vari}), respectively. \textbf{Figure \ref{IllustrationMalthus}}: Leading layer index as a function of $b_1$ and $p_S^{(1)}$ (left panel) and proportion of cells per layer in asymptotic regime with respect to $p_S^{(1)}$ (right panel). In both panels, $b$ satisfies (\ref{Rec_b}) and $p_S^{(j)} = -15*p_L^{(1)} * (j-1)^{2} - 110*p_L^{(1)} * (j-1)  + p_S^{(1)}$.}
					\label{fig:NumericalIllustration_MP}
				\end{figure}

						\begin{figure}[h]
							\centering
							\includegraphics[height=4.cm]{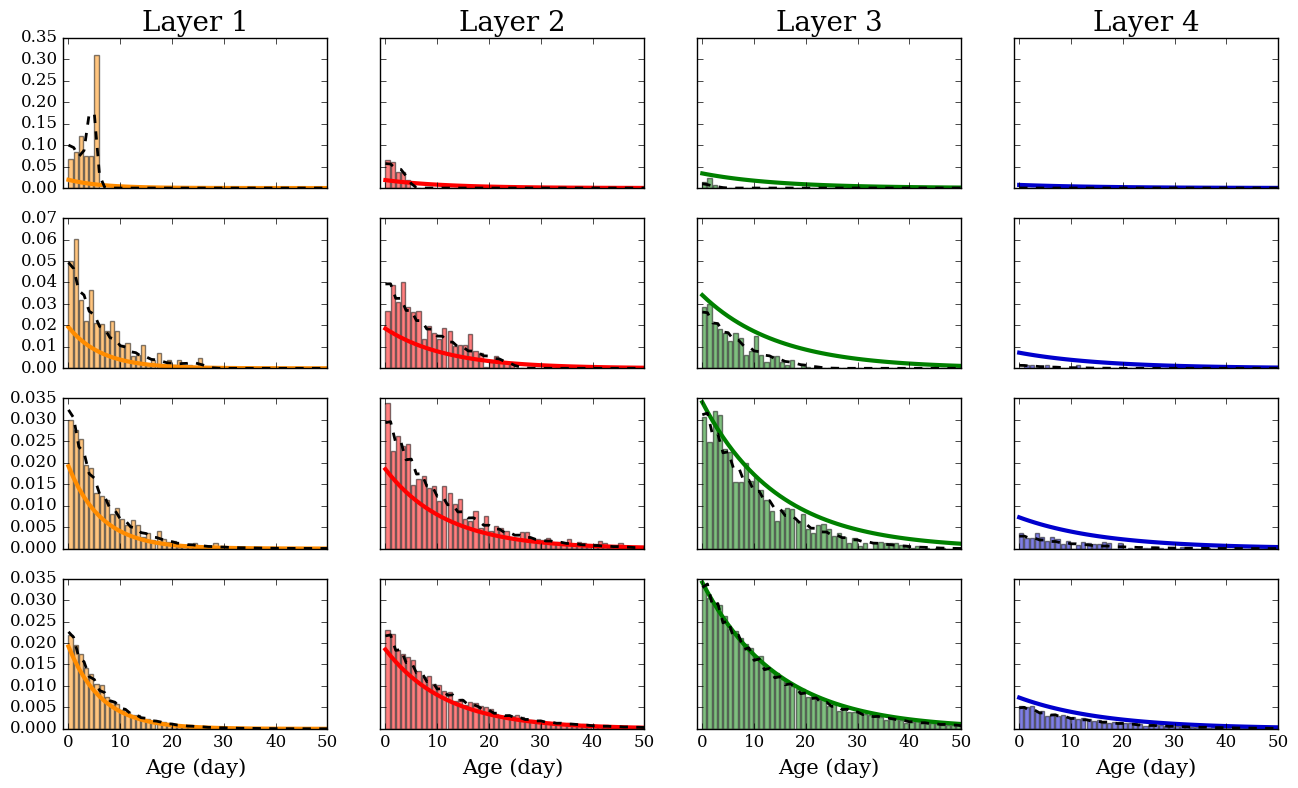}
							\caption{ \textbf{Stable age distribution per layer.} Age distribution at different times of one simulation of the SDE (\ref{Z_t_Equation}) and of the PDE \eqref{EDP_equation} using the algorithms described in respectively \ref{alg:SimulationStochasticProcess} and \ref{SimulationProtocole_PDE}. We use the same parameters as in Figure \ref{fig:NumericalIllustration_MP}. From top to bottom: t = 5, 25, 50 and 100 days. The color bars represent the normalized stochastic distributions. The black dashed lines correspond to the normalized PDE distributions, the color solid lines to the stable age distributions $\eigenfunction^{(j)} $, $j \in \llbracket 1, 4 \rrbracket $. The details of the normalization of each lines are provided in \ref{Supp_ProtocolForFigure4}.}
							\label{fig:AgeProfil}
						\end{figure}

\section{Parameter calibration}
Throughout this part, we will work under hypotheses \ref{Hypothesis_Probability}, \ref{Hypothesis_ageIndep} and \ref{Hypothesis_Start}. As a consequence, the intrinsic growth rate per layer can be computed easily: 
	\begin{equation}\label{Malthus_MP}
		\MalthusPerLayer_j = (2p_S^{(j)} - 1)b_j \in ]-b_j, b_j[, \text{when } j < J \, .
	\end{equation}

	\subsection{Structural identifiability}
		We prove here the structural identifiability of our system following \cite{perasso_identifiability_2016}.  We start by a technical lemma. 
		\begin{lemma} \label{Ident_Lemme_Ker}
		Let $M $ be the solution of \eqref{ODE_equ}. For any linear application $U: \mathbb{R}^J \rightarrow \mathbb{R}^J$, we have $	\left [ \forall t,	M(t) \in \ker(U)  \right ] \Rightarrow   \left[ U  = 0 \right]  $.
	\end{lemma}

	\begin{proof}
		\textit{Ad absurdum}, if $U \neq 0$ and $M(t) \in \ker(U)$, for all $t$, then there exists a non-zero vector $u := (u_1, ...,u_J)$ such that for all $t$, $ u^TM(t) = 0 $. This last relation, evaluated at $t= 0$ and thanks to the initial condition of \eqref{ODE_equ}, implies $u_1  = 0$. Then, derivating $M$, solution of \eqref{ODE_equ}, we obtain:
 		\begin{align*}
			\frac{d}{dt}\sum_{j = 2}^{J}u_jM^{(j)}(t) = 0 & \Rightarrow & 
			\sum_{j = 2}^{J}u_j [(b_{j-1} - \MalthusPerLayer_{j-1} )M^{(j-1)}(t)+ \MalthusPerLayer_j M^{(j)}(t)] = 0 \, .
		\end{align*}
		Again, at $t = 0$, we obtain $u_2(b_1 - \MalthusPerLayer_1 ) = 0$. Because $ \MalthusPerLayer_1 \neq b_1$, $u_2 = 0$. Iteratively, 
		\begin{align*}
		\forall j \in \llbracket 2,J \rrbracket, \quad	u_j \prod_{k = 1}^{j-1}(b_{k-1} - \lambda_{k-1} )= 0 \quad \Rightarrow u_j = 0 \, .
		\end{align*}
		We obtain a contradiction. 
	\end{proof}
	We can now prove theorem \ref{Iden_Them_All_Layer}. 
	\begin{proof}[Proof of theorem \ref{Iden_Them_All_Layer}]
		According to \cite{perasso_identifiability_2016}, the system \eqref{ODE_equ} is $\mathbf{P}$-identifiable if, for two sets of parameters $\mathbf{P}$ and  $\widetilde{\mathbf{P}}$, $M(t;\mathbf{P}) = M(t;\widetilde{\mathbf{P}})$ implies that $\mathbf{P} = \widetilde{\mathbf{P}}$. 
		\begin{align*}
			\forall t \geq 0 ,	M(t;\mathbf{P}) = M(t;\widetilde{\mathbf{P}}) & \Rightarrow & \frac{d}{dt}	M(t;\mathbf{P}) =  \frac{d}{dt}M(t;\widetilde{\mathbf{P}}) \\ 
			& \Rightarrow & A_\mathbf{P}M(t;\mathbf{P}) =  A_{\widetilde{\mathbf{P}}}M(t;\widetilde{\mathbf{P}}) = A_{\widetilde{\mathbf{P}}}M(t;\mathbf{P})  \\
			& \Rightarrow & (A_\mathbf{P}- A_{\widetilde{\mathbf{P}}})M(t;\mathbf{P}) = 0 
		\end{align*}
		So, $M(t;\mathbf{P}) \in \ker(A_\mathbf{P}- A_{\widetilde{\mathbf{P}}})$ and, from lemma \ref{Ident_Lemme_Ker}, we deduce that $A_\mathbf{P} = A_{\widetilde{\mathbf{P}}} $. Thus,
		\begin{equation*}
			 \left\{ \begin{array}{ll}
			(2p_S^{(j)} - 1)b_j= (2\widetilde{p}_S^{(j)} - 1)\widetilde{b}_j,& \forall j \in\llbracket 1, J \rrbracket, \\
			2p_L^{(j)}b_j = 2\widetilde{p}_L^{(j)}\widetilde{b}_j, &\forall j \in\llbracket 1, J-1 \rrbracket.
			\end{array}\right. 
		\end{equation*}
		Using that $ p_L^{(j)} = 1 - p_S^{(j)}$ and hypothesis \ref{Hypothesis_Probability}, we deduce $\mathbf{P} = \widetilde{\mathbf{P}} $.
	\end{proof}

	\subsection{Biological application} We now consider the application to the development of ovarian follicles.

	\subsubsection{Biological background}
		    The ovarian follicles are the basic anatomical and functional units of the ovaries. Structurally, an ovarian follicle is composed of a germ cell, named oocyte, surrounded by somatic cells (see Figure \ref{fig:follicle}).  In the first stages of their development, ovarian follicles grow in a compact way, due to the proliferation of somatic cells and their organization into successive concentric layers starting from one layer at growth initiation up to four layers.
		    \begin{figure}[h]
		    	\centering
		    	\includegraphics[height=2.cm]{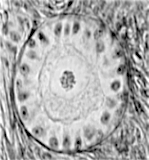}
		    	\includegraphics[height=2cm]{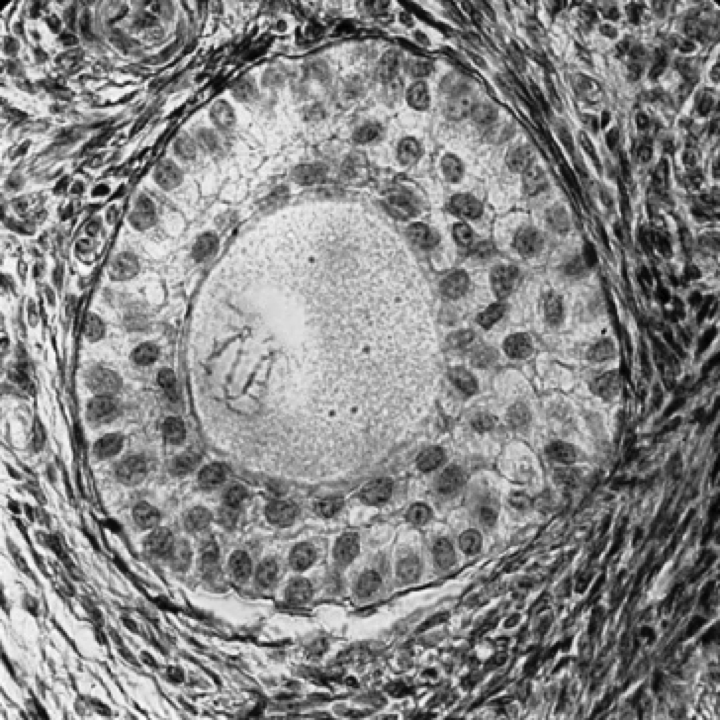}
		    	\includegraphics[height=2cm]{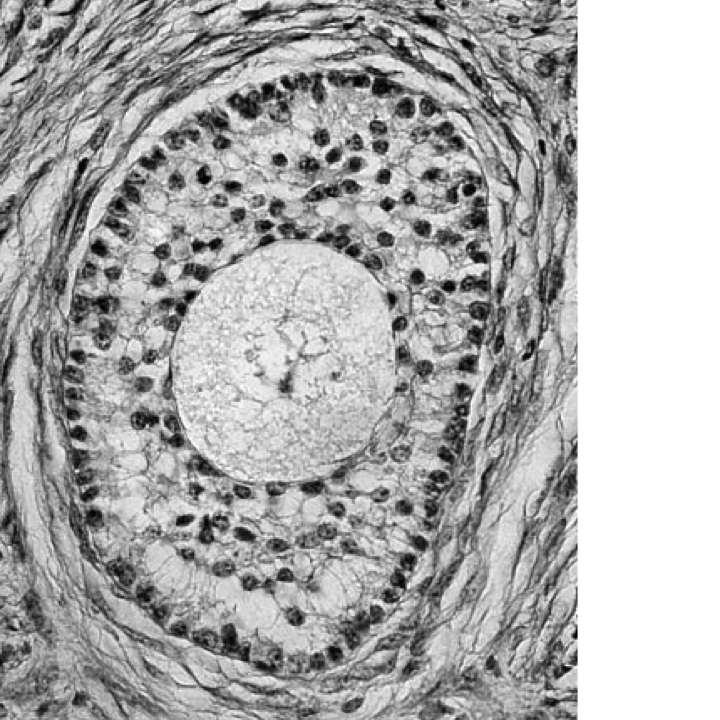}
		    	\caption{\textbf{Histological sections of ovarian follicles in the compact growth phase.} Left panel: one-layer follicle, center panel: three-layer follicle, right panel: four-layer follicle. Courtesy of Danielle Monniaux.}
		    	\label{fig:follicle}
		    \end{figure}	
	\subsubsection{Dataset description}
			We dispose of a dataset providing us with morphological information at different development stages (oocyte and follicle diameter, total number of cells), and acquired from \textit{ex vivo} measurements in sheep fetus \cite{lundy_populations_1999}. In addition, from  \cite{smith_effects_1993,smith_ontogeny_1994}, we can infer the transit times between these stages: it takes $15$ days to go from one to three layers and $10$ days from three to four layers. Hence (see Table \ref{Table_SummarizeOfDataset}), the dataset consists of the total numbers of somatic cells at three time points. 
		
			\begin{table}[h]
				\centering
				\subfloat[\textbf{Summary of the dataset}]{	\resizebox{0.55\textwidth}{!}{  	\begin{tabular}{|p{0.2\textwidth}||p{0.12\textwidth}|p{0.12\textwidth}|p{0.12\textwidth}|} 
							\hline 
							& $t = 0 $& $t = 20 $  & $t = 35 $  \\ 
							\hline 
							Data points	(62)& 34   &  10 & 18   \\ 
							\hline 
							Total cell number	& 113.89 $\pm$ 57.76  & 885.75 $\pm$ 380.89 & 2241.75 $\pm$ 786.26 \\
							\hline 
							Oocyte diameter ($\mu m$)& 49.31 $\pm$ 8.15 &75.94 $\pm$ 10.89  & 88.08 $\pm$ 7.43 \\ 
							\hline 
							Follicle diameter ($\mu m$)	& 71.68 $\pm$ 13.36 &141.59 $\pm$ 17.11  & 195.36 $\pm$ 23.95 \\ 
							\hline 
							\hline 
					\end{tabular}}
					\label{Table_SummarizeOfDataset}} 
					\subfloat[\textbf{Estimated values of the parameters.}]{	\resizebox{0.4\textwidth}{!}{  	
						\begin{tabular}{|p{0.1\textwidth}||p{0.08\textwidth}|p{0.08\textwidth}|p{0.09\textwidth}|}
							\hline 
							Layer $j$ &  $p_S^{(j)}$ & $b_j$ & $\MalthusPerLayer_j$ \\ 
							\hline 
							1 &  0.6806 & 0.1146 & \textcolor{orange}{0.0414}  \\ 
							\hline 
							2 &  0.4837 & \textcolor{blue}{0.0435} &  \textcolor{orange}{-0.0014}  \\ 
							\hline 
							3 &  0.9025 &\textcolor{blue}{0.0354 } &  \textcolor{orange}{0.0285}  \\ 
							\hline 
							4 &  1 &\textcolor{blue}{0.0324 }   &  \textcolor{orange}{0.0324 }\\ 
							\hline 
							\hline 
					\end{tabular}}
						\label{Table_ParametersValues}}
					\caption{\textbf{Experimental dataset and estimated values of the parameters.} Table \ref{Table_ParametersValues}. The estimated value of $\alpha$ and the initial number of cells are respectively $\alpha = 1.633$ and $N \approx 155 \, $. For $j \geq 2$, the $b_j$ parameter values (in blue) were computed using formula (\ref{Rec_b}).  The $\MalthusPerLayer_j$ values were computed using formula (\ref{Malthus_MP}). The $95\%$-confidence intervals are $b_1 \in [0.0760;0.1528] $, $\alpha \in [0.0231;5.685] $, $N \in [126.4;185.4]$, $p_S^{(1)} \in [0.6394;0.7643] $, $p_S^{(2)} \in [0; 0.7914[$ and $p_S^{(3)} \in [0.6675;0.9739]$. }
			\end{table}	
			\vspace{-0.5cm}
			We next take advantage of the spheroidal geometry and compact structure of ovarian follicles to obtain the number of somatic cells in each layer. Spherical cells are distributed around a spherical oocyte by filling identical width layers one after another, starting from the closest layer to the oocyte. Knowing the  oocyte and somatic cell diameter (respectively $d_O$ and $d_s$) and, the total number of cells $N^{exp}$, we compute the number of cells on the $j$th layer according to the ratio between its volume $V^j$ and the volume of a somatic cell $V^s$:  
			\begin{itemize}
				\item[] \textsc{Initialization}:   $j \leftarrow 1 , V^s \leftarrow  \frac{\pi d_s^3}{6}, N \leftarrow N^{exp}$
				\item[] While $ N > 0$ :
				\begin{itemize}
					\item[] $V^j \leftarrow  \frac{\pi}{6} \big[(d_O + 2*j* d_s)^3 - (d_O + 2*(j-1)* d_s)^3  \big]$ 
					\item[] $N_j  \leftarrow  \min (\frac{V^j}{V^s}, N), N \leftarrow N - N_j, j \leftarrow  j +1$
				\end{itemize}
				$J \leftarrow j - 1$
			\end{itemize}
		
			The corresponding dataset is shown on the four panels of Figure \ref{fig:ParametersEstim_Fit}.
	 
    \subsubsection{Parameter estimation} 
    Before performing parameter estimation, we take into account additional biological specifications on the division rates. The oocyte produces growth factors whose diffusion leads to a decreasing gradient of proliferating chemical signals along the concentric layers, which results to the recurrence law (\ref{Rec_b}) similar as that initially proposed in \cite{clement_coupled_2013}. 
    Considering a regression model with an additive gaussian noise, we estimate the model parameters to fit the changes in cell numbers in each layer (see \ref{SM_ParmeterEstimaProcedure} for details). The estimated parameters are provided in Table \ref{Table_ParametersValues} and the fitting curves are shown in Figure \ref{fig:ParametersEstim_Fit}. We compute the profil likelihood estimates \cite{raue_structural_2009} and observe that all parameters are practically identifiable except $p_S^{(2)}$ (Figure \ref{fig:pledatamn} ). In contrast, when we perform the same estimation procedure on the total cell numbers, most of the parameters are not practicality identifiable (dataset in Table \ref{Table_SummarizeOfDataset}, see detailed explanations in \ref{SM_ParmeterEstimaProcedure}).  
		
\section{Conclusion}
	In this work, we have analyzed a multi-type age-dependent model for cell populations subject to unidirectional motion, in both a stochastic and deterministic framework. Despite the non-applicability of either the Perron-Frobenius or Krein-Rutman theorem, we have taken advantage of the asymmetric transitions between different types to characterize long time behavior as an exponential Malthus growth, and obtain explicit analytical formulas for the asymptotic cell number moments and stable age distribution. We have illustrated our results numerically, and studied the influence of the parameters on the asymptotic proportion of cells, Malthus parameter and stable age distribution.
	We have applied our results to a morphodynamic process occurring during the development of ovarian follicles. The fitting of the model outputs to biological experimental data has enabled us to represent the compact phase of follicle growth. Thanks to the flexibility allowed by the expression of morphodynamic laws in the model, we intend to consider other non-compact growth stages. 
	
	
	\section{Acknowledgments} 
		We thank Ken McNatty for sharing for the experimental dataset and Danielle Monniaux for helpful discussions.

\section{Supplemental proofs}
		
		\subsection{Deterministic model} \label{Supplemental_Deterministic}
					\begin{proof}[Proof of corollary \ref{Moment_Distribution_dB}]
						According to hypothesis \ref{Hypothesis_lambdaj}, 
						\begin{equation}\label{EQ1}
							\exists A > 0, \, \epsilon > 0 \text{ such that } \forall a \geq A, \, b_j(a) + \MalthusPerLayer_j > \epsilon.
						\end{equation}
						Let $k \in \mathbb{N}$. Using hypothesis \ref{Hypothesis_DivisionRate}, for all $t \geq A$, we have:
						\begin{equation*}
							0 \leq t^k b_j(t)  e^{- \displaystyle \int_{0}^{t} \left [ b_j(s) + \MalthusPerLayer_j \right ] ds} \leq  \overline{b}_j t^k e^{- \displaystyle \int_{0}^{A} \left [ b_j(s) + \MalthusPerLayer_j \right ] ds} e^{- \displaystyle \int_{A}^{t} \left [ b_j(s) + \MalthusPerLayer_j \right ] ds} \, .
						\end{equation*}
						Then, using \eqref{EQ1} we obtain:
						\begin{equation*}
							0 \leq t^k b_j(t)  e^{- \displaystyle \int_{0}^{t} \left [ b_j(s) + \MalthusPerLayer_j \right ] ds} \leq  t^k K_{A,\epsilon} e^{- \epsilon t},  
						\end{equation*}
						where $K_{A,\epsilon}$ is a constant given by $K_{A,\epsilon}:= \overline{b}_j  e^{- \displaystyle \int_{0}^{A} \left [ b_j(s) + \MalthusPerLayer_j  - \epsilon \right ] ds} $. As $\epsilon > 0$, 
						the function $t \mapsto t^k e^{- \epsilon t} $ is integrable on $\mathbb{R}_+$, and we deduce that $\int_{A}^{+ \infty}t^ke^{-\MalthusPerLayer_j t}d\mathcal{B}_j(t)dt < \infty$. Using the continuity of $b_j$ (hypothesis \ref{Hypothesis_DivisionRate}), we conclude that $t \mapsto e^{-\MalthusPerLayer_j t}d\mathcal{B}_j(t)dt $ is integrable on $\mathbb{R}_+$.
					\end{proof} 
				
					\begin{proof}[Proof of corollary \ref{Bornesd}]
					According to \eqref{Eigenproblem_phi}, we obtain: 
					\begin{equation}
						\displaystyle	 \forall j \in \llbracket 1, \Malthus \rrbracket, \,	\frac{\dualfunction^{(j)}(a)}{2 [p_S^{(j)}\dualfunction^{(j)}(0) + p_L^{(j)}\dualfunction^{(j+1)}(0) ]} =  \displaystyle \int_{a}^{+ \infty}b_{j}(s)e^{-\int_{a}^{s}\Malthus + b_{j}(u)du}ds   \, .
					\end{equation}
					According to remark \ref{remarks_dB*} and  hypothesis \ref{Malthus_parameter_def}, we deduce that $\Malthus > -  \overline{b}_j $, $\forall j \in \llbracket 1, J \rrbracket$. Hence, using also hypothesis \ref{Hypothesis_DivisionRate}, we have:
					\begin{equation*}
					\frac{\dualfunction^{(j)}(a)}{2 [p_S^{(j)}\dualfunction^{(j)}(0) + p_L^{(j)}\dualfunction^{(j+1)}(0) ]}	\geq	 \underline{b}_j \int_{a}^{+ \infty}e^{-(\Malthus + \overline{b}_j)(s-a)}ds =  \frac{\underline{b}_j}{\Malthus + \overline{b}_j}, 
					\end{equation*}
					and reminding that $\Malthus > 0$ (see remark \ref{remarks_dB*}), we also obtain the right-side of \eqref{Coro_ineq}: 
					\begin{align*}
					\displaystyle	\frac{\dualfunction^{(j)}(a)}{2 [p_S^{(j)}\dualfunction^{(j)}(0) + p_L^{(j)}\dualfunction^{(j+1)}(0) ]}
						\leq	\quad \int_{a}^{+ \infty}b_{j}(s)e^{-\int_{a}^{s}b_{j}(u)du}ds   \\ = [-e^{-\int_{a}^{s}b_j(u)du}]_a^{+ \infty} = 1  \, .
					\end{align*}
				\end{proof}
				
						\begin{proof}[Proof of lemma \ref{BornedHat}]
					For $j \in \llbracket 1, J \rrbracket$, any solution of \eqref{phihat_equ} in $\mathbf{L}^1(\mathbb{R}_+)$ is given by:
					\begin{equation*}
					\dualfunctionperlayer^{(j)}(a) =  \dualfunctionperlayer^{(j)}(0)e^{\int_{0}^{a} [\MalthusPerLayer_j + b_j(s)]ds}[1 -2p_{S}^{(j)}\int_{0}^{a}b_{j}(s)e^{-\int_{0}^{s} [\MalthusPerLayer_j + b_j(u)]du}ds] \, .
					\end{equation*}
					According to hypothesis \ref{Malthus_parameter_per_layer_def}, $1 = 2p_{S}^{(j)}\int_{0}^{+\infty}b_{j}(s)e^{-\int_{0}^{s}[\MalthusPerLayer_j + b_j(u)]du}ds$, thus
					\begin{equation*}
					\dualfunctionperlayer^{(j)}(a) =  2p_S^{(j)} \dualfunctionperlayer^{(j)}(0)\int_{a}^{+\infty}b_{j}(s)e^{-\int_{a}^{s} [\MalthusPerLayer_j + b_j(u)]du}ds \, .
					\end{equation*}	
					Finally, according to remark \ref{remarks_dB*}, $\MalthusPerLayer_j > -\overline{b}_j$ and we obtain, using hypothesis \ref{Hypothesis_DivisionRate},
					\begin{multline*}
					\frac{\dualfunctionperlayer^{(j)}(a)}{\dualfunctionperlayer^{(j)}(0)} =  2p_S^{(j)}\int_{a}^{+\infty}b_{j}(s)e^{-\int_{a}^{s} [\MalthusPerLayer_j + b_j(u)]du}ds \\
					\geq \quad   2p_S^{(j)}\underline{b}_j\int_{a}^{+\infty}e^{-(\MalthusPerLayer_j + \overline{b}_j)(s-a)}ds =  2p_S^{(j)}\frac{\underline{b}_j}{\MalthusPerLayer_j + \overline{b}_j}  \, .
					\end{multline*}
					Then, we want to show that $\dualfunctionperlayer^{(j)}(a) < \infty $ for all $a \in \mathbb{R}_+ \cup \{ \infty\}$. Let 
					\begin{equation*}
					I(a):=	\int_{a}^{+\infty}b_{j}(s)e^{-\int_{0}^{s} [\MalthusPerLayer_j + b_j(u)]du}ds  \, .
					\end{equation*} 
					Applying an integration by part to $I(a)$, we obtain that, for all $a \geq 0$,
					\begin{equation*}
					\displaystyle	I(a) = \displaystyle \left [e^{-\int_{0}^{s} [\MalthusPerLayer_j + b_j(u)]du} \right]_a^\infty - \MalthusPerLayer_j \int_{a}^{\infty} e^{-\int_{0}^{s} [\MalthusPerLayer_j + b_j(u)]du}ds \, .
					\end{equation*}
					Hypotheses \ref{Malthus_parameter_per_layer_def} and \ref{Hypothesis_DivisionRate} imply that, for all $a \geq 0$, $ \int_{a}^{\infty}e^{-\int_{0}^{s} [\MalthusPerLayer_j + b_j(s)]ds} < \infty $ and so, \\ $ \lim\limits_{s \rightarrow 0} e^{-\int_{0}^{s} [\MalthusPerLayer_j + b_j(u)]du} = 0$. Thus, we have:
					\begin{equation}\label{I_int_value}
						I(a) = e^{- \int_{0}^{a}  [\MalthusPerLayer_j + b_j(u)]du } - \MalthusPerLayer_j \int_{a}^{\infty} e^{-\int_{0}^{s} [\MalthusPerLayer_j + b_j(u)]du}ds \,.
					\end{equation}
					Multiplying \eqref{I_int_value} by $e^{\int_{0}^{a}  [\MalthusPerLayer_j + b_j(u)]du } $, we deduce:
					\begin{equation}\label{phi_hat_val}
						\frac{\dualfunctionperlayer^{(j)}(a)}{2p_S^{(j)}\dualfunctionperlayer^{(j)}(0)} = 1 -   \MalthusPerLayer_j \int_{a}^{\infty} e^{-\int_{a}^{s} [\MalthusPerLayer_j + b_j(u)]du}ds \, .
					\end{equation}
					If $\MalthusPerLayer_j \geq  0 $, we deduce directly from \eqref{phi_hat_val} that, for all $a \in \mathbb{R}_+ \cup \{ \infty\}$, $\frac{\dualfunctionperlayer^{(j)}(a)}{2p_S^{(j)}\dualfunctionperlayer^{(j)}(0)} \leq 1$. We assume that $\MalthusPerLayer_j < 0$. Using hypothesis \ref{Hypothesis_lambdaj}, we deduce that there exists constants $A > 0$ and $\epsilon> 0$ such that 
					\begin{equation*}
						\forall a \geq A, \quad \MalthusPerLayer_j + b_j(a) > \epsilon > 0.
					\end{equation*}
					Hence, with $C = \frac{-\MalthusPerLayer_j }{\epsilon} > 0$, we have:
					\begin{equation*}
					\forall a \geq A, \quad  	-\MalthusPerLayer_j  \leq  C (\MalthusPerLayer_j + b_j(a)) \, .
					\end{equation*}
					Applying this inequality to \eqref{phi_hat_val}, we obtain: 
					\begin{equation*}
					\forall a \geq A, \quad 	\frac{\dualfunctionperlayer^{(j)}(a)}{2p_S^{(j)}\dualfunctionperlayer^{(j)}(0)} \leq 1 + C \int_{a}^{\infty}[ \MalthusPerLayer_j+ b_j(s)] e^{-\int_{0}^{s} [\MalthusPerLayer_j+ b_j(u)]du}ds \times e^{\int_{0}^{a} [\MalthusPerLayer_j + b_j(s)]ds} \, .
					\end{equation*}
					Again, using hypotheses \ref{Malthus_parameter_per_layer_def} and \ref{Hypothesis_DivisionRate}, we obtain:
					\begin{equation*}
						\int_{a}^{\infty}[ \MalthusPerLayer_j+ b_j(s)] e^{-\int_{0}^{s} [ \MalthusPerLayer_j+ b_j(u)]du}ds = \left  [-e^{-\int_{0}^{s} [ \MalthusPerLayer_j+ b_j(u)]du} \right ]_{a}^{\infty} = e^{-\int_{0}^{a} [ \MalthusPerLayer_j+ b_j(u)]du} \, .
					\end{equation*}
					We deduce 
					\begin{equation*}
					\forall a \geq A, \quad 	\frac{\dualfunctionperlayer^{(j)}(a)}{2p_S^{(j)}\dualfunctionperlayer^{(j)}(0)} \leq 1 + C \, .
					\end{equation*}
					As $\dualfunctionperlayer^{(j)}$ is continuous, we conclude that
						\begin{equation*}
						\forall a \in \mathbb{R}_+ \cup \{ + \infty
						 \}, \quad 	\frac{\dualfunctionperlayer^{(j)}(a)}{2p_S^{(j)}\dualfunctionperlayer^{(j)}(0)} < \infty \, .
					\end{equation*}
					\vspace{0.1cm}
				\end{proof}

					\begin{proof}[Proof of lemma \ref{PrincipeDeConservation}]
					Deriving $ \ll e^{-\Malthus t} \densityfunction(t,\cdot), \dualfunction \gg$ with respect to $t$, we obtain
					\begin{equation*}
					\frac{d}{dt} \ll e^{-\Malthus t} \densityfunction(t,\cdot), \dualfunction \gg  \,  = - e^{-\Malthus t}  \ll (\Malthus \mathds{1} + \mathcal{B} + \partial_a)  \densityfunction(t,\cdot), \dualfunction \gg  \, .
					\end{equation*}
					By integration by part and using that $\densityfunction \in \mathbf{L}^1(\mathbb{R}_+)^J $ and $\dualfunction \in \mathcal{C}^1_b(\mathbb{R}_+)^J $, we have
					\begin{equation*}
					\ll  \partial_{a}   \densityfunction(t,\cdot), \dualfunction\gg  =  -   \densityfunction(t,0)^T\dualfunction(0) \, -  \ll  \densityfunction(t,\cdot), \partial_a \dualfunction\gg \, , 
					\end{equation*}
					and we deduce 
					\begin{multline*}
					\frac{d}{dt} \ll e^{-\Malthus t} \densityfunction(t,\cdot), \dualfunction \gg \, = e^{-\Malthus t} \left[ \densityfunction(t,0)^T\dualfunction(0) +  \ll  \densityfunction(t,\cdot), \partial_a \dualfunction\gg 	   \right. \\
			     	\left.  -  \ll (\Malthus \mathds{1} + \mathcal{B}) \densityfunction(t,\cdot), \dualfunction\gg \right] \, .
					\end{multline*}
					As we have $(\Malthus \mathds{1} + \mathcal{B})^T = (\Malthus \mathds{1} + \mathcal{B})$, it comes $  \ll  \densityfunction(t,\cdot), \partial_a \dualfunction\gg   - \ll (\Malthus \mathds{1} + \mathcal{B}) \densityfunction(t,\cdot), \dualfunction\gg   \,= \,  \ll  \densityfunction(t,\cdot), \partial_a \dualfunction- (\Malthus \mathds{1} + \mathcal{B}) \dualfunction\gg$. Then, using that $\mathcal{L}^D \dualfunction = \Malthus \dualfunction $, we deduce $ ( \partial_a  - \Malthus \mathds{1} - \mathcal{B} )\dualfunction  = -K(\cdot)^T\dualfunction(0)$. Thus, 
					\begin{equation*}
					\frac{d}{dt} \ll e^{-\Malthus t} \densityfunction(t,\cdot), \dualfunction \gg \, = e^{-\Malthus t} \left[  \densityfunction(t,0)^T\dualfunction(0) -  \ll \densityfunction(t,\cdot), K(\cdot)^T\dualfunction(0)\gg \right]  \, .
					\end{equation*}
					Note that $\ll \densityfunction(t,\cdot), K(\cdot)^T\dualfunction(0)\gg \, = \, \ll K(\cdot)\densityfunction(t,\cdot), \dualfunction(0)\gg  \, = \densityfunction(t,0)^T\dualfunction(0) $. Consequently, 
					$$	\frac{d}{dt} \ll e^{-\Malthus t} \densityfunction(t,\cdot), \dualfunction \gg  \, = 0 \, .$$ 
					Hence, 
					\begin{equation*}
					\forall t, \quad \ll e^{-\Malthus t} \densityfunction(t,\cdot), \dualfunction\gg  \quad = \quad \ll \densityfunction_0(\cdot), \dualfunction\gg  \quad =  \eta \, .
					\end{equation*}
					Thanks to the renormalization $\ll  \eigenfunction,\dualfunction\gg  = 1 $, we obtain the conservation principle:
					
					\begin{equation*}
					\ll  e^{-\Malthus t} \densityfunction(t,\cdot)- \eta \eigenfunction,\dualfunction\gg  =  \ll  e^{-\Malthus t} \densityfunction(t,\cdot),\dualfunction\gg - \eta \ll  \eigenfunction,\dualfunction\gg  = 0 \, .
					\end{equation*}
					\vspace{0.15cm}
				\end{proof}

				\begin{proof}[Proof of lemma \ref{ValeurAbsolue}]
					From the linearity of the system, it can be easily shown that $h $ is solution of
					\begin{equation*}
					\left\{
					\begin{array}{l}
					\partial_{t}h(t,a) + \partial_{a}h(t,a) + [\Malthus + B(a)] h(t,a) = 0, \quad t \geq 0, \quad a \geq 0, \\
					h(t,a = 0) = \int_0^\infty K(a)h(t,a)da. 
					\end{array} 
					\right.
					\end{equation*}
					Let f be a derivable function. Applying the chain rules, it comes, for $j \in \llbracket 1, J \rrbracket$,
					\begin{equation*}
						\begin{array}{rl}
								\partial_{t}f[h^{(j)}(t,a)] + \partial_{a}f[h^{(j)}(t,a)] = &  f'(h^{(j)}(t,a))[\partial_{t}h^{(j)}(t,a) + \partial_{a}h^{(j)}(t,a)] \\
							= & - [\Malthus + b_j(a)] \times h^{(j)}(t,a)f'(h^{(j)}(t,a))  \, .
						\end{array}
					\end{equation*}
					For $f(x) = |x|$, $f'(x) = \frac{|x|}{x}$, we deduce
					
					\begin{equation*}
						 \partial_{t}|h^{(j)}(t,a)| + \partial_{a}|h^{(j)}(t,a)| = - [\Malthus + b_j(a)] |h^{(j)}(t,a)| \, .
					\end{equation*}
					\vspace{0.15cm}
				\end{proof}

			\begin{lemma}\label{LemmeGronwallForUs}[Modified Gr\"{o}nwall lemma]
					Let $N \in \mathbb{N}^*$. Suppose that $\forall i \in \llbracket 1, N \rrbracket $, there exist $\kappa_i \in \mathbb{R}_+^* $, $\gamma \in \mathbb{R}_+^* $ and  $P_i$ polynomials of degree $\alpha_i \in \mathbb{N}$ such that 
				\begin{equation*}
				\displaystyle	F'(t) \leq \sum_{i =1}^{N} P_i(t) e^{-\kappa_i t} - \gamma F(t) .
				\end{equation*}
			   Then,
				\begin{equation*}
					F(t) \leq Ke^{-\gamma  t} +  \sum_{i =1}^{N} \widetilde{P}_i(t)e^{-\kappa_i t},
				\end{equation*}
			where $K$ is a constant and for all $i \in \llbracket 1, N \rrbracket $, $\widetilde{P}_i$ is a polynomial of degree $\widetilde{\alpha}_i \leq \alpha_i  + 1$.
			\end{lemma}

			\begin{proof}
				Note that $\frac{d}{dt}(e^{\gamma t}F(t)) =  (F'(t) + \gamma F(t)  ) \times e^{\gamma t}$. Hence, 
			\begin{equation*}
					\displaystyle	\frac{d}{dt}(e^{\gamma t}F(t)) \leq  \sum_{i =1}^{N}P_i(t) e^{ (-\kappa_i +\gamma) t} .
			\end{equation*}
		Then, integrating on the interval $[0,t]$, we obtain:
			\begin{align*}
					e^{\gamma t}F(t) - F(0) \leq \sum_{i =1}^{N} \widetilde{P}_i(t) e^{ (\gamma - \kappa_i) t} + K \\
				 F(t)  \leq (F(0) + K)e^{ -\gamma t} +  \sum_{i =1}^{N} \widetilde{P}_i(t) e^{ -\kappa_i t},
			\end{align*} 
			where $K$ is a constant and for all $i \in \llbracket 1, N \rrbracket $, $\widetilde{P}_i$ a polynomial of degree $\widetilde{\alpha}_i \leq \alpha_i  + 1$ (the degree increases when $\gamma = \kappa_i$).
			\end{proof}
		
		\subsection{Stochastic model} \label{Supplemental_Stochastic}
				For any $ f:(t,a) \mapsto( f^{(j)}_t(a))_{j \in \llbracket 1, J \rrbracket }\in \mathcal{B}_b^1(\mathbb{R}_+ \times \mathbb{R}_+, \mathbb{R}) ^{J}$ (the space product of the set of bounded functions with bounded derivatives), we note $\partial_1$ and $\partial_2$ respectively its derivative with respect to time ($t$) and age ($a$).
				\begin{lemma}\label{MultiCouche_lemme}
				Let $F \in \mathcal{C}^1(\mathbb{R}, \mathbb{R})$, $ f \in \mathcal{B}_b^1(\mathbb{R}_+ \times \mathbb{R}_+, \mathbb{R}) ^{J}$.
				\begin{equation*}
				\begin{array}{r}
				 F[\ll Z_t,f_t\gg ]   =  \displaystyle F[\ll f_0,Z_0\gg ] + \int_{0}^{t} \ll \partial_1 f_s+ \partial_2  f_s,Z_s\gg F'[\ll f_s,Z_s\gg ]ds \\[0.3cm]
								   + \displaystyle  \int_{[0,t] \times \mathcal{E}} \left [  \mathds{1}_{k \leq N_{s^-}} \left ( F[\ll f_s ,2\delta_{I^{(k)}_{s^-},0} - \delta_{I^{(k)}_{s^-},A^{(k)}_{s^-}} +Z_{s^-}\gg ] \right. \right. \\[0.3cm] 
								 \left.  - F[\ll f_s,Z_{s^-}\gg ]  \right ) \mathds{1}_{0 \leq \theta \leq m_1(s,k,Z)}  \\[0.4cm]
								 + \displaystyle \left ( F[ \ll f_s  , \delta_{I^{(k)}_{s-}+1,0} + \delta_{I^{(k)}_s,0} - \delta_{I^{(k)}_{s-},A^{(k)}_{s-}} +Z_{s^-}\gg ]  \right. \\[0.3cm]
				 \left.  - F[\ll f_s,Z_{s^-}\gg ]  \right ) \mathds{1}_{m_1(s,k,Z)  \leq \theta \leq m_2(s,k,Z)} \\[0.4cm]
								 + \left ( F[\ll f_s ,2\delta_{I^{(k)}_{s^-}+1,0} - \delta_{I^{(k)}_{s-},A^{(k)}_{s^-}} +Z_{s^-}\gg ] \right. \\[0.3cm]
							\left.	\left.  - F[\ll f_s,Z_{s^-}\gg ] \right )\mathds{1}_{m_2(s,k,Z)  \leq \theta \leq m_3(s,k,Z)} \right ]Q(ds,dk,d\theta).
				\end{array}
				\end{equation*} 
			\end{lemma}
			\begin{proof}
				We integrate $f_t$ against the measure $Z_t$
				\begin{equation*}
					\begin{array}{l}
					\ll Z_t,f_t\gg   =   \displaystyle \sum_{k = 1}^{N_0} f_t^{(I^{(k)}_0)}(A^{(k)}_{0} + t)    \\
				+ \displaystyle  \int_{[0,t] \times \mathcal{E}} \left [  \mathds{1}_{k \leq N_{s^-}} \left (2 f_t^{(I^{(k)}_{s^-})}(t-s) - f_t^{(I^{(k)}_{s^-})}(A^{(k)}_{s^-} + t-s) \right ) \mathds{1}_{0 \leq \theta \leq m_1(s,k,Z)}  \right. \\[0.3cm]
					+  \left ( f_t^{(I^{(k)}_{s^-})}(t-s) + f_t^{(I^{(k)}_{s^-} +1)}(t-s) -  f_t^{(I^{(k)}_{s^-})}(A^{(k)}_{s^-} + t-s) \right ) \mathds{1}_{m_1(s,k,Z)  \leq \theta \leq m_2(s,k,Z)} \\[0.3cm]
					\left.	 +  \left (2f_t^{(I^{(k)}_{s^-} + 1)}(t-s)  -  f_t^{(I^{(k)}_{s^-})}(A^{(k)}_{s^-} + t-s) \right ) \mathds{1}_{m_2(s,k,Z) \leq \theta \leq m_3(s,k,Z)}\right ] Q(ds,dk,d\theta)	.
				\end{array}
				\end{equation*}
			Derivating $f_t^{(j)} [a + t - s]$, we obtain 
				\begin{equation*}
				\begin{array}{l}
					\frac{d}{dt}\left [f_t^{(j)} [a + t - s]\right ] = \partial_1 f_t^{(j)} [a + t - s]  + 	\partial_2 f_t^{(j)} [a + t - s]  \\
				\Rightarrow \quad \int_{s}^{t}\frac{d}{du} \left [f_u^{(j)} [a + u - s]\right ]du  = \int_{s}^{t} \left [\partial_1 f_u^{(j)} [a + u -s] + \partial_2 f_u^{(j)} [a + u -s]   \right ]du \nonumber \\[0.3cm]
				\Rightarrow 	f_t^{(j)} [a + t - s] = f_s^{(j)} [a] + \int_{s}^{t} \left [\partial_1 f_u^{(j)} [a + u -s] + \partial_2 f_u^{(j)} [a + u -s]   \right ]du .
				\end{array}
				\end{equation*}
				Then, replacing $j$ by the index $I^{(k)}_{s^-} $ and $a$ by $A^{(k)}_{s^-} $ or $0$, it comes 
				\begin{equation*}
					\begin{array}{l}
					\!\!\!\!\!\!\!\!\!\!\!\!\!\!\!\!\!\!	\ll Z_t,f_t\gg  =  \sum_{k = 1}^{N_0} f_0^{(I^{(k)}_0)}(A^{(k)}_{0}) + T_0+ T_1 + T_2 + T_3  \\[0.5cm]
					+ \int_{[0,t] \times \mathcal{E}}\mathds{1}_{k \leq N_{s^-}} \left [ \left ( 2 f_s^{(I^{(k)}_{s^-})}(0) -  f_s^{(I^{(k)}_{s^-})}(A^{(k)}_{s^-}) \right )  \mathds{1}_{0 \leq \theta \leq m_1(s,k,Z)} \right .\\[0.5cm]
					+  \left (f_s^{(I^{(k)}_{s^-})}(0) + f_s^{(I^{(k)}_{s^-} +1)}(0)  -  f_s^{(I^{(k)}_{s^-})}(A^{(k)}_{s^-})  \right ) \mathds{1}_{m_1(s,k,Z) \leq \theta \leq m_2(s,k,Z)} \\[0.5cm]
					\left. 	+ \left  (2 f_s^{(I^{(k)}_{s^-} + 1)}(0)  
					-   f_s^{(I^{(k)}_{s^-})}(A^{(k)}_{s^-}) \right ) \mathds{1}_{m_2(s,k,Z) \leq \theta \leq m_3(s,k,Z)} \right ] Q(ds,dk,d\theta),
					\end{array}
				\end{equation*}
			
				where
				\begin{equation*}
				\begin{array}{ll} 
					T_0 = & \displaystyle \sum_{k = 1}^{N_0} \int_{0}^{t} \left  [\partial_1 f_u^{(I^{(k)}_0)}(A^{(k)}_{0} + u) + \partial_2 f_u^{(I^{(k)}_0)}(A^{(k)}_{0} + u) \right  ]du, \\[0.7cm]
					T_1 = & \displaystyle  \int_{[0,t] \times \mathcal{E}} \mathds{1}_{k \leq N_{s^-}} \int_{s}^{t} \left [ 2 \partial_1 f_u^{(I^{(k)}_{s^-} )}(u-s)  \right. \\[0.3cm] 
					& + 2\partial_2  f_u^{(I^{(k)}_{s^-} )}(u-s) - \partial_1 f_u^{(I^{(k)}_{s^-})}(A^{(k)}_{s^-} + u-s)  \\[0.3cm] 
					& \left. \displaystyle - \partial_2 f_u^{(I^{(k)}_{s^-})}(A^{(k)}_{s^-}  + u-s)  \mathds{1}_{0 \leq \theta \leq m_1(s,k,Z)} \right ]du Q(ds,dk,d\theta),  \\[0.7cm]
					T_2 =&\displaystyle \int_{[0,t] \times \mathcal{E}} \mathds{1}_{k \leq N_{s^-}} \int_{s}^{t} \left [ \partial_1 f_u^{(I^{(k)}_{s^-})}(u-s) + \partial_1 f_u^{(I^{(k)}_{s^-}  + 1)}(u-s)  + \partial_2  f_u^{(I^{(k)}_{s^-} )}(u-s) \right.  \\[0.3cm]
					& + \partial_2  f_u^{(I^{(k)}_{s^-} + 1)}(u-s)- \partial_1 f_u^{(I^{(k)}_{s^-})}(A^{(k)}_{s^-} + u-s) \\[0.3cm]
					& \left.  \displaystyle - \partial_2 f_u^{(I^{(k)}_{s^-})}(A^{(k)}_{s^-}  + u-s)   \mathds{1}_{m_1(s,k,Z) \leq \theta \leq m_2(s,k,Z)} \right ] du Q(ds,dk,d\theta),  \\[0.7cm]
					T_3 =& \displaystyle \int_{[0,t] \times \mathcal{E}} \mathds{1}_{k \leq N_{s^-}}  \int_{s}^{t} \left [2 \partial_1 f_u^{(I^{(k)}_{s^-} + 1)}(u-s)  \right. \\[0.3cm] 
			     	&	+ 2\partial_2  f_u^{(I^{(k)}_{s^-}  + 1)}(u-s)  - \partial_1 f_u^{(I^{(k)}_{s^-} )}(A^{(k)}_{s^-}  + u-s)  \\[0.3cm] 
					& \left. \displaystyle - \partial_2 f_u^{(I^{(k)}_{s^-})}(A^{(k)}_{s^-}  + u-s)   \mathds{1}_{m_2(s,k,Z) \leq \theta \leq m_3(s,k,Z)} \right ] du Q(ds,dk,d\theta) \, .
				\end{array}	
				\end{equation*}
				As the partial differential of each $f^{(j)}$ are uniformly bounded, we can apply Fubini theorem on $T_0 $, $T_1$, $T_2$ and $T_3$:
			
				\begin{equation*}
					\begin{array}{ll}
					T_0 = &\displaystyle	\int_{0}^{t}   \ll \partial_1 f_u + \partial_2 f_u, \sum_{k = 1}^{N_0}\delta_{I^{(k)}_{0},A^{(k)}_{0} + u} \gg du,\\[0.5cm]
				T_1 = &\displaystyle \int_0^t \left [    \ll  \partial_1 f_u + \partial_2 f_u,  \int_0^u \int_{\mathcal{E}} \mathds{1}_{k \leq N_{s^-}}\left (2\delta_{I^{(k)}_{s^-},u-s } \right. \right.  \\ 
				& \left. \left.  - \delta_{I^{(k)}_{s^-}, A^{(k)}_{s^-} + u-s }  \right )\mathds{1}_{0 \leq \theta \leq m_1(s,k,Z)}Q(ds,dk,d\theta) \gg  \right ]du, \\[0.5cm]
				T_2 = & \displaystyle \int_0^t \left [    \ll  \partial_1 f_u + \partial_2 f_u, \int_0^u \int_{\mathcal{E}} \mathds{1}_{k \leq N_{s^-}}  \left ( \delta_{I^{(k)}_{s^-},u-s } +  \delta_{I^{(k)}_{s^-} + 1 ,u-s } \right. \right.  \\
				&	 \left. \left.- \delta_{I^{(k)}_{s^-}, A^{(k)}_{s^-} + u-s } \right )\mathds{1}_{ m_1(s,k,Z) \leq \theta \leq m_2(s,k,Z)}Q(ds,dk,d\theta) \gg  \right ]du, \\[0.5cm]
				T_3 = &\displaystyle \int_0^t \left [    \ll \partial_1 f_u + \partial_2 f_u,  \int_0^u \int_{\mathcal{E}} \mathds{1}_{k \leq N_{s^-}} \left ( 2\delta_{I^{(k)+1}_{s^-},u-s } \right. \right. \\ 
				& \left. \left. - \delta_{I^{(k)}_{s^-}, A^{(k)}_{s^-} + u-s } \right )\mathds{1}_{ m_2(s,k,Z) \leq \theta \leq m_3(s,k,Z) }Q(ds,dk,d\theta) \gg   \right ]du .
					\end{array}
				\end{equation*}
				
				Finally, using the stochastic differential equation (\ref{Z_t_Equation})
					\begin{equation*}
					T_0 + 	T_1 + T_2 + T_3 = 
					\int_{0}^{t}\ll  \partial_1 f_u+ \partial_2  f_u, Z_u\gg  du .
				\end{equation*}
				Consequently, we obtain
				\begin{equation*}\label{Formule_ITO_perso)}
				\begin{array}{l}
			\displaystyle \!\!\!\!\!\!\!\!\!\!\!\!\!\!\!\!\!\! \ll f_t,Z_t\gg  =  \ll f_0,Z_0\gg   + \int_{0}^{t}\ll  \partial_1 f_s+ \partial_2  f_s, Z_s\gg  ds \\[0.5cm]
			+ 	\displaystyle \int_{[0,t] \times \mathcal{E}} \mathds{1}_{k \leq N_{s^-}}  \left [  \ll f_s, 2\delta_{I^{(k)}_{s^-},0} - \delta_{I^{(k)}_{s^-},A^{(k)}_{s^-}}  \gg \mathds{1}_{0 \leq \theta \leq m_1(s,k,Z)}  \right.\\[0.5cm]
			+	\displaystyle  \ll f_s, \delta_{I^{(k)}_{s^-},0} + \delta_{I^{(k)}_{s^-} + 1,0} - \delta_{I^{(k)}_{s^-},A^{(k)}_{s^-}}  \gg \mathds{1}_{m_1(s,k,Z)  \leq \theta \leq m_2(s,k,Z)} \\[0.5cm]
		\left. 	\displaystyle	+  \ll f_s, 2\delta_{I^{(k)}_{s^-} + 1,0} - \delta_{I^{(k)}_{s^-},A^{(k)}_{s^-}}  \gg \mathds{1}_{m_2(s,k,Z) \leq \theta \leq m_3(s,k,Z)}  \right ] Q(ds,dk,d\theta),
				\end{array}
			\end{equation*}
				which gives us lemma \ref{MultiCouche_lemme} for $F(x) = x$. We conclude by applying the Ito's formula (see \cite{protter_stochastic_2004}, p.68-70).
			\end{proof}
		
				We introduce the sequence of stopping times $\xi_N$. 
			\begin{definition}\label{Temps_darret_xi}
				Let $\xi_N$ a sequence of stopping times defined as
				\begin{eqnarray*}
					\xi_N  = \sup \big( t: N_t < N, \quad  \ll a ,Z_t \gg \, < N\big). 
				\end{eqnarray*}
			\end{definition} 
				
				\begin{proof}[Proof of theorem \ref{MultiCouche_Gene_Theo}] We first start by showing \eqref{SDE_N_tandZ_t}.
					\begin{equation*}
					N_t = N_0 + \int_{[0,t] \times \mathcal{E}} \mathds{1}_{k \leq N_{s^-}} \mathds{1}_{0 \leq \theta \leq b_{I^{(k)}_{s^-}}(A^{(k)}_{s^-})} Q(ds,dk,d\theta).
					\end{equation*} 
					Thus, 
						\begin{equation*}
						\sup_{s \leq t \wedge \xi_N}  N_s \leq N_0 + \int_{0}^{t \wedge \xi_N} \int_{\mathcal{E}} \mathds{1}_{k \leq N_{s^-}} \mathds{1}_{0 \leq \theta \leq \bar{b}} Q(ds,dk,d\theta),
					\end{equation*} 
					where  $\bar{b} := \underset{j \in \llbracket 0, J \rrbracket}{\sup} \bar{b}_j$. Taking the expectation and using Poisson measure properties, we obtain 
						\begin{equation*}
					\mathbb{E}[\sup_{s \leq t \wedge \xi_N}  N_s] \leq \mathbb{ E}[N_0] + \bar{b} \mathbb{E}\left [\int_{0}^{t \wedge \xi_N} N_s ds \right ] \leq \mathbb{ E}[N_0] + \bar{b} \mathbb{E}\left [\int_{0}^{t \wedge \xi_N} \sup_{u \leq s \wedge \xi_N} N_u ds \right ] .
					\end{equation*} 
					Hence, 
					\begin{equation*}
						\mathbb{E}[\sup_{s \leq t \wedge \xi_N}  N_s] \leq \mathbb{ E}[N_0] + \bar{b} \mathbb{E}\left [\int_{0}^{t} \sup_{u \leq s \wedge \xi_N} N_u ds \right ] .
					\end{equation*} 
					By Fubini theorem, we deduce that 
						\begin{equation*}
						\mathbb{E}[\sup_{s \leq t \wedge \xi_N}  N_s] \leq \mathbb{ E}[N_0] + \bar{b} \int_{0}^{t} \mathbb{E}\left [\sup_{u \leq s \wedge \xi_N} N_u \right ]  ds. 
					\end{equation*} 
					
					Applying Gr\"{o}nwall lemma, we deduce for all $t \leq T$ that 
						\begin{equation*}
							\mathbb{E}[\sup_{s \leq t \wedge \xi_N}  N_s] \leq \mathbb{ E}[N_0] e^{\bar{b} t } .
						\end{equation*} 
				    Hence, 
					\begin{equation*}
						\mathbb{E}[\sup_{t \leq T \wedge \xi_N}  N_t] \leq \mathbb{ E}[N_0] e^{\bar{b} T } \quad < \infty \, .
					\end{equation*} 
					Using the same method, we also deduce that $	\mathbb{E} \big[ \sup_{t \leq T}  \ll a, Z_t \gg  \big] < \infty $.
					
					Then, we use the same approach as \cite{tran_modeles_2006} (Theorem 2.2.8) to compute the infinitesimal generator of $Z_t $, denoted by $\mathcal{G} $. By construction, $(Z_t)_{t \in \mathbb{R}_+}$  is a markovian process of $\mathbb{D}([0,T],  \mathcal{M}_P(\llbracket 1, J \rrbracket  \times \mathbb{R}_+ )) $.
					 Let $f \in \mathbf{C}^1_b(\mathcal{E}, \mathbb{R})$, by definition, $\mathcal{G}F :=  \lim\limits_{ t \rightarrow 0}\frac{d}{dt} \mathbb{E}\big[  F[\ll f,Z_t\gg  ]\big] $. 
					Taking the expectation of the expression of $\ll f,Z_{t \wedge \xi_N}\gg  $ given in lemma  \ref{MultiCouche_lemme}, we obtain 
				\begin{equation*}
					\begin{array}{l}
					\displaystyle \!\!\!\!\!\!\!\!\!\!\!\!\!\!\!\!\!\!	\mathbb{E}\left [F [\ll f,Z_{t \wedge \xi_N}\gg ] \right ]= \mathbb{E} \left [ F [\ll f,Z_0\gg ] \right] \\
					\displaystyle	 + \mathbb{E} \left [\int_{0}^{t \wedge \xi_N} \ll \partial_a  f,Z_s\gg F'[\ll f,Z_s\gg ]ds \right ] + \mathbb{E} \left [ \chi^{f,F}(t \wedge \xi_N, Z) \right ] ,
					\end{array}
				\end{equation*}
				where 
				\begin{equation*}
				\begin{array}{rl}	
				\displaystyle	\chi^{f,F}(t, Z) := & \displaystyle   \int_{0}^{t}\int_{\mathcal{E}} \left [  \left ( F[ \ll f ,2\delta_{j,0} - \delta_{j,a} +Z_s \gg ] - F[<f,Z_s>]  \right ) p^{(j)}_{2,0}  \right. \\[0.3cm]
						&+  \left ( F[ <f  , \delta_{j+1,0} + \delta_{j,0} - \delta_{j,a} +Z_s>] - F[\ll f,Z_s\gg ]  \right ) p^{(j)}_{1,1} \\[0.3cm]
						& + \left ( F[<f ,2\delta_{j+1,0} - \delta_{j,a} +Z_s>] \right.  \\[0.3cm]
						& \left. \left. - F[<f,Z_s>]  \right ) p^{(j)}_{0,2} \right ]b_j(a)Z_s(dj,da)ds .
					\end{array}	
				\end{equation*}
				We have the following estimates, 
				\begin{equation*}
					\begin{array}{l}
					\mathbb{E}\left [\chi^{f,F}(t \wedge \xi_N, Z) \right ] \leq 2 	\mathbb{E} \left [ \sup_{t \leq T}  N_t \right ] T \lVert F \lVert_{\infty}   \bar{b}, \\[0.5cm]
					\mathbb{E}[\int_{0}^{t \wedge \xi_N} \ll  \partial_a  f_s,Z_s\gg F'[\ll f,Z_s\gg ]ds] 	 \leq 	\mathbb{E} \big[ \sup_{t \leq T}  N_t \big] T \times \lVert \partial_a f \lVert_{\infty} \times \lVert F'  \lVert_{\infty} .
					\end{array}
				\end{equation*}
				Those bounds are independent of $N$ thanks to \eqref{SDE_N_tandZ_t}, so that we may let $N$ goes to infinity. Moreover, 
				 $$ \frac{d}{dt }\int_{0}^{t} \ll \partial_a  f,Z_s\gg F'[\ll f,Z_s\gg ]ds =  \ll \partial_a  f,Z_t\gg F'[\ll f,Z_t\gg ]$$ 
				 which is also dominated by $	\mathbb{E} \big[ \sup_{t \leq T}  N_t \big]  \times \lVert \partial_a f \lVert_{\infty} \times \lVert F'  \lVert_{\infty}  $. Also, 
				\begin{equation*}
				\begin{array}{l}\label{MultiCouche_Derive}
				 \!\!\!\!\!\!\!\!\!\!\!\!\!\!\! \frac{\partial}{\partial t} \chi(t, Z_t)  =  \int_{\mathcal{E}} \left [ \left ( F[<f ,2\delta_{j,0} - \delta_{j,a} +Z_t>] - F[<f,Z_t>]  \right) p^{(j)}_{2,0} \right. \\[0.3cm]
				 + \left (F[ \ll f  , \delta_{j+1,0} + \delta_{j,0} - \delta_{j,a} +Z_t \gg ] - F[\ll f,Z_t\gg ]  \right ) p^{(j)}_{1,1}\\[0.3cm]
				  \left. +  \left ( F[ \ll f ,2\delta_{j+1,0} - \delta_{j,a} +Z_t \gg] - F[ \ll f,Z_t \gg ]  \right ) p^{(j)}_{0,2} \right ] b_j(a)Z_t(dj,da).
					\end{array}
			\end{equation*}
				 $ \rvert \frac{\partial}{\partial t} \chi(t, Z_t) \rvert$ is dominated $\mathbb{P}$-p.s  by $	2 	\mathbb{E} \big[ \sup_{t \leq T}  N_t \big]  \lVert F \lVert_{\infty}   \bar{b} $.
				We can thus apply the differentiating theorem under the integral sign $\mathbb{E}$ and conclude.
			\end{proof}

			\begin{proof}[Proof of lemma \ref{Multi_Dynkin}]
				Introducing the compensated Poisson measure  $\tilde Q$, \\
				$\tilde Q(ds,dk,d\theta) := Q(ds,dk,d\theta) - dsdkd\theta$, we define the process:
				\begin{equation*}
					\begin{array}{ll}
						M^{F,f}_t := & 
						  \displaystyle \int \int_{[0,t \wedge \xi_N] \times \mathcal{E}} \mathds{1}_{k < N_{s-}} \left [ \left ( F[\ll f ,2\delta_{I_{s^-}^k,0} - \delta_{I_{s^-}^k,A_{s^-}^k} +Z_{s-}\gg ]  \right. \right. \\[0.5cm] 
						& \left. - F[\ll f,Z_{s-}\gg ]  \right ) \mathds{1}_{0 \leq \theta \leq m_1(s,k,Z)} \\[0.5cm]
						& +\left ( F[ \ll f  , \delta_{I_{s^-}^k+1,0} + \delta_{I_{s^-}^k,0} - \delta_{I_{s^-}^k,A_{s^-}^k} +Z_{s-}\gg ]  \right.  \\[0.5cm] 
					 & \left.	- F[\ll f,Z_{s-}\gg ]  \right ) \mathds{1}_{m_1(s,k,Z)  \leq \theta \leq m_2(s,k,Z)}\\[0.5cm]
						&   + \left (F[\ll f ,2\delta_{I_{s^-}^k+1,0} - \delta_{I_{s^-}^k,A_{s^-}^k} +Z_{s-}\gg ]  \right.  \\[0.5cm] 
						& \left. \left. - F[\ll f,Z_{s-}\gg ]  \right) \mathds{1}_{m_2(s,k,Z)  \leq \theta \leq m_3(s,k,Z)} \right ] \tilde Q(ds,dk,d\theta) .
						\end{array}
				\end{equation*}
				
				We can verify that $M^{F,f}_t$ is a martingale as an integral against a compensated Poisson measure. 
				Then, applying lemma \ref{MultiCouche_lemme} and the definition of the generator given in theorem  \ref{MultiCouche_Gene_Theo}, we show that
				\begin{equation}\label{theore_Stoch_equ1}
					M_t^{F,f} =  	F[\ll f,Z_t\gg ] - F[\ll f,Z_0\gg ] - \displaystyle \int_{0}^{t} \mathcal{G}F[\ll f,Z_s\gg ]ds \, .
				\end{equation}
				We turn now to the computation of the quadratic variation and use the same approach as in \cite{champagnat_individual_2008}. We apply \eqref{theore_Stoch_equ1} for $F(x) = x^2$. Note that we cannot use directly this result as $x \mapsto x^2$ is not bounded and we need to first use  a localizing sequence (see \cite{klebaner_introduction_2012} p. 382, theorem 13.14). We obtain that
				\begin{equation}\label{Stoch_martin_e1}
				\begin{array}{ll}
				& \ll f,Z_t\gg^2 - 	\ll f,Z_0\gg^2 -	\displaystyle\int_{0}^{t} 2  \ll f,Z_s\gg  \times \ll \partial_af,Z_s\gg ds   \\
				& - \displaystyle\int_{0}^{t}\sum_{j = 1}^{J} \int_{\mathbb{R}_+} \left [ \left ( \ll f, 2 \delta_{j,0} - \delta_{j,a} + Z_s\gg ^2 - \ll f,Z_s\gg ^2  \right)b_{j}(a)p^{(j)}_{2,0}  \right. \\[0.3cm]
					& -   \left ( \ll  f , \delta_{j,0} + \delta_{j+1,0} - \delta_{j,a} + Z_s\gg ^2 - \ll f,Z_s\gg ^2 \right) b_{j}(a)p^{(j)}_{1,1}Z_s(dj,da) \\[0.3cm]
					& \left. -  \left ( \ll  f ,2\delta_{j+1,0} - \delta_{j,a} +Z_s\gg ^2 - \ll f,Z_s\gg ^2  \right ) b_{j}(a)p^{(j)}_{0,2}Z_s(dj,da)ds \right ]	
					\end{array}
				\end{equation}
				is a martingale. Then, applying  \eqref{theore_Stoch_equ1} for $F(x) = x$ (using a localizing sequence again), we get that 
				\begin{equation*}
			\begin{array}{ll}
			& \ll f,Z_t\gg \quad = \quad 	\ll f,Z_0\gg  + \displaystyle \int_{0}^{t} \ll \partial_af,Z_s\gg ds \\[0.3cm]
			& + \displaystyle \int_{0}^{t} \big[ \sum_{j = 1}^{J} \int_{\mathbb{R}_+}\ll f,2 \delta_{j,0} - \delta_{j,a} \gg  b_{j}(a)p^{(j)}_{2,0}Z_s(dj,da) \big]ds\\[0.3cm]
			& + \displaystyle \int_{0}^{t} \big[ \sum_{j = 1}^{J} \int_{\mathbb{R}_+} \ll f,   \delta_{j,0} + \delta_{j+1,0} - \delta_{j,a} \gg  b_{j}(a)p^{(j)}_{1,1}Z_s(dj,da) \big]ds \\[0.3cm]
			& + \displaystyle \int_{0}^{t} \big[ \sum_{j = 1}^{J} \int_{\mathbb{R}_+} \ll f,2\delta_{j+1,0} - \delta_{j,a} \gg  b_{j}(a)p^{(j)}_{0,2}Z_s(dj,da)  \big]ds  + M^f_t
			\end{array}
			\end{equation*}
			is a semi-martingale. Applying the Ito formula (see \cite{protter_stochastic_2004}, p. 78-79), we obtain:
				\begin{equation}\label{Stoch_martin_e2}
			\begin{array}{l}
			 \ll f,Z_t\gg^2  \quad - \quad 	\ll f,Z_0\gg^2  - \displaystyle \int_{0}^{t}2 \ll f,Z_s\gg  \times \ll \partial_af,Z_s\gg ds \\
			 \displaystyle  + \sum_{j = 1}^{J} \int_{\mathbb{R}_+}\left (  2 \ll f,Z_s\gg \times \ll f, 2 \delta_{j,0} - \delta_{j,a} \gg   \right ) b_{j}(a)p^{(j)}_{2,0}Z_s(dj,da) \\[0.3cm]
			  \displaystyle + \sum_{j = 1}^{J} \int_{\mathbb{R}_+}  \left (  2  \ll f,Z_s\gg \times  \ll  f , \delta_{j,0} + \delta_{j+1,0} - \delta_{j,a} \gg  \right ) b_{j}(a)p^{(j)}_{1,1}Z_s(dj,da) \\[0.3cm]
			 \displaystyle + \sum_{j = 1}^{J} \int_{\mathbb{R}_+} \left (  2  \ll f,Z_s\gg \times  \ll  f ,2\delta_{j+1,0} - \delta_{j,a} \gg   \right ) b_{j}(a)p^{(j)}_{0,2}Z_s(dj,da) \big]ds 	\\
			  - \left \langle M^{f},M^{f}\right \rangle_t
			\end{array}
			\end{equation}
			is a martingale. We consider the jump corresponding to the case when the two daughter cells remain on their mother layer. Note that 
			\begin{equation*}
				\begin{array}{l}
			\!\!\!\!\!\!\!\!\!\!\!\!\!\!\!\!	\ll f, 2 \delta_{j,0} - \delta_{j,a} + Z_s\gg ^2 - \ll f,Z_s\gg ^2 =  \\ 
				2\ll f, Z_s\gg\times \ll f, 2 \delta_{j,0} - \delta_{j,a}\gg + \ll f, 2 \delta_{j,0} - \delta_{j,a}\gg^2.  
				\end{array}
			\end{equation*}
			
			We proceed similarly for the two other jumps. Applying the Doob-Meyer theorem (\cite{protter_stochastic_2004}, p. 106), we deduce the quadratic variation $ \left \langle M^{f},M^{f}\right \rangle_t $ comparing \eqref{Stoch_martin_e1} and \eqref{Stoch_martin_e2}. 
			\end{proof}

			\subsection{Moment study} \label{Supplemental_Branching}
			\paragraph{Generating functions}
			\begin{proof}[Proof of lemma \ref{FoncGenForme2}]
				Let $a \geq 0$. 
				Remind that the generating function is given by
				$$F^{(i,a)}[\mathbf{s};t] =  \underset{\mathbf{k} \in \mathbb{N}^J}{\sum} \mathbf{s}^\mathbf{k} \mathbb{P}\big[Y^a_t = \mathbf{k} |Z_0 = \delta_{i,0} \big] .$$
				Let $i \in \llbracket 1, J \rrbracket $ and $\mathbf{j},\mathbf{k} \in \mathbb{N}^J$. We note $P^a_{\mathbf{j}, \mathbf{k}}(t) := \mathbb{P}\big[Y^a_t = \mathbf{k} |Z_0 = \sum_{i = 1}^{J} j_i\delta_{i,0} \big]$. We write the backward equation for the probability  $P^a_{e_i, \mathbf{k}}(t) := \mathbb{P}\big[Y_t^a = \mathbf{k} |Z_0 = \delta_{i,0} \big]$. Starting from a single mother cell of age $0$ and layer $i$, there are three possibilities at time $t$: (i) the cell has not divided and $t \leq a$, (ii) the cell has not divided and $t  >  a$, and (iii) the cell has divided. Thus, 
				\begin{multline}\label{Trans_gene_e1}
				P^a_{e_i, \mathbf{k}}(t)=  \left (\delta_{e_i,\mathbf{k}}\mathds{1}_{t \leq a} + \delta_{\mathbf{0},\mathbf{k}}\mathds{1}_{t > a} \right )\mathbb{P}[\tau^{(i)}(a_0 = 0)\geq t]  \\ + 	\int_{0}^{t}[p_{2,0}^{(i)} P^a_{2e_{i},\mathbf{k}}(t-y) + p_{1,1}^{(i)}P^a_{e_{i}+e_{i+1},\mathbf{k}}(t-y) + p_{0,2}^{(i)}P^a_{2e_{i+1},\mathbf{k}}(t-y)]d\mathcal{B}_{i}(y)dy 
				\end{multline}
				where $\mathbb{P}[\tau^{(i)}(a_0 = 0)\geq t ] =  e^{-\int_{0}^{t}b_{i}(s)ds}\mathds{1}_{t \geq 0} = 1 - \mathcal{B}_{i}(t)$.
				
				Applying the branching property, we have for all $y \in [0,t]$, for all $ i \in \llbracket 1, J \rrbracket$ 
				$$P^a_{2e_{i},\mathbf{k}}(y) = \sum_{\mathbf{k}_{1},\mathbf{k}_{2} / \mathbf{k}_{1} + \mathbf{k}_{2} = \mathbf{k}}[P^a_{e_{i},\mathbf{k}_{1}}(y)P^a_{e_{i},\mathbf{k}_{2}}(y)], $$ and also, for all $ i \in \llbracket 1, J -1 \rrbracket$,  $$P^a_{e_{i}+e_{i+1},\mathbf{k}}(y) = \sum_{\mathbf{k}_{1},\mathbf{k}_{2} / \mathbf{k}_{1} + \mathbf{k}_{2} = \mathbf{k}}P^a_{e_{i+1},\mathbf{k}_{1}}(y)P^a_{e_{i},\mathbf{k}_{2}}(y) \, .$$ 
				
				Hence, we can rewrite the expression of $$ \displaystyle A_t :=	\int_{0}^{t}[p_{2,0}^{(i)} P^a_{2e_{i},\mathbf{k}}(t-y) + p_{1,1}^{(i)}P^a_{e_{i}+e_{i+1},\mathbf{k}}(t-y) + p_{0,2}^{(i)}P^a_{2e_{i+1},\mathbf{k}}(t-y)]d\mathcal{B}_{i}(y)dy   $$ as
				\begin{multline*}
				A_t =  p_{2,0}^{(i)}\int_{0}^{t} \sum_{\mathbf{k}_{1},\mathbf{k}_{2} / \mathbf{k}_{1} + \mathbf{k}_{2} = \mathbf{k}}P^a_{e_{i},\mathbf{k}_{1}}(t-y)P^a_{e_{i},\mathbf{k}_{2}}(t-y) d\mathcal{B}_{i}(y)dy \\
				+ p_{1,1}^{(i)}\int_{0}^{t}\sum_{\mathbf{k}_{1},\mathbf{k}_{2} / \mathbf{k}_{1} + \mathbf{k}_{2} = \mathbf{k}}P^a_{e_{i+1},\mathbf{k}_{1}}(t-y)P^a_{e_{i},\mathbf{k}_{2}}(t-y)d\mathcal{B}_{i}(y)dy \\
				+ p_{0,2}^{(i)}\int_{0}^{t}\sum_{\mathbf{k}_{1},\mathbf{k}_{2} / \mathbf{k}_{1} + \mathbf{k}_{2} = \mathbf{k}}P^a_{e_{i+1},\mathbf{k}_{1}}(t-y)P^a_{e_{i+1},\mathbf{k}_{2}}(t-y)d\mathcal{B}_{i}(y)dy .
				\end{multline*}

				Note that
				\begin{multline*}
				\sum_{\mathbf{k} \in \mathbb{N}^{J}}\mathbf{s}^{\mathbf{k}}P^a_{2e_{i},\mathbf{k}}(t-y)  = \sum_{\mathbf{k} \in \mathbb{N}^{J}}\mathbf{s}^{\mathbf{k}} \sum_{\mathbf{k}_{1},\mathbf{k}_{2} / \mathbf{k}_{1} + \mathbf{k}_{2} = \mathbf{k}}P^a_{e_{i},\mathbf{k}_{1}}(t-y)P^a_{e_{i},\mathbf{k}_{2}}(t-y)\\
				= \sum_{\mathbf{k}\in \mathbb{N}^{J}} \sum_{\mathbf{k}_{1}= 0}^{\mathbf{k}}\mathbf{s}^{\mathbf{k}_{1}}P^a_{e_{i},\mathbf{k}_{1}}(t-y) \mathbf{s}^{\mathbf{k}-\mathbf{k}_{1}}P^a_{e_{i},\mathbf{k}-\mathbf{k}_{1}}(t-y) .
				\end{multline*}
				We note $\sum_{\mathbf{k}_{1}= 0}^{\mathbf{k}} $ the sum of all the $\mathbf{k}_1 \in \mathbb{N}^J$ vectors such that $\mathbf{k}_1 \leq \mathbf{k}$ component by component. We have
				\begin{multline*}
				\sum_{\mathbf{k} \in \mathbb{N}^{J}}\mathbf{s}^{\mathbf{k}}P^a_{2e_{i},\mathbf{k}}(t-y)  = \sum_{\mathbf{k}\in \mathbb{N}^{J}} \sum_{\mathbf{k}_{1}= 0}^{\mathbf{k}}\mathbf{s}^{\mathbf{k}_{1}}P^a_{e_{i},\mathbf{k}_{1}}(t-y) \mathbf{s}^{\mathbf{k}-\mathbf{k}_{1}}P^a_{e_{i},\mathbf{k}-\mathbf{k}_{1}}(t-y) \\
				= \sum_{\mathbf{k_1}\in \mathbb{N}^{J}} \mathbf{s}^{\mathbf{k}_{1}}P^a_{e_{i},\mathbf{k}_{1}}(t-y) \sum_{ \mathbf{k}\geq \mathbf{k}_{1}} \mathbf{s}^{\mathbf{k}-\mathbf{k}_{1}}P^a_{e_{i},\mathbf{k}-\mathbf{k}_{1}}(t-y) \\
				= (\sum_{\mathbf{k}_{1} \in \mathbb{N}^{J}}\mathbf{s}^{\mathbf{k}_{1}}P^a_{e_{i},\mathbf{k}_{1}}(t-y))(\sum_{\mathbf{k}_{2} \in \mathbb{N}^{J}}\mathbf{s}^{\mathbf{k}_{2}}P^a_{e_{i},\mathbf{k}_{2}}(t-y)).
				\end{multline*}
				Hence, 
				\begin{equation}\label{Trans_gen_e2}
				\sum_{\mathbf{k} \in \mathbb{N}^{J}}\mathbf{s}^{\mathbf{k}}P^a_{2e_{i},\mathbf{k}}(t-y)  =(F^{(i,a)}[\mathbf{s};t-y])^{2}.
				\end{equation}
				In the same way, we also obtain
				\begin{equation}\label{Trans_gen_e3}
				\sum_{\mathbf{k} \in \mathbb{N}^{J}}\mathbf{s}^{k}P^a_{e_{i}+e_{i+1},k}(t-y) = F^{(i,a)}[\mathbf{s};t-y]F^{(i+1,a)}[\mathbf{s};t-y] 
				\end{equation}
				and
				\begin{equation}\label{Trans_gen_e4}
				\sum_{\mathbf{k} \in \mathbb{N}^{J}}\mathbf{s}^{\mathbf{k}}P^a_{2e_{i+1},\mathbf{k}}(t-y) = (F^{(i+1,a)}[\mathbf{s};t-y])^{2} .
				\end{equation}
				Finally, multiplying \eqref{Trans_gene_e1} by $\mathbf{s}^{\mathbf{k}}$, summing on $\mathbf{k} \in \mathbb{N}^J$ and applying \eqref{Trans_gen_e2}-\eqref{Trans_gen_e4}, we obtain:
				
				\begin{equation*}
				\forall i \in \llbracket 1, J \rrbracket,	F^{(i,a)}[\mathbf{s};t]  = (s_{i}\mathds{1}_{t \leq a} + \mathds{1}_{t > a})(1 - \mathcal{B}_{i}(t)) + \int_{0}^{t} f^{(i)}(F[\mathbf{s};t-y])  d\mathcal{B}_{i}(y)dy .
				\end{equation*}
				\vspace{-0.15cm}
			\end{proof} 
			
			\paragraph{First moments}
			\begin{proof}[Proof of lemma \ref{MomentDOrdre1}]
				By classical property, $M^a_{i,j}(t) = \frac{\partial}{\partial s_{j}}F^{(i,a)}[s;t]|_{\mathbf{s} = 1} $. From (\ref{FonctionGeneratriceEquation}) it comes that
				\begin{equation}\label{Transitory_Gene_Deriv_Equ}
				\frac{\partial}{\partial s_{j}}F^{(i,a)}\big[s;t\big]  =  \delta_{i,j}(1 - \mathcal{B}_{i,i}(t))\mathds{1}_{ a \geq t }  + \int_{0}^{t}\frac{\partial}{\partial s_{j}}f^{(i)}\big[F^a(s,y)\big]d\mathcal{B}_{i}(t-y) dy
				\end{equation}
				where
				\begin{multline*}
				\frac{\partial}{\partial s_{j}}f^{(i)}[F^a(s,t)] = 2p^{(i)}_{2,0}F^{(i,a)}[s;t]\frac{\partial}{\partial s_{j}}F^{(i,a)}[s;t] +  2p^{(i)}_{0,2}F^{(i+1,a)}[s;t]\frac{\partial}{\partial s_{j}}F^{(i+1,a)}[s;t]\\
				+ p^{(i)}_{1,1}\big[F^{(i+1,a)}[s;t]\frac{\partial}{\partial s_{j}}F^{(i,a)}[s;t]  + F^{(i,a)}[s;t]\frac{\partial}{\partial s_{j}}F^{(i+1,a)}[s;t] \big] .
				\end{multline*}
				For $\mathbf{s} = 1$, knowing that $ F^{(i,a)}(1,t) = 1 $, we get
				\begin{multline*}
				M^a_{i,j}(t)=  \delta_{i,j}(1 - \mathcal{B}_{i}(t))\mathds{1}_{ t \leq a }   \\
				+ \int_{0}^{t}\big[2p^{(i)}_{2,0} M^a_{i,j}(y) + p^{(i)}_{1,1} [M^a_{i,j}(y) + M^a_{i+1,j}(y)] + 2p^{(i)}_{0,2}M^a_{i+1,j}(y)\big]d\mathcal{B}_{i}(t-y)) dy
				\end{multline*}
				which can be rewritten as
				
				\begin{equation*}
				M^a_{i,j}(t)= \delta_{i,j}(1-\mathcal{B}_{i}(t) )\mathds{1}_{ t \leq a } + \big[2p^{(i)}_{S} M^a_{i,j} + 2p^{(i)}_{L}M^a_{i+1,j}\big]\ast d \mathcal{B}_{i}(t)  .
				\end{equation*}
				\vspace{-0.15cm}
			\end{proof}
		
			\paragraph{Harris lemmas}
			We recall some results on the renewal theory presented in \cite{harris_theory_1963}, p.161-163. \\
			Let $G$ be a distribution function on $(0, \infty)$ with the additional assumption $G(0 +) = 0 $. We consider the renewal equation 
			\begin{equation}\label{Supp_RenewalEqu}
				K(t) = f(t) + m \int_{0}^{t} K(t-u)dG(u) =  f(t) + m K \ast G(t)
			\end{equation}
			where $m$ is a positive constant representing the mean number of children, $f$ is a continuous function representing a source term and $G$ is the life time distribution. In addition, we suppose that $G$ is not lattice. 
			
			\begin{lemma}[Harris's lemma 2, p.161]\label{HarrisLemma_2}
				Suppose that there exists a Malthus parameter $\alpha$ such that $m\int_{0}^{\infty}e^{- \alpha t} dG(t) = 1 $, and that the following conditions also hold:
				\begin{itemize}
					\item[(a)] $f(t)e^{- \alpha t}$ is a continuous function such that $f(t)e^{- \alpha t} \in \mathbf{L}^1(\mathbb{R}_+) $.
					\item[(b)] $\displaystyle \int_{0}^{\infty} t^2 dG(t) < \infty $.
				\end{itemize}
			Then, $ K(t) \sim n_f e^{\alpha t}$, where 
			\begin{equation*}
				n_f = \displaystyle \frac{\int_{0}^{\infty} f(t) e^{- \alpha t} dt }{m \int_{0}^{\infty} t e^{- \alpha t} dG(t)}.
			\end{equation*}
			\end{lemma}
			
				\begin{lemma}[Harris's lemma 4, p.163]\label{HarrisLemma_4}
				Suppose that $m< 1$ and $\lim\limits_{t \rightarrow \infty} f(t) = c$. then $K(t) \rightarrow \frac{c}{1 - m}$.
			\end{lemma}

				\paragraph{Additional computation details for the proof of theorem \ref{MoyenneTempsLong}} \label{AdditDetails}
				
				We detail how to obtain formula \eqref{M_k_E3}. We first take the Laplace transform of \eqref{MomentDOrdre1Equ} for $\alpha = \Malthus$ for $i = \MalthusIndex + 1$ and $j \in \llbracket \MalthusIndex+1, J \rrbracket $. We distinguish the case $i = j$ from the others. If $j = \MalthusIndex + 1$, we obtain 
				\begin{multline*}
				\int_{0}^{\infty}M^a_{j,j}(t)e^{-\Malthus t}dt = \\ \frac{1}{\eigenfunction ^{(j)}(0)}	\int_{0}^{a} \eigenfunction ^{(j)}(t)dt + 2p_S^{(j)} 	\int_{0}^{\infty} \left [ \int_{0}^{t}d \mathcal{B}_j(t -u)M^a_{j,j}(u) du \right ] e^{-\Malthus t}dt.
				\end{multline*}
				By the Laplace transform property for the convolution, we deduce that 
				\begin{equation*}
				\int_{0}^{\infty} \left [ \int_{0}^{t}d \mathcal{B}_j(t -u)M^a_{j,j}(u) du \right ] e^{-\Malthus t}dt =  d \mathcal{B}^*_j(\Malthus)\int_{0}^{\infty}M^a_{j,j}(t)e^{-\Malthus t}dt, 
				\end{equation*}
				hence
				\begin{multline*}
				\int_{0}^{\infty}M^a_{j,j}(t)e^{-\Malthus t}dt =  \frac{1}{\eigenfunction ^{(j)}(0)}	\int_{0}^{a} \eigenfunction ^{(j)}(t)dt + 2p_S^{(j)} d \mathcal{B}^*_j(\Malthus)\int_{0}^{\infty}M^a_{j,j}(t)e^{-\Malthus t}dt \\
				=  \frac{1}{\eigenfunction ^{(j)}(0) \times (1 - 2p_S^{(j)} d \mathcal{B}^*_j(\Malthus))}	\int_{0}^{a} \eigenfunction ^{(j)}(t)dt .
				\end{multline*}
				
				When $j > \MalthusIndex+ 1$, we have:
				\begin{multline*}
				\int_{0}^{\infty}M^a_{\MalthusIndex + 1,j}(t)e^{-\Malthus t}dt = \\
				2p_S^{(\MalthusIndex + 1)} d \mathcal{B}^*_{\MalthusIndex + 1}(\Malthus)\int_{0}^{\infty}M^a_{\MalthusIndex + 1,j}(t)e^{-\Malthus t}dt +  2p_L^{(\MalthusIndex + 1)} d \mathcal{B}^*_j(\Malthus)\int_{0}^{\infty}M^a_{\MalthusIndex + 2,j}(t)e^{-\Malthus t}dt .
				\end{multline*}
				
				Hence, 
				\begin{equation*}
				\int_{0}^{\infty}M^a_{\MalthusIndex + 1,j}(t)e^{-\Malthus t}dt = \frac{ 2p_L^{(\MalthusIndex + 1)} }{1 - 2p_S^{(\MalthusIndex + 1)} d \mathcal{B}^*_{\MalthusIndex + 1}(\Malthus)} \int_{0}^{\infty}M^a_{\MalthusIndex + 2,j}(t)e^{-\Malthus t}dt .
				\end{equation*}
				Here, we obtain a recurrence formula between $\displaystyle	\int_{0}^{\infty}M^a_{\MalthusIndex + 1,j}(t)e^{-\Malthus t}dt $ and \\ $ \displaystyle \int_{0}^{\infty}M^a_{\MalthusIndex + 2,j}(t)e^{-\Malthus t}dt  $, and we obtain \eqref{M_k_E3}.

			\paragraph{Second moments}
				\begin{definition}\label{Definition_M2}
					Let $a\geq 0$. We define the second moment $$L^a(t) := (\mathbb{E}[ (Y_t^{(a,j)})^2 \lvert Z_0 = \delta_{i,0}])_{i,j \in \llbracket 1, J \rrbracket } \, .$$
				\end{definition}
					\begin{lemma}\label{MomentDOrdre2}
					$L^a(t)$ is solution of the renewal equation: $	\forall (i,j) \in \llbracket 1, J\rrbracket ^{2},$
					\begin{equation}\label{MomentDOrdre2Equ}
				 \begin{array}{lr}
						L^a_{i,j}(t) = & \delta_{i,j}(1- \mathcal{B}_{i}(t))\mathds{1}_{  t \leq a }   + [2p^{(i)}_{S}  L^a_{i,j} + 2p^{(i)}_{L} L^a_{i+1,j} ]\ast d\mathcal{B}_{i}(t)  \\
						& +  [2p^{(i)}_{2,0} (M^a_{i,j})^2 + 2 p^{(i)}_{1,1}M^a_{i,j}M^a_{i+1,j}+ 2p^{(i)}_{0,2} (M^a_{i+1,j})^2 ]\ast d\mathcal{B}_{i}(t).
						\end{array}
					\end{equation}
				\end{lemma}
				\begin{proof}[Proof of lemma (\ref{MomentDOrdre2})]
					Note that  $\frac{\partial^2}{\partial s_{j}}F^{(i,a)}[s;t]|_{\mathbf{s} = 1} = L^a_{i,j}(t) - M^a_{i,j}(t) $ . We derive (\ref{Transitory_Gene_Deriv_Equ}) with respect to $s_j$ and obtain:
					\begin{equation*}
						\frac{\partial^2}{\partial s_{j}^2}F^{(i,a)}\big[\mathbf{s};t\big]  =  \int_{0}^{t}\frac{\partial}{\partial s_{j}^2}f^{(i)}\big(F^a[\mathbf{s},u]\big)d\mathcal{B}_{i}(t-u) du
					\end{equation*}
					where 
					\begin{multline*}
						\frac{\partial^2}{\partial s_{j}^2}f^{(i)}\big(F^a[\mathbf{s},t]\big) =  2p^{(i)}_{2,0}\left (F^{(i,a)}[\mathbf{s};t]\frac{\partial^2}{\partial s_{j}}F^{(i,a)}[\mathbf{s};t] + (\frac{\partial}{\partial s_{j}}F^{(i,a)}[\mathbf{s};t] )^2 \right )  \\
						 +  2p^{(i)}_{0,2} \left (F^{(i+1,a)}[\mathbf{s};t]\frac{\partial^2}{\partial s_{j}}F^{(i+1,a)}[\mathbf{s};t] + (\frac{\partial}{\partial s_{j}}F^{(i+1,a)}[\mathbf{s};t] )^2 \right ) \\
						 + p^{(i)}_{1,1} \left ( F^{(i+1,a)}[\mathbf{s};t]\frac{\partial^2}{\partial s_{j}}F^{(i,a)}[\mathbf{s};t]  + 2\frac{\partial}{\partial s_{j}}F^{(i,a)}[\mathbf{s};t]\frac{\partial}{\partial s_{j}}F^{(i+1,a)}[\mathbf{s};t]  \right.  \\
						 \left. + F^{(i,a)}[\mathbf{s};t]\frac{\partial^2}{\partial s_{j}}F^{(i+1,a)}[\mathbf{s};t] \right ) .
					\end{multline*}
					
					When $\mathbf{s} = 1$, we get
					\begin{multline*}
					L^a_{i,j}(t) - M^a_{i,j}(t)= 2p^{(i)}_{2,0} \left (	L^a_{i,j} - M^a_{i,j} +  (M^a_{i,j})^{2} \right )\ast d\mathcal{B}_{i}(t)  \\
						 +  2p^{(i)}_{0,2}\left (L^a_{i+1,j} - M^a_{i+1,j}+  (M^a_{i+1,j})^{2} \right )\ast d\mathcal{B}_{i}(t) \\
						 + p^{(i)}_{1,1}\left (L^a_{i,j} - M^a_{i,j}+ 2M^a_{i,j}M^a_{i+1,j}  +  L^a_{i+1,j}- M^a_{i+1,j} \right ) \ast d\mathcal{B}_{i}(t).
					\end{multline*}
					Using the system of equations (\ref{MomentDOrdre1Equ}), we deduce (\ref{MomentDOrdre2Equ}).
				\end{proof}

				\begin{theorem}\label{MomentDOrdre2Long}
						Under the same hypotheses as in theorem \ref{MoyenneTempsLong}, and supposing that for all $i \in \llbracket 1, J \rrbracket, \MalthusPerLayer_i > 0$, we have, for all $a \geq 0$:
					\begin{equation*}
						\forall i \in \llbracket 1, J \rrbracket, \quad  \forall k \in \llbracket 0, J-i \rrbracket \quad 
						L^a_{i,i+k}(t) \sim \widetilde{L}_{i,i+k}(a) e^{2\MalthusPerLayer_{i,i+k} t}, \text{  as  } t\rightarrow \infty
					\end{equation*}
					such that 
					\begin{equation*}
						\widetilde{L}_{i,i}(a) = \frac{2p^{(i)}_{2,0} d\mathcal{B}^*_{i}(2\MalthusPerLayer_{i})( \widetilde{ M}^a_{i,i})^2}{1- 2p^{(i)}_{S}d\mathcal{B}^*_{i}(2\MalthusPerLayer_{i} )},
					\end{equation*}
					and for $k \in \llbracket 1, J - i \rrbracket$,
					\begin{equation}\label{M_tilde_ik_def_moment2}
					\widetilde{L}_{i,i+k}(a) = 
					\left\{
					\begin{array}{ll}
					\frac{2p^{(i)}_{2,0}(\widetilde{M}^a_{i,i+k})^2 d\mathcal{B}^*_{i}(2\MalthusPerLayer_{i,i +k}) }{1 -2p_S^{(i)}d \mathcal{B}_i^*(2\MalthusPerLayer_{i,i+k}) } + l_{i,i+k}(a), & \text{ if } \MalthusPerLayer_{i,i+k} \neq  \MalthusPerLayer_i   \\[0.3cm]
					\frac{2p^{(i)}_{2,0}(\widetilde{M}^a_{i,i+k})^2 d\mathcal{B}^*_{i}(2\MalthusPerLayer_{i,i +k}) }{1 -2p_S^{(i)}d \mathcal{B}_i^*(2\MalthusPerLayer_{i,i+k}) }, & \text{ if } \MalthusPerLayer_{i,i+k} =  \MalthusPerLayer_i
					\end{array}
					\right. 
					\end{equation}
					where $$l_{i,i+k}(a) = 	\frac{ \left [  \widetilde{L}_{i+1,i+k}(a) + 2 p^{(i)}_{1,1}\widetilde{M}^a_{i,i+k}\widetilde{M}^a_{i+1,i+k}+ 2p^{(i)}_{0,2} (\widetilde{M}^a_{i+1,i+k})^2  \right ]d\mathcal{B}^*_{i}(2\MalthusPerLayer_{i,i +k})}{1-2p_S^{(i)}d \mathcal{B}_i^*(2\MalthusPerLayer_{i,i+k}) }.$$
				\end{theorem}
				\begin{proof}[Proof] Let $a \geq 0$. We introduce the following notations 
					\begin{equation*}
					\widehat{L}^a_{i,i+k}(t) = L^a_{i,i+k}(t)e^{- 2\MalthusPerLayer_{i,i+k}t}, 	\quad \widehat{d\mathcal{B}_{i}}(t) = \frac{d\mathcal{B}_{i}(t)}{d\mathcal{B}^*_{i}(2\MalthusPerLayer_{i,i+k})}e^{-2\MalthusPerLayer_{i,i+k}t}. 
					\end{equation*}
					We use the same approach as that performed for the proof of theorem \ref{MoyenneTempsLong}, and proceed by recurrence:
				\begin{equation*}
					\mathcal{H}^{k} : \quad \forall i \in \llbracket 1, J - k \rrbracket, \begin{array}{l}
				L^a_{i,i+k}(t) \sim \widetilde{L}^{a}_{i,i+k} e^{2\MalthusPerLayer_{i,i+k} t}, \text{  as  } t\rightarrow \infty.
				\end{array}
				\end{equation*}
				When $ k = 0$, according to \eqref{MomentDOrdre2Equ} $L^a_{i,i} $ is solution of the renewal equation:
				\begin{equation}\label{MomentOrderTwo_eq_diag}
				L^a_{i,i}(t) =  \left ( 1- \mathcal{B}_{i}(t) \right ) \mathds{1}_{t \leq a} + 2p^{(i)}_{2,0} (M^a_{i,i})^2 \ast d\mathcal{B}_{i}(t) + 2p^{(i)}_{S} L^a_{i,i} \ast d\mathcal{B}_{i}(t)  .
				\end{equation}
				We rescale \eqref{MomentOrderTwo_eq_diag} by $ e^{- 2\MalthusPerLayer_{i} t}$ and obtain:
				\begin{equation*}
					\widehat{L}^a_{i,i}(t) =  e^{- 2\MalthusPerLayer_{i} t}\left [\left ( 1- \mathcal{B}_{i}(t) \right ) \mathds{1}_{t \leq a} + 2p^{(i)}_{2,0} (M^a_{i,i})^2 \ast d\mathcal{B}_{i}(t) \right ] + 2p^{(i)}_{S}d\mathcal{B}^*_{i}(2\MalthusPerLayer_{i} )  \widehat{L}^a_{i,i} \ast \widehat{d\mathcal{B}}_{i}(t)  .
				\end{equation*}
				Note that as $ 2\MalthusPerLayer_{i} > \MalthusPerLayer_{i} > 0$, we have $2p^{(i)}_{S}d\mathcal{B}^*_{i}(2\MalthusPerLayer_{i} )  < 1$, so that we can use lemma \ref{HarrisLemma_4}. We compute the limit of the source term :
				\begin{equation*}
					\lim\limits_{t \rightarrow \infty} e^{- 2\MalthusPerLayer_{i} t} \left [\left ( 1- \mathcal{B}_{i}(t) \right ) \mathds{1}_{t \leq a}  + 2p^{(i)}_{2,0} (M^a_{i,i})^2 \ast d\mathcal{B}_{i}(t) \right ]. 
				\end{equation*}
				From hypothesis \ref{Hypothesis_DivisionRate}, we have:
				\begin{equation*}
				\int_{0}^{\infty} \left ( 1- \mathcal{B}_{i}(t) \right )\mathds{1}_{t \leq a}    e^{-\MalthusPerLayer_i t}dt \leq \frac{1}{\bar{b}_i} \int_{0}^{\infty}d\mathcal{B}_{i}(t) e^{-\MalthusPerLayer_i t}dt < \infty.
				\end{equation*} 
				Thus, $ \left ( 1- \mathcal{B}_{i}(t) \right )\mathds{1}_{t \leq a}e^{-\MalthusPerLayer_i t} \in \mathbf{L}^1(\mathbb{R}_+) $ and, $\lim\limits_{t \rightarrow \infty} e^{- \MalthusPerLayer_{i} t} \left [1- \mathcal{B}_{i}(t) \right ]  = 0 $. Using the hypothesis $\MalthusPerLayer_{i} > 0 $, we obtain that $	\lim\limits_{t \rightarrow \infty} e^{- 2\MalthusPerLayer_{i} t} \left [1- \mathcal{B}_{i}(t) \right ]  = 0 $. Then, 
				\begin{equation*}
					e^{- 2\MalthusPerLayer_{i} t}(M^a_{i,i})^2 \ast d\mathcal{B}_{i}(t) = \int_{0}^{\infty} \mathds{1}_{[0,t]}(M^a_{i,i}(t-u)e^{- \MalthusPerLayer_{i} (t-u)})^2 d\mathcal{B}_{i}(u)e^{-2 \MalthusPerLayer_{i} u}du.
				\end{equation*}
				Using theorem \ref{MoyenneTempsLong}, we have $M^a_{i,i}(t) \sim e^{\MalthusPerLayer_{i} t} \widetilde{ M}_{i,i}(a)$, as $t \rightarrow \infty$. Applying Lebesgue dominated convergence theorem, we obtain 
				\begin{equation*}
					\lim\limits_{t \rightarrow \infty} e^{- 2\MalthusPerLayer_{i} t}(M^a_{i,i})^2 \ast d\mathcal{B}_{i}(t)  =( \widetilde{ M}_{i,i}(a))^2 d\mathcal{B}^*_{i}(2\MalthusPerLayer_{i}).
				\end{equation*}
				Then, applying lemma \ref{HarrisLemma_4}, we deduce:
				\begin{equation*}
					\begin{array}{l}
						L^a_{i,i}(t) \sim \widetilde{L}_{i,i}(a) e^{ 2\MalthusPerLayer_{i} t}, \text{  as  } t\rightarrow \infty,  \text{  where  } 
						\widetilde{L}_{i,i}(a) =   \frac{2p^{(i)}_{2,0} d\mathcal{B}^*_{i}(2\MalthusPerLayer_{i})( \widetilde{ M}_{i,i}(a))^2}{1- 2p^{(i)}_{S}d\mathcal{B}^*_{i}(2\MalthusPerLayer_{i} )}.
					\end{array}
				\end{equation*}
				Hence, $\mathcal{H}^{0}$ is true.
				Then, we suppose that $\mathcal{H}^{k - 1}$ holds and we show $\mathcal{H}^{k}$. According to \eqref{MomentDOrdre2Equ}, we write the equation for $L^a_{i,i+k} $ and rescale it by $e^{-2\MalthusPerLayer_{i,i +k}t}$:
				\begin{multline*}
					\widehat{L}^a_{i,i+k}(t) =  
					2p^{(i)}_{S} d\mathcal{B}^*_{i}(2\MalthusPerLayer_{i,i +k} ) \widehat{L}^a_{i,i+k} \ast \widehat{d\mathcal{B}}_{i}(t) +  e^{-2\MalthusPerLayer_{i,i +k}t}2p^{(i)}_{L}L^a_{i+1,i+k}\ast d\mathcal{B}_{i}(t)   \\
					+  e^{-2\MalthusPerLayer_{i,i +k}t}\left [2p^{(i)}_{2,0} (M^a_{i,i+k})^2 + 2 p^{(i)}_{1,1}M^a_{i,i+k}M^a_{i+1,i+k}+ 2p^{(i)}_{0,2} (M^a_{i+1,i+k})^2  \right  ]\ast d\mathcal{B}_{i}(t).
				\end{multline*}
				Here, $m = 2p^{(i)}_{S} d\mathcal{B}^*_{i}(2\MalthusPerLayer_{i,i +k} )< 1$, so that we can use lemma \ref{HarrisLemma_4}. We first compute the limit of $e^{-2\MalthusPerLayer_{i,i +k}t}L^a_{i+1,i+k}\ast d\mathcal{B}_{i}(t) $ when t goes to infinity when either $\MalthusPerLayer_{i,i +k} =  \MalthusPerLayer_{i} $ or $\MalthusPerLayer_{i,i +k} \neq  \MalthusPerLayer_{i}$. We start with the case $\MalthusPerLayer_{i,i +k} \neq  \MalthusPerLayer_{i}$ (so, $\MalthusPerLayer_{i,i +k} = \MalthusPerLayer_{i +1 ,i +k}$). For all $t \geq 0$, we have:
				\begin{multline*}
				e^{-2\MalthusPerLayer_{i,i +k}t}L^a_{i+1,i+k}\ast d\mathcal{B}_{i}(t)  = \\ d\mathcal{B}^*_{i}(2\lambda_{i,i+k})\int_{0}^{\infty} \mathds{1}_{[0,t]} e^{-2\lambda_{i,i+k}(t-u)}L^a_{i+1,i+k}(t-u)  \widehat{d\mathcal{B}_i}(u)du.
				\end{multline*}
				According to $\mathcal{H}^{k - 1} $, we know that  $L^a_{i + 1,i+k }(t)  \sim \widetilde{L}_{i + 1,i+k }(a)e^{2\MalthusPerLayer_{i,i +k }t}$. We deduce with a Lebesgue dominated  convergence theorem that:
				\begin{equation*}
					\lim\limits_{t \rightarrow \infty}e^{-2\MalthusPerLayer_{i,i +k}t}L^a_{i+1,i+k}\ast d\mathcal{B}_{i}(t) = d\mathcal{B}^*_{i}(2\lambda_{i,i+k}) \widetilde{L}_{i + 1,i+k }(a).
				\end{equation*}
				We apply the same method as above for the other terms of the source term. Theorem \ref{MoyenneTempsLong} gives us that $M^a_{i,i+k} \sim e^{\MalthusPerLayer_{i,i +k}t} \widetilde{M}_{i,i+k}(a)$ and $M^a_{i+1,i+k} \sim e^{\MalthusPerLayer_{i,i +k}t} \widetilde{M}_{i+1,i+k}(a)$.  Using Lebesgue dominated convergence theorem, we obtain:
				\begin{multline*}
					\lim\limits_{t \rightarrow \infty} e^{-2\MalthusPerLayer_{i,i +k}t}\left [2p^{(i)}_{2,0} (M^a_{i,i+k})^2 + 2 p^{(i)}_{1,1}M^a_{i,i+k}M^a_{i+1,i+k}+ 2p^{(i)}_{0,2} (M^a_{i+1,i+k})^2  \right  ]\ast d\mathcal{B}_{i}(t)  \\
					= \left [ 2p^{(i)}_{2,0}(\widetilde{M}_{i,i+k}(a))^2 + 2 p^{(i)}_{1,1}\widetilde{M}_{i,i+k}(a)\widetilde{M}_{i+1,i+k}(a) \right. \\
					\left.	+ 2p^{(i)}_{0,2} (\widetilde{M}_{i+1,i+k}(a))^2  \right ]d\mathcal{B}^*_{i}(2\MalthusPerLayer_{i,i +k}).
				\end{multline*}
				
				We then consider the case $\MalthusPerLayer_{i,i +k} =  \MalthusPerLayer_{i} > \MalthusPerLayer_{i + 1,i +k} $ and start by computing the limit of $e^{-2\MalthusPerLayer_{i,i +k}t}L^a_{i+1,i+k}\ast d\mathcal{B}_{i}(t) $. 
				\begin{multline*}
					e^{-2\MalthusPerLayer_{i,i +k}t}L^a_{i+1,i+k}\ast d\mathcal{B}_{i}(t) = \\ d\mathcal{B}^*_{i}(2\lambda_{i,i+k}) e^{-2(\lambda_{i,i+k} - \lambda_{i + 1,i+k} )t} \int_{0}^{\infty} \mathds{1}_{[0,t]} e^{-2\lambda_{i + 1,i+k}(t-u)}L^a_{i+1,i+k}(t-u)  \widehat{d\mathcal{B}_i}(u)du.
				\end{multline*}
				Using $\mathcal{H}^{k - 1} $ and Lebesgue dominated convergence theorem, we first obtain that 
				\begin{multline*}
				\lim\limits_{t \rightarrow \infty}  \int_{0}^{\infty} \mathds{1}_{[0,t]} e^{-2\lambda_{i + 1,i+k}(t-u)}L^a_{i+1,i+k}(t-u)  \widehat{d\mathcal{B}_i}(u)du \\
				 =  d\mathcal{B}^*_{i}(2\lambda_{i+1,i+k}) \widetilde{L}_{i+1,i+k}(a) < \infty,
				\end{multline*}
			    hence,
				\begin{equation*}
					\lim\limits_{t \rightarrow \infty} e^{-2\MalthusPerLayer_{i,i +k}t}L^a_{i+1,i+k}\ast d\mathcal{B}_{i}(t) = 0.
				\end{equation*}
				Then, theorem \ref{MoyenneTempsLong} give us that $M^a_{i,i+k} \sim e^{\MalthusPerLayer_{i,i +k}t} \widetilde{M}_{i,i+k}(a)$ and \\
				$M^a_{i+1,i+k} \sim e^{\MalthusPerLayer_{i+1,i +k}t} \widetilde{M}_{i+1,i+k}(a)$. Using similar methods, we obtain:
				\begin{multline*}
					\lim\limits_{t \rightarrow \infty} e^{-2\MalthusPerLayer_{i,i +k}t}\left [2p^{(i)}_{2,0} (M_{i,i+k})^2 + 2 p^{(i)}_{1,1}M_{i,i+k}M_{i+1,i+k}+ 2p^{(i)}_{0,2} (M_{i+1,i+k})^2  \right  ]\ast d\mathcal{B}_{i}(t) \\
					=  2p^{(i)}_{2,0}(\widetilde{M}_{i,i+k})^2 d\mathcal{B}^*_{i}(2\MalthusPerLayer_{i,i +k}).
				\end{multline*}
				We conclude by applying lemma \ref{HarrisLemma_4} that $\mathcal{H}^k $ holds.
				\end{proof}
			
				\paragraph{Variance}
				\begin{definition}
					We write $ v_j^a(t)$, the variance of $Y_t^{(j,a)}$ starting from a mother cell on the first layer such that:
					\begin{equation}
						v_j^a(t) = \mathbb{E}[(Y_t^{(j,a)})^2 \lvert Z_0 = \delta_{1,0}] - \mathbb{E}[Y_t^{(j,a)} \lvert Z_0 = \delta_{1,0}]^2. 
					\end{equation}
				\end{definition}
				We study the asymptotic behavior of the variance $ v_j^a(t)$ when the first layer is the leading one.
				\begin{corollary}\label{Corollary_Vari}
					Let $a \geq 0$. Under the same hypotheses as in theorem \ref{MomentDOrdre2Long} and supposing that $\MalthusIndex = 1$, we have
					\begin{equation*}
						\forall k \in \llbracket 1, J \rrbracket \quad 
						v^a_{j}(t) \sim \widetilde{v}_{j}(a) e^{2\Malthus t}, \text{  as  } t\rightarrow \infty
					\end{equation*}
					where 	
					\begin{equation*}
					\widetilde{v}_{j}(a) =  \widetilde{L}_{1,j}(a) - (m_j(a))^2 = \left[\frac{2p^{(1)}_{2,0} d\mathcal{B}^*_{1}(2\MalthusPerLayer_{1})}{1- 2p^{(1)}_{S}d\mathcal{B}^*_{1}(2\MalthusPerLayer_{1} )} - 1  \right]  (m_j(a))^2.
					\end{equation*}
				\end{corollary}
				\begin{proof}
					Let $a \geq 0$. According to theorem \ref{MomentDOrdre2Long} and using that $\MalthusIndex = 1$, we have:
					\begin{equation*}
						\forall j \in \llbracket 1, J \rrbracket, \quad L^a_{1,j}(t) \sim \widetilde{L}_{1,j}(a)e^{2 \MalthusPerLayer_1 t}, \quad \text{ as } t \rightarrow \infty \, .
					\end{equation*}
					Using theorem \ref{MoyenneTempsLong} and \ref{MomentDOrdre2Long}, we deduce for all $j \in \llbracket 1, J \rrbracket$:
					
					\begin{equation*}
						\widetilde{v}_{j}(a) =  \widetilde{L}_{1,j}(a) - (m_j(a))^2 = \left[\frac{2p^{(1)}_{2,0} d\mathcal{B}^*_{1}(2\MalthusPerLayer_{1})}{1- 2p^{(1)}_{S}d\mathcal{B}^*_{1}(2\MalthusPerLayer_{1} )} - 1  \right]  (m_j(a))^2.
					\end{equation*}
					\vspace{-0.15cm}
				\end{proof}

\section{Numerical simulation procedures}

		\subsubsection{Stochastic simulation procedures}
			\paragraph{Markov case}
			Considering a markovian case, we simulate the process $Z_t$ solution of the SDE \eqref{Z_t_Equation} with the Gillespie algorithm. We use the package StochSS \cite{drawert_stochastic_2016}. \\
			We consider that for each layer $j \in \llbracket 1, 3 \rrbracket$, $p_{1,1}^{(j)} = 0$. Hence, $p_{2,0}^{(j)}  = p_S^{(j)}$ and $p_{0,2}^{(j)}  = 1 - p_S^{(j)}$. Considering a system with 4 layers, our system is ruled by the $7$-th reactions below:
			\begin{equation*}
				\begin{array}{ccc}
					\mathcal{R}_1: & N_1 \rightarrow  N_1  + N_1  & \text{ with rate } b_1 p_S^{(1)},\\
					\mathcal{R}_2: & N_1 \rightarrow  N_2  + N_2  & \text{ with rate } b_1 (1 - p_S^{(1)} ),  \\
					\mathcal{R}_3: & N_2 \rightarrow  N_2  + N_2  & \text{ with rate } b_2 p_S^{(2)}, \\
					\mathcal{R}_4: & N_2 \rightarrow  N_3 + N_3  & \text{ with rate } b_2 (1 - p_S^{(2)} ),  \\
					\mathcal{R}_5: & N_3 \rightarrow  N_3  + N_3  & \text{ with rate } b_3 p_S^{(3)}, \\
					\mathcal{R}_6: & N_3 \rightarrow  N_4 + N_4  & \text{ with rate } b_3 (1 - p_S^{(3)} ),  \\
					\mathcal{R}_7: & N_4 \rightarrow  N_4 + N_4  & \text{ with rate } b_4. \\
				\end{array}
			\end{equation*}

			\paragraph{General case}
			We simulate our process using the algorithm \ref{alg:SimulationStochasticProcess}, on a predefine time horizon $T_{\max}$.
			\begin{figure}[!htb]		
				\caption{Simulation stochastic process}
				\label{alg:SimulationStochasticProcess}
				\begin{itemize}
					\item[] Define a sequence $\mathcal{S}$ of cells of a given age and layer.
					\item Simulate the time of division of each cell in $\mathcal{S}$
					\item[While] $t < T_{\max}$
						\begin{enumerate}
							\item Select the next cell $m$ that will divide. $l^m$ is its layer index and $l^m$ is the age at division.
							\item Randomly draw the layer of its daughters cell $l^{d_1 }$ and $l^{d_2}$ according to the probabilities $p^{(l^m)}_{2,0}$,$p^{(l^m)}_{1,1}$ and $p^{(l^m)}_{2,0}$.
							\item Randomly draw the next time of division of daughter cell $d_1$ according to its layer index $l^{d_1 } $.
							\item Randomly draw the next time of division of daughter cell $d_2$ according to its layer index $l^{d_2} $. 
							\item Add $d_1$ and $d_2$ into the sequence $\mathcal{S}$
							\item $t \leftarrow t + t^m$
						\end{enumerate}
				\end{itemize}
			\end{figure}
		
		\subsubsection{Deterministic simulation protocol}\label{SimulationProtocole_PDE}
			To solve numerically the pro\-blem \eqref{EDP_equation}, we design a dedicated finite volume scheme adapted to the non-conservative form with proper boundary conditions. We define the time step $\Delta t$ and the age step $\Delta a$. The time discretization is defined by 
			\begin{equation*}
				t_0 = 0, \quad t_{n + 1} = t_n + \Delta t, \quad \text{ for } n = 0, ..., N_t
			\end{equation*}
			where $(N_t + 1)\Delta t$ is the time horizon of the simulation. Similarly, $N_a$ is the number of cells\footnote{The cell is here the standard name used for each elementary volume in the framework of finite volume methods.} in the domain. The cells $ \mathcal{C}_i$ are indexed by a rational number i ($\frac{1}{2} $,$\frac{3}{2} $, etc.) with $i \in \llbracket \frac{1}{2}, N_a - \frac{1}{2} \rrbracket $. The edges of each cell are located at $a_{i - \frac{1}{2}} = ( i - \frac{1}{2}) \Delta a $ and $a_{i + \frac{1}{2}} = ( i + \frac{1}{2}) \Delta a$ (remark that $\Delta a =  a_{i + \frac{1}{2}}  - a_{i - \frac{1}{2}} $ and $a_0 = 0$). As age and time evolve at the same speed, we chose $N_a$ such that $t_{N_t} - a^0_{\max}< N_a \Delta a $ where $a^0_{\max}$ is the maximal age of the initial distribution. \\

			Let $j \in \llbracket 1, J \rrbracket$. We define $P^j_{n,i}$ as the mean value of the density $\rho^{(j)}$ in cell $\mathcal{C}_i $ at time $t_n$:
			\begin{equation*}
				P^j_{n,i} := \displaystyle \frac{1}{\Delta a}  \int_{a_{i - \frac{1}{2}}}^{a_{i + \frac{1}{2}}} \rho^{(j)}(t_n,a)\,da \, .
			\end{equation*}
			We integrate the equation $\partial_t \rho^{(j)} + \partial_a \rho^{(j)} = -b_j \rho^{(j)} $ with respect to age in cell $ \mathcal{C}_i$ and obtain:
			\begin{equation*}
				\frac{d}{dt} \int_{a_{i - \frac{1}{2}}}^{a_{i + \frac{1}{2}}} \rho^{(j)}(t,a) \, da = -  \rho^{(j)}(t,a_{i + \frac{1}{2}}) + \rho^{(j)}(t,a_{i - \frac{1}{2}}) - \int_{a_{i - \frac{1}{2}}}^{a_{i + \frac{1}{2}}} b_j(a) \rho^{(j)}(t,a) \, da.
			\end{equation*}
			Then, we suppose that all $b_j$s functions are regular enough so that we can approximate $b_j$, for all $j \in \llbracket 1, J \rrbracket$ on each cell $ \mathcal{C}_i$ by their mean value $\bar{b}^i_j$. We obtain:
			\begin{equation*}
				\frac{d}{dt} \int_{a_{i - \frac{1}{2}}}^{a_{i + \frac{1}{2}}} \rho^{(j)}(t,a) \, da = -  \rho^{(j)}(t,a_{i + \frac{1}{2}}) + \rho^{(j)}(t,a_{i - \frac{1}{2}}) - \bar{b}^i_j \int_{a_{i - \frac{1}{2}}}^{a_{i + \frac{1}{2}}}\rho^{(j)}(t,a) \, da .
			\end{equation*}
			We approximate the derivative in time with a finite difference scheme:
			\begin{equation*}
				\frac{d}{dt} \int_{a_{i - \frac{1}{2}}}^{a_{i + \frac{1}{2}}} \rho^{(j)}(t_n,a) \, da = \frac{1}{\Delta t} \left [ \int_{a_{i - \frac{1}{2}}}^{a_{i + \frac{1}{2}}} \rho^{(j)}(t_{n+1},a) da  -  \int_{a_{i - \frac{1}{2}}}^{a_{i + \frac{1}{2}}} \rho^{(j)}(t_{n},a) da  \right ] + \mathcal{O}(\Delta t)
			\end{equation*}
			and we deduce:
			\begin{equation*}
				\frac{\Delta a}{\Delta t} \left [ P^{j}_{n+1,i} - P^{j}_{n,i}  \right ] = -  \rho^{(j)}(t,a_{i + \frac{1}{2}}) + \rho^{(j)}(t,a_{i - \frac{1}{2}}) - \bar{b}^i_j  \Delta a P^{j}_{n,i}  .
			\end{equation*}
			The edge terms $\rho^{(j)}(t,a_{i + \frac{1}{2}}) $ and $\rho^{(j)}(t,a_{i - \frac{1}{2}})$ correspond to the fluxes that cross the boundaries of cell $\mathcal{C}_i$.  When $i = \frac{1}{2} $, the boundary condition of equation \eqref{EDP_equation} gives us the value of this term:
			\begin{multline*}
						 \rho^{(j)}(t_n,a_{0})  =  2p_S^{(j)} \int_{0}^{\infty}b_j(a)\rho^{(j)}(t_{n},a) da +  2(1 - p_S^{(j - 1)}) \int_{0}^{\infty}b_{j-1}(a)\rho^{(j-1)}(t_{n},a) da   \\[0.3cm]
						 = 2p_S^{(j)} \sum_{i}\int_{\mathcal{C}_i}b_j(a)\rho^{(j)}(t_{n},a) da +  2(1 - p_S^{(j - 1)}) \sum_{i}\int_{\mathcal{C}_i}b_{j-1}(a)\rho^{(j-1)}(t_{n},a) da   \\[0.3cm]
						 =2p_S^{(j)} \Delta t \sum_{i} \bar{c}_i P^j_{n,i} +  2(1 - p_S^{(j - 1)})\Delta t \sum_{i} \bar{c}^{j-1}_i P^{j-1}_{n,i} \, .  \\
			\end{multline*}
			When $i \neq \frac{1}{2}$, we approximate each term $\rho^{(j)}(t_n,a_{i + \frac{1}{2}}) $ by 
			\begin{equation*}
				\rho^{(j)}(t_n,a_{i + \frac{1}{2}})  = P^j_{n,i + \frac{1}{2}} + \mathcal{O}(\Delta a).
			\end{equation*}
			Hence, we obtain the following numerical scheme:
			\begin{equation*}
				P^{j}_{n+1,i}  = 	 \left [ 1 -\bar{b}^j_{i}  \Delta t - 	\frac{\Delta t}{\Delta x} \right ]  P^{j}_{n,i} + \frac{\Delta t}{\Delta x}P^j_{n,i-1} 
			\end{equation*}
			\begin{equation*}
				P^{j}_{n+1,\frac{1}{2}}  = 	 \left [ 1 -\bar{c}^j_{\frac{1}{2}}  \Delta t - 	\frac{\Delta t}{\Delta x} \right ]  P^{j}_{n,\frac{1}{2}} + 2s_j\Delta t \sum_{i} \bar{c}_i P^j_{n,i} +  2(1 - s_{j-1})\Delta t \sum_{i} \bar{c}^{j-1}_i P^{j-1}_{n,i}  .
			\end{equation*}

\subsection{Construction of figure \ref{fig:AgeProfil}\label{Supp_ProtocolForFigure4}}
	In this part, we give some details about the construction of figure \ref{fig:AgeProfil}. 
	We simulate the SDE \eqref{Z_t_Equation} using the algorithm \ref{alg:SimulationStochasticProcess} and the PDE \eqref{EDP_equation} using the algorithm described in the subsection below (see \ref{SimulationProtocole_PDE}) taking $\Delta a = 9.5 \times 10e-3$ and  $\Delta t = 10e-4$.
	
	We discretized the age according to a sequence of integers $k \in \llbracket 1, 50 \rrbracket$. Let $j \in \llbracket 1, J \rrbracket$ be a layer index. 
	The color bar associated with age $k$ for the $j$-th layer corresponds to the total number of cells on the $j$-th layer of age $a \in  [k , k + 1 [ $ renormalized by the total number of cells:
		\begin{equation*}
			\frac{<< Z_t, \mathds{1}_{j,  k \leq a < k + 1 }>>}{<< Z_t, \mathds{1}>>} \, .
		\end{equation*}
	 The dashed black line with the age $k$ for the $j$-th layer corresponds to:
		\begin{equation*}
			\frac{\int_{k}^{k+1}\rho^{(j)}(t,a)da}{\sum_{j = 1}^{4}\int_{0}^{+ \infty}\rho^{(j)}(t,a)da}  \sim \frac{ \displaystyle \sum_{i = \lfloor \frac{k}{\Delta x}\rfloor}^{\lfloor \frac{k + 1}{\Delta x}\rfloor - 1} P^j_{n,i}}{ \displaystyle \sum_{j =1}^{4}\sum_{i} P^j_{n,i}} \, .
		\end{equation*}	
		
		The color solid lines which represent the stable distribution $\eigenfunction$ and compute their value at each age point $k$ by 
		\begin{equation*}
			\frac{\int_{k}^{k+1} \eigenfunction^{(j)}(a)da}{\sum_{j =1}^{4}\int_{0}^{+ \infty} \eigenfunction^{(j)}(a)da } \, .
		\end{equation*}
		
\subsection{Parameter estimation procedure}\label{SM_ParmeterEstimaProcedure}
	Using the software D2D \cite{raue_data2dynamics:_2015}, we estimate the parameters of our model, using an additive Gaussian noise statistical model (standard least squares likelihood). The standard deviation and the initial number $N$ of cells on the first layer  are also estimated. To investigate the practical identifiability, we compute the profile likelihood estimate (PLE) \cite{raue_structural_2009}. We observe that all the parameters are practically identifiable except the probability of staying on the second layer $p_S^{(2)} $ (see Figure \ref{fig:pledatamn}). In contrast, most of the parameters are not practically identifiable when we consider the total number of cells as the observable function ($\sigma(t;p) = \sum_{j=1}^{J}M^{(j)}(t;p)$, Figure \ref{fig:pledatamntotalcell}).
	
\begin{figure}[!htb]
	\centering
	\subfloat[Practical Identifiability for $\sigma(t;p) = (M^{(j)}(t;p))_{j \in \llbracket 1, J \rrbracket }$]{\includegraphics[width=1.0\linewidth]{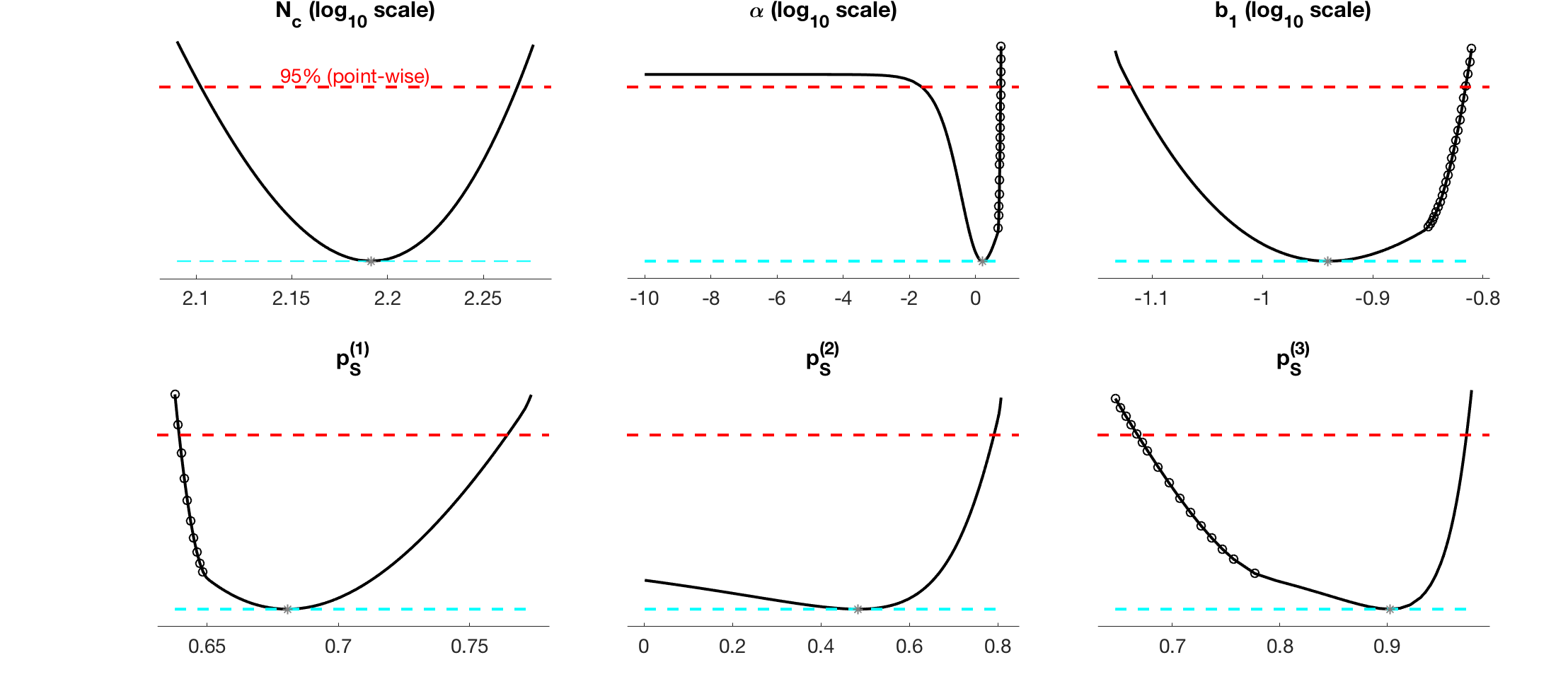}\label{fig:pledatamn}} \\
	\subfloat[Practical Identifiability for $\sigma(t;p) = \sum_{j=1}^{J}M^{(j)}(t;p)$]{\includegraphics[width=0.35\linewidth]{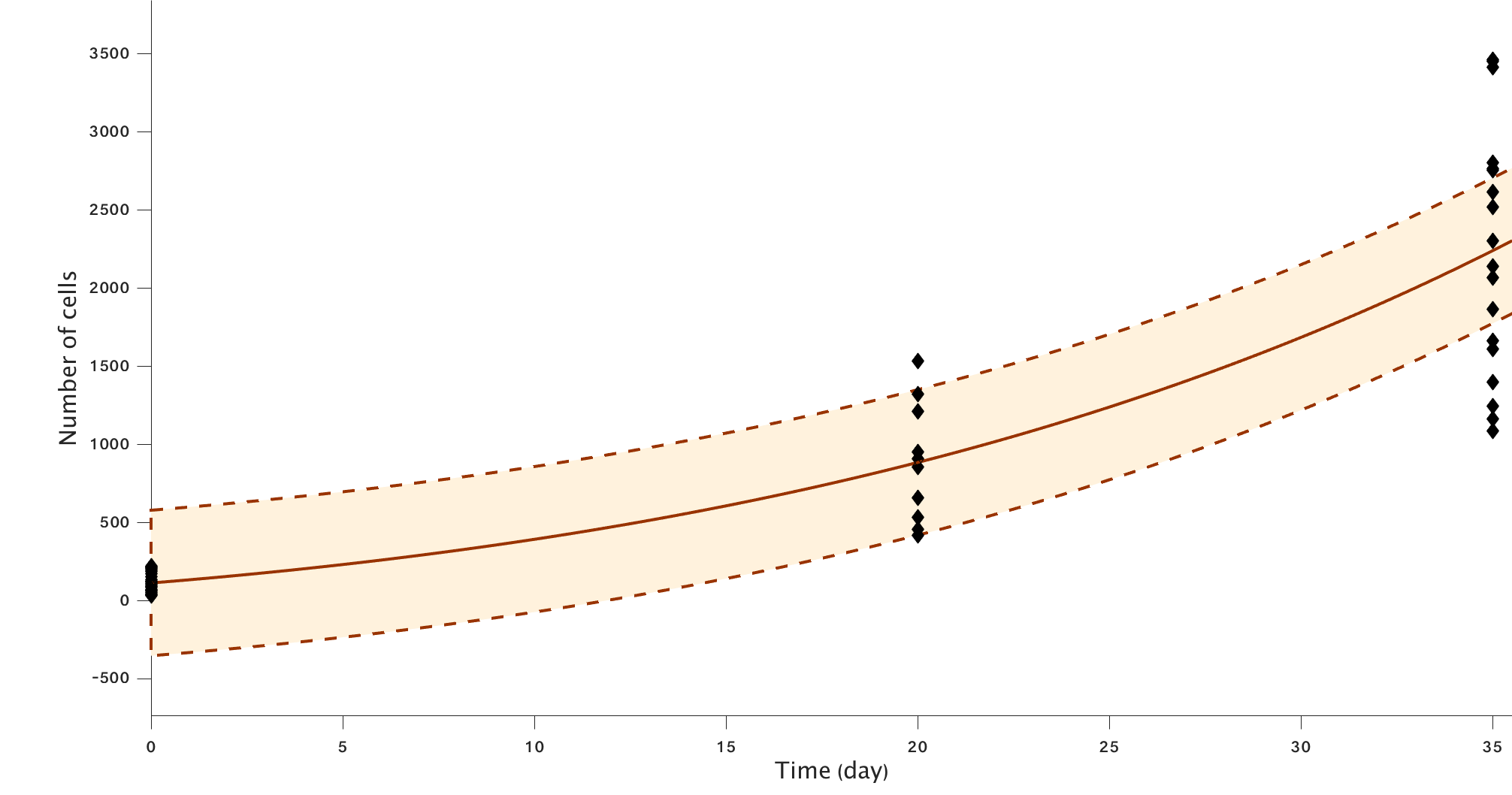}\label{fig:fitdatamntotalcell} \includegraphics[width=0.55\linewidth]{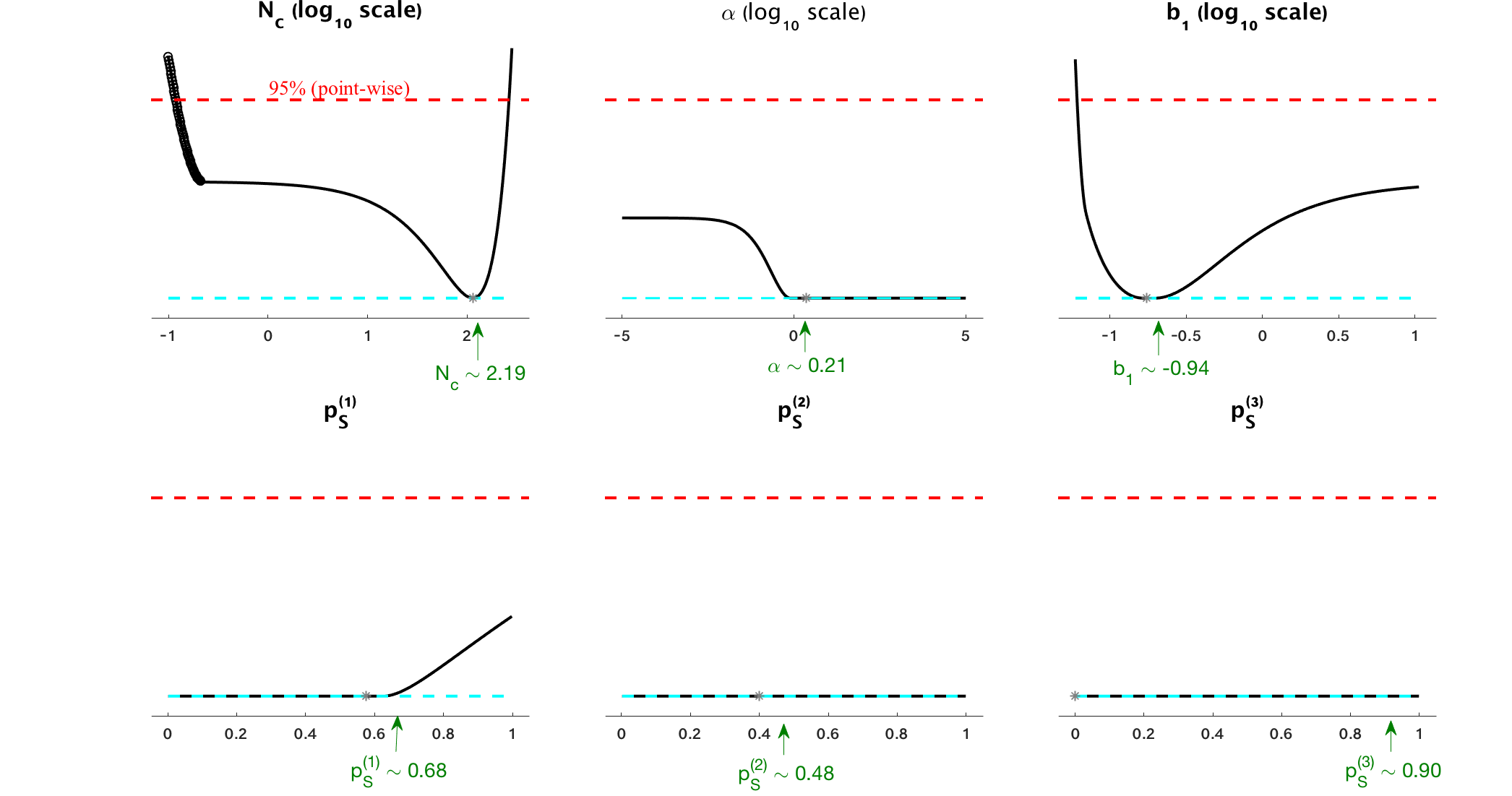}\label{fig:pledatamntotalcell}}
	\caption{\textbf{Practical Identifiability}. \textbf{Figure \ref{fig:pledatamn}} Profile likelihood estimate (PLE) for each parameter in the set $\mathbf{P}_{exp} = \{N,b_1,\alpha, p_S^{(1)}, p_S^{(2)}, p_S^{(3)}\}$ when the observation function is $\sigma(t;p) =  (M^{(j)}(t;p))_{j \in \llbracket 1, J \rrbracket }$.The red dashed lines correspond to the $95\%$-statistical threshold while the blue dashed lines correspond to the optimum value of the likelihood.  \textbf{Figure \ref{fig:pledatamntotalcell}} Parameter estimation results for $\sigma(t;p) = \sum_{j=1}^{J}M^{(j)}(t;p)$. \textbf{Left panel}: Data fitting model with model (\ref{ODE_equ}). The black diamonds represent the experimental data (total number of cells), the solid line is the best fit solution of (\ref{ODE_equ}) and the dashed lines are drawn from the estimated variance. \textbf{Left panel}: Profile likelihood estimates of each parameter in the set $P_{exp}$.}
\end{figure}

\bibliographystyle{plain}
\bibliography{article_2017_ark}
\end{document}